\documentclass[a4paper,oneside,10pt]{amsart}
\usepackage{amsmath,amssymb,mathtools,bm,lipsum}
\usepackage[english]{babel}
\usepackage[T1]{fontenc}
\usepackage[utf8]{inputenc}
\usepackage{dcolumn}
\usepackage{blkarray}
\usepackage{textcomp}               
\usepackage{amstext}
\usepackage{amsthm}
\usepackage{stmaryrd}
\usepackage{mathrsfs}
\usepackage{bbm} 
\usepackage{units}
\usepackage{enumerate}

\numberwithin{equation}{section}
\newtheorem{thm}{Theorem}[section]
\newtheorem{cor}[thm]{Corollary}
\newtheorem{lem}[thm]{Lemma}
\newtheorem{prop}[thm]{Proposition}
\newtheorem{defn}[thm]{Definition}

\newtheorem{ex}[thm]{Example}
\newtheorem{hyp}{Hypothesis}

\newcommand{\CC}{\mathbb{C}}

\newcommand{\NN}{\mathbb{N}}
\newcommand{\PP}{\mathbb{P}}
\newcommand{\QQ}{\mathbb{Q}}
\newcommand{\RR}{\mathbb{R}}

\newcommand{\ZZ}{\mathbb{Z}}
\newcommand{\cA}{\mathcal{A}}\newcommand{\cN}{\mathcal{N}}
\newcommand{\cB}{\mathcal{B}} 
\newcommand{\cC}{\mathcal{C}}
\newcommand{\cD}{\mathcal{D}}\newcommand{\cP}{\mathcal{P}} 
\newcommand{\cE}{\mathcal{E}}
\newcommand{\cF}{\mathcal{F}}
\newcommand{\cG}{\mathcal{G}}\newcommand{\cS}{\mathcal{S}}
\newcommand{\cH}{\mathcal{H}}
\newcommand{\cI}{\mathcal{I}}\newcommand{\cU}{\mathcal{U}}
\newcommand{\cJ}{\mathcal{J}}
\newcommand{\cK}{\mathcal{K}}
\newcommand{\cL}{\mathcal{L}}\newcommand{\cX}{\mathcal{X}} 
\newcommand{\cM}{\mathcal{M}}

\title{Spin-Boson type models analysed using symmetries}
\author{Thomas Norman Dam, Jacob Schach Møller }
\begin{document}

\begin{abstract}
In this paper, we analyse a family of models for a qubit interacting with a bosonic field. This family of models is very large and contains models where higher order perturbations of field operators are added to the Hamiltonian. The Hamiltonian has a special symmetry, called spin-parity symmetry, which plays a central role in our analysis. Using this symmetry, we find the domain of selfadjointness and decompose the Hamiltonian into two fiber operators each defined on Fock space. We then prove an HVZ theorem for the fiber operators and single out a particular fiber operator which has a ground state if and only if the full Hamiltonian has a ground state. From these results we deduce a simple criterion for the existence of an exited state.
\\

\\

\end{abstract}

\maketitle

\begin{center}
\today
\end{center}
\let\thefootnote\relax\footnote{Thomas Norman Dam: Department of mathematics, Aarhus University,  8000 Aarhus C  Denmark; tnd@math.au.dk}
\let\thefootnote\relax\footnote{Jacob Schach Møller: Department of mathematics, Aarhus University, 8000 Aarhus C  Denmark; Jacob@math.au.dk.}

\section{Introduction}	
This paper is devoted to the analysis of so called spin-boson type models which is a family of models describing a qubit interacting with a bosonic field. The assumptions in our framework are very weak, which allows us to cover both the Rabi model and the standard spin-boson model simultaneously. Furthermore, higher order perturbations of field operators are also considered. QFT Models with higher order perturbations have lately become relevant in physics. They appear in cavity QED (see \cite{Jacob}) and in the theory of bose polarons (see \cite{Shchadilova}).

Models with higher order perturbations are treated in \cite{Hidaka}, \cite{Jacob}, \cite{Miyao} and  \cite{Takaesu}. Spin-boson type models are treated in \cite{Jacob}, \cite{Miyao} and \cite{Takaesu}, but the authors assume either that the field is massive or that the coupling is weak. The results in \cite{Hidaka} does not assume weak coupling or a massive field, but the model treated in that paper is not the spin-boson model and rather strong infrared conditions are assumed. Furthermore, the author of \cite{Hidaka} only proves selfadjointness of the Hamiltonian and existence of ground states, while we treat several other questions as well.

The analysis in this paper relies on the fact that spin-boson type Hamiltonians commute with the spin-parity operator. The spin-parity operator has two invariant subspaces, which are both isomorphic to the Fock space. This fact was used in \cite{Volker} and \cite{Hasler} to prove that ground states exist in the massless spin-boson model. We use this fact to decompose the Hamiltonian into two so called fiber operators. We shall see, that the two fiber operators differ only by the value of a scalar parameter, but they behave quite differently.

We start by proving selfadjointness of all involved operators and move on to prove an HVZ theorem for the fiber operators. The method we use is related to the approach in \cite{Jacob1}, but is written up in a more general way, which allows one to handle massless fields and abstract Hilbert spaces. The HVZ theorem for the fiber operators also gives an HVZ theorem for the full Hamiltonian.

Using arguments similar to those presented in \cite{Hasler}, we prove that if ground states exists for the full Hamiltonian, then the bottom of the spectrum is a non degenerate eigenvalue. Using this result, we single out a particular fiber which has a ground state if and only if the full Hamiltonian has a ground state. Ground states for the other fiber operator must therefore correspond to exited states. The HVZ theorem then gives a simple criterion for the existence of an exited state.

The reader is then encouraged to have a look at Appendix D, where a new framework for pointwise annihilation operators is developed. Most maps are continuous in this framework, so calculations are reduced to simple algebraic manipulations. This makes it very easy to rigorously prove higher order pull-through formulas. Using these pull-through formulas, we prove that ground states are in the domain of the number operator raised to any positive power (if infrared regularity is assumed).

Lastly, we follow the general strategy outlined in \cite{Gerard} to prove the existence of ground states in massless (but infrared regular) models. Our proofs are simpler than the ones presented in \cite{Gerard} and we are able to work under weaker assumptions on the bosonic dispersion relation. This is possible due to a novel approach to the last step in \cite{Gerard}.

\section{Notation and definitions}
We start by introducing the notation. If $X$ is a topological space then we will write $\cB(X)$ for the Borel $\sigma$-algebra. Furthermore, if $(\cM,\cF,\mu)$ is a measure space, $X$ is a Banach space and $1\leq p\leq \infty$ then we will write $L^p(\cM,\cF,\mu,X)$ for the vector valued $L^p$ space. If $X=\CC$ we will drop $X$ from the notation.

Throughout this paper, $\cH$ will always denote a separable Hilbert space. Write $\cH^{\otimes n}$ for the $n$-fold tensor product of $\cH$ and let $\cH^{\otimes_s n}\subset \cH^{\otimes n}$ be the subspace of symmetric tensors. The bosonic (or symmetric) Fock space is defined as
\begin{equation*}
\cF_b(\cH)=\bigoplus_{n=0}^\infty  \cH^{\otimes_s n}.
\end{equation*}
If $\cH=L^2(\cM,\cF,\mu)$ where $(\cM,\cF,\mu)$ is $\sigma$-finite then we may give a concrete description of $\cH^{\otimes_s n}$ as $L_{sym}^2(\cM^{n},\cF^{\otimes n},\mu^{\otimes n})$. We will write an element $\psi\in \cF_b(\cH)$ in terms of its coordinates as $\psi=(\psi^{(n)})$ and define the vacuum $\Omega=(1,0,0,\dots)$. The set of finite particle vectors is defined by
\begin{equation*}
\cN=\{  (\psi^{(n)})\in \cF_b(\cH)\mid \text{ $\exists K\in \NN$ s.t. $\psi^{(n)}=0$ for all $n\geq K$} \}.
\end{equation*}
For $g\in \cH$ one defines the annihilation operator $a(g)$ and the creation operator $a^{\dagger}(g)$ on symmetric tensors in $\cF_b(\cH)$ by $a(g)\Omega=0$, $a^\dagger(g)\Omega=g$ and
\begin{align*}
a(g)( f_1\otimes_s\cdots\otimes_s f_n )&=\frac{1}{\sqrt{n}}\sum_{i=1}^{n} \langle g,f_i \rangle f_1\otimes_s\cdots\otimes_s \widehat{f}_i\otimes_s\cdots\otimes_s f_n\\
a^\dagger(g)( f_1\otimes_s\cdots\otimes_s f_n )&=\sqrt{n+1}g\otimes_s f_1\otimes_s\cdots\otimes_s f_n
\end{align*}
where $\widehat{f}_i$ means that $f_i$ is omitted from the tensor product. One can show that these operators extends to closed operators in $\cF_b(\cH)$ and that $(a(g))^*=a^{\dagger}(g)$. Furthermore, we have the canonical commutation relations which are:
\begin{equation*}
\overline{[a(f),a(g)]}=0=\overline{[a^\dagger(f),a^\dagger(g)]} \,\,\text{and}\,\,\, \overline{[a(f),a^\dagger(g)]}=\langle f,g\rangle.
\end{equation*}
We also define the field operators
\begin{equation*}
\varphi(g)=\overline{ a(g)+a^\dagger(g) }.
\end{equation*}
They are selfadjont and
\begin{equation}\label{eq:Commutation-Rel-phiphi}
\overline{[\varphi(f),\varphi(g)]}=2i\textup{Im}(\langle f,g\rangle).
\end{equation}
Let $A$ be a selfadjoint operator on $\cH$ with domain $\cD(A)$. Then we define the second quantisation of $A$ to be the selfadjoint operator
\begin{equation}\label{Sumdecomp}
d\Gamma(A)=0\oplus \bigoplus_{n=1}^{\infty} \overline{\sum_{k=1}^{n} (1\otimes)^{k-1} A(\otimes 1)^{n-k}}\biggl \lvert _{\cH^{\otimes_s n}}.
\end{equation}
 The number operator is defined as $N=d\Gamma(1)$. If $\cK$ is another Hilbert space and $U: \cH\rightarrow \cK$ is a bounded operator with $\lVert U\lVert\leq 1$ then we define
\begin{equation*}
\Gamma(U)=1\oplus \bigoplus_{n=1}^\infty  U^{\otimes n}\mid_{\cH^{\otimes_s n}}.
\end{equation*}
 We will write $d\Gamma^{(n)}(A)=d\Gamma(A)\mid_{\cH^{\otimes_s n}}$ and $\Gamma^{(n)}(U)=\Gamma(U)\mid_{\cH^{\otimes_s n}}$ throughout the text. If $\omega$ is a multiplication operator then $d\Gamma^{(n)}(\omega)$ is the multiplication operator defined by the map $\omega_n(k_1,\dots,k_n)=\omega(k_1)+\cdots+\omega(k_n)$. For any $v\in \cD(A)$ one has the commutation relation
\begin{equation}\label{eq:CommutatorPhi2ndQuantised}
\overline{[d\Gamma(A),\varphi(v)]}=-i\varphi(iAv)
\end{equation}
where $\cN\cap \cD(d\Gamma(A))\subset \cD([d\Gamma(A),\varphi(v)])$. We now introduce the Weyl representation. For any $g\in \cH$ we define the corresponding exponential vector
\begin{equation}\label{defn:corherentstate}
\epsilon(g)=\sum_{n=0}^{\infty} \frac{g^{\otimes n}}{\sqrt{n!}}.
\end{equation}
One may prove that if $\cD\subset \cH$ is a dense subspace then $\{ \epsilon(f)\mid f\in \cD \}$ is a linearly independent and total subset of $\cF_b(\cH)$. Write $\cU(\cH)$ for the set of unitary maps from $\cH$ into $\cH$. Let $U\in \cU(\cH)$ and $h\in \cH$. Then there is a unique unitary map $W(h,U)$ such that
\begin{equation*}
W(h,U)\epsilon(g)=e^{-\lVert h\lVert^2/2-\langle h,Ug \rangle}\epsilon(h+Ug) \,\,\,\,\, \forall g\in \cH.
\end{equation*}
One may easily check that $(h,U)\mapsto W(h,U)$ is strongly continuous and that
\begin{equation*}
W(h_1,U_1)W(h_2,U_2)=e^{-i\text{Im}(\langle h_1,U_1h_2 \rangle)}W((h_1,U_1)(h_2,U_2)),
\end{equation*}
where $(h_1,U_1)(h_2,U_2)=(h_1+U_1h_2,U_1U_2)$. If $A$ is a selfadjoint operator on $\cH$ and $f\in \cH$ we have
\begin{align*}
e^{itd\Gamma(A)}&=\Gamma(e^{itA})=W(0,e^{itA})\\
e^{it\varphi(if)}&=W(tf,1).
\end{align*}
The following lemma is important and well known (see \cite{Hirokawa1} or \cite{Lecture}):
\begin{lem}\label{Lem:FundamentalIneq}
Let $\omega$ be a selfadjoint, nonnegative and injective operator on $\cH$ and let $g_1,g_2,\dots,g_n\in \cD(\omega^{-\frac{1}{2}})$. Then $\cN\subset \cD(\varphi(g_1)\cdots\varphi(g_n))$ and $\varphi(g_1)\cdots\varphi(g_n)$ is $d\Gamma(\omega)^{\frac{n}{2}}$ bounded. We have the following bounds
	\begin{align*}
	\lVert \varphi(g_1) \psi \lVert&\leq 2 \lVert (\omega^{-\frac{1}{2}}+1)g_1 \lVert  \lVert (d\Gamma(\omega)+1)^{\frac{1}{2}}\psi \lVert \\
	\lVert \varphi(g_1)\varphi(g_2) \psi \lVert&\leq 15 \lVert (\omega^{-\frac{1}{2}}+1)g_1 \lVert \lVert (\omega^{-\frac{1}{2}}+1)g_2 \lVert \lVert (d\Gamma(\omega)+1)\psi \lVert  
	\end{align*}
	which holds on $\cD(d\Gamma(\omega)^{\frac{1}{2}})$ and $\cD(d\Gamma(\omega))$ respectively. In particular, $\varphi(g_1)$ is infinitesimally $d\Gamma(\omega)$ bounded. Furthermore, $d\Gamma(\omega)+\varphi(g_1)\geq -\lVert \omega^{-\frac{1}{2}}g_1 \lVert^2$. 
\end{lem}
\begin{lem}\label{Lem:SeconduantisedBetweenSPaces}
Let $U:\cH\rightarrow \cK$ be unitary, $A$ be a selfadjoint operator on $\cH$, $V\in \cU(\cH)$ and $f\in \cH$. Then $\Gamma(U)$ is unitary and
\begin{align*}
\Gamma(U)d\Gamma(A)\Gamma(U)^*&=d\Gamma(UAU^*).\\
\Gamma(U)W(f,V)\Gamma(U)^*&=W(Uf,UVU^*).\\
\Gamma(U)\varphi(f)\Gamma(U)^*&=\varphi(Uf).
\end{align*}
Furthermore, $\Gamma(U)(f_1\otimes_s\cdots\otimes_s f_n)=Uf_1\otimes_s\cdots\otimes_s Uf_n$ and $\Gamma(U)\Omega=\Omega$.
\end{lem}

\section{The spin-boson model}
Let $\sigma_x$, $\sigma_y$, $\sigma_z$ denote the Pauli matrices and define $e_1=(1,0)$ and $e_{-1}=(0,1)$. Note that $e_j$ is an eigenvector of $\sigma_z$ with eigenvalue $j$. We consider a qubit coupled to a radiation field. The state space of the qubit is $\CC^2$ and the energy of the qubit can be represented by $\eta \sigma_z$. Let $\cH$ be the state space of a single boson and $\omega$ be the energy operator of a single boson. Then the state space of the field is $\cF_b(\cH)$ and the energy operator of the field is $d\Gamma(\omega)$. This leads to the state space $\CC^2 \otimes \cF_b(\cH)$ for the total system and we have the Hamiltonian 
\begin{equation*}
H_{\eta}(\alpha,f,\omega):=\eta \sigma_z\otimes 1+1\otimes d\Gamma(\omega)+\sum_{i=1}^{2n}\alpha_i(\sigma_x\otimes \varphi(f_i))^i,
\end{equation*}
which is parametrised by $\alpha\in \CC^{2n}$, $f\in \cH^{2n}$, $\eta\in \CC$ and $\omega$ selfadjoint on $\cH$. We will also need the fiber operators:
\begin{equation*}
F_{\eta}(\alpha,f,\omega)=\eta\Gamma(-1)+d\Gamma(\omega)+\sum_{i=1}^{2n}\alpha_i \varphi(f_i)^i.
\end{equation*}
If the spectra are real we define
\begin{align*}
E_{\eta}(\alpha,f,\omega)&:=\inf(\sigma(H_{\eta}(\alpha,f,\omega)))\\
\cE_{\eta}(\alpha,f,\omega)&:=\inf(\sigma(F_{\eta}(\alpha,f,\omega))).
\end{align*}
For an element $f\in \cH^{2n}$ we define the leading terms
\begin{equation*}
\cL(f)=\{ i\in \{ 2,3,\dots,2n \}\mid f_i\neq f_j\,\, \forall j> i  \}.
\end{equation*}
The expression $\cL(f)^c$ is to be interpreted as the complement within $\{1, 2,\dots,2n\}$ so $1\in \cL(f)^c$ for all $f\in \cH^{2n}$. For a selfadjoint operator $\omega$ we define
\begin{equation*}
m(\omega)=\inf\{ \sigma(\omega) \} \,\,\,\,\, \text{and}\,\,\,\,\, m_{\textup{ess}}(\omega)=\inf\{ \sigma_{\textup{ess}}(\omega) \}.
\end{equation*}
\noindent The basic set of assumptions are:
\begin{hyp}\label{Hyp1}
Let $\alpha\in \CC^{2n}$, $f\in \cH^{2n}$ and $\omega$ be a selfadjoint operator in $\cH$. We say $(\alpha,f,\omega)$ satisfies Hypothesis \ref{Hyp1} if
\begin{enumerate}
\item[\textup{(1)}] $\cL(f)$ consists only of even numbers, $\alpha_i>0$ for all $i\in \cL(f)\backslash \{2 \}$ and $\alpha_2\geq 0$ if $2\in \cL(f)$. \label{Hyp1:1}

\item[\textup{(2)}] $\omega$ is injective and nonnegative.\label{Hyp1:2}

\item[\textup{(3)}] $f_j\in \cD(\omega^{-\frac{1}{2}})\cap\cD(\omega^{\frac{1}{2}})$ for all $j\in \{ 2,\dots,2n \}$ and $f_1\in \cD(\omega^{-\frac{1}{2}})$.\label{Hyp1:3}
\end{enumerate}

\end{hyp}

\begin{hyp}\label{Hyp2}
Let $f\in \cH^{2n}$, $\omega$ be a selfadjoint operator on $\cH$ and $\cM_b(\sigma(\omega),\RR)$ be the set of bounded and measurable maps from $\sigma(\omega)$ to $\RR$. We say $(f,\omega)$ satisfies Hypothesis \ref{Hyp2} if $\langle f_i,g(\omega)f_j \rangle\in \RR$ for all $i,j\in \{1,\dots,2n\}$ and $g\in \cM_b(\sigma(\omega),\RR)$.
\end{hyp}
\begin{hyp}\label{Hyp3}
Let $f\in \cH^{2n}$ and $\omega$ be a selfadjoint operator on $\cH$. We say $(f,\omega)$ satisfies Hypothesis \ref{Hyp3} if either $n\leq 2$ or $m(\omega)>0$ and $(f,\omega)$ satisfies Hypothesis \ref{Hyp2}.
\end{hyp}

\noindent For $n>2$ we need hypercontractive bounds to make our proofs work. This is the only reason we assume $m(\omega)>0$ in Hypothesis 3.
\begin{hyp}\label{Hyp4}
Let $f\in \cH^{2n}$ and $\omega$ be a selfadjoint operator on $\cH$. We say Hypothesis 4 is satisfied if either $n\leq 2$ or the following conditions are satisfied:
	\begin{enumerate}
	\item[\textup{(1)}] $\cH=L^2(\cM,\cF,\mu)$ where $(\cM,\cF,\mu)$ satisfies the assumptions in Theorem \ref{Thm:EssentalPropertyCutSpaces} and $\omega$ is a multiplication operator on $\cH$.
	
	\item[\textup{(2)}] There is a measurable function $h:\cM\rightarrow \CC$ with $\lvert h\lvert=1$ such that $h f$ is $\RR^{2n}$ valued almost everywhere. A function $h$ with these properties is called a \textbf{phase function} for $f$.
	\end{enumerate}
\end{hyp}

\begin{hyp}\label{Hyp5}
	Let $f\in \cH^{2n}$ and $\omega$ be a selfadjoint operator on $\cH$. We say $(f,\omega)$ satisfies Hypothesis \ref{Hyp5} if $f_j\in \cD(\omega^{-1})$ for all $j\in \{1,\dots, 2n \}$.
\end{hyp}

\begin{ex}
\textup{Let $\cH=L^2(\RR^\nu,\cB(\RR^\nu),\lambda^{\otimes \nu})$ where $\lambda^{\otimes \nu}$ is the Lebesgue measure. Define $\omega(k)=\sqrt{\lvert k\lvert ^2+m^2}$ and $f_1=f_2=\dots=f_{2n}=\omega^{-a/2}1_{ \{ \lvert k\lvert \leq \Lambda \} }$ for some $m\geq 0$, $a>0$ and $\Lambda>0$. In this case $m(\omega)=m=m_{\textup{ess}}(\omega)$ and $\sigma(\omega)=[m,\infty)=\sigma_{ \textup{ess}}(\omega)$. We consider the $m>0$ and the $m=0$ cases separately.}

\begin{enumerate}
	\item [$\mathbf{m> 0}:$] \textup{If $ \eta \in \RR$, $\Lambda>0$, $a>0$ and $\alpha\in \RR^{2n}$ is chosen such that $\alpha_{2n}>0$ then Hypothesis 1, 2, 3, 4 and 5 are satisfied. Furthermore, if $0<2\lvert \eta\lvert<m$ then the drawing in Figure 1 below depicts the spectra of $F_{\pm \lvert \eta\lvert }(\alpha,f,\omega)$ and $H_{ \eta}(\alpha,f,\omega)$ (see Theorem \ref{unique} part (4) below).}
	
	\item [$\mathbf{m= 0}:$] \textup{If $a<\nu-2$, $\eta \in \RR$,  $n\leq 2$, $\Lambda>0$ and $\alpha\in \RR^{2n}$ is chosen such that $\alpha_{2n}>0$ then Hypothesis 1, 2, 3, 4 and 5 are satisfied.}

	 \textup{If $a\in [\nu-2,\nu-1)$, $\eta \in \RR$, $n\leq 2$, $\Lambda>0$ and $\alpha\in \RR^{2n}$ is chosen such that $\alpha_{2n}>0$ then Hypothesis 1, 2, 3 and 4 are satisfied but Hypothesis 5 is not satisfied. Hence we cannot apply Theorem \ref{Thm:Numberstructure_Massless} below.}
	
\end{enumerate}
\end{ex}
\noindent We can now present our results
\begin{prop}\label{Lem:BasicPropertiesSBmodel}
Let $\eta\in \CC$, $\alpha\in \CC^{2n}$, $f\in \cH^{2n}$ and $\omega$ be a selfadjoint operator on $\cH$. If $(\alpha,f,\omega)$ satisfies Hypotheses \ref{Hyp1} and \ref{Hyp3} then the operators $F_{\eta}(\alpha,f,\omega)$ and $H_\eta(\alpha,f,\omega)$ are closed on the domains
\begin{align*}
\cD(F_{\eta}(\alpha,f,\omega))&=\cD(d\Gamma(\omega))\cap \bigcap_{i\in \cL(f)\backslash \{2\}}\cD( \varphi(f_i)^i)\\
\cD(H_{\eta}(\alpha,f,\omega))&=\cD(1\otimes d\Gamma(\omega))\cap \bigcap_{i\in \cL(f)\backslash \{2\}}\cD(1\otimes  \varphi(f_i)^i).
\end{align*}
Let $\cD$ be a core for $\omega$ and define 
\begin{align*}
\cJ(\cD)&:=\{\Omega \}\cup \bigcup_{n=1}^\infty \{  g_1\otimes_s\cdots \otimes_s g_n\mid g_j\in \cD \}\\ \widetilde{\cJ}(\cD)&:=\{ e_j\otimes v\mid j\in \{ -1,1 \} \,\,\, \text{and}\,\,\, v\in \cJ(\cD)  \}.
\end{align*}
Then $\cJ(\cD)$ spans a core for $F_{\eta}(\alpha,f,\omega)$ and $\widetilde{\cJ}(\cD)$ spans a core for $H_{\eta}(\alpha,f,\omega)$. Furthermore, both $F_{\eta}(\alpha,f,\omega)$ and $H_{\eta}(\alpha,f,\omega)$ are selfadjoint and semibounded if $(\alpha,\eta)\in \RR^{2n+1}$ and they have compact resolvents if $\omega$ has compact resolvents.  
\end{prop}

\begin{prop}\label{Thm:Spectral Theory of decomposition}
	Let $\phi=(\phi_1,\phi_{-1})=e_1\otimes \phi_1+e_{-1}\otimes\phi_{-1}$ be an element in $ \cF_b(\cH)^2=\cF_b(\cH)\oplus \cF_b(\cH)\approx \CC^2\otimes \cF_b(\cH)$ and write $\phi_j=(\phi^{(k)}_j)$ for $j\in \{-1,1\}$. Let $j\in \{-1,1\}$ and define $\widetilde{\phi}_{j}=(\widetilde{\phi}^{(k)}_j)$ where
	\begin{equation*}
	\widetilde{\phi}^{(k)}_j=\begin{cases} \phi^{(k)}_j & \text{if k is even.}\\ \phi^{(k)}_{-j} & \text{if k is odd.} 
	\end{cases}
	\end{equation*}
	Then $\widetilde{\phi}_j\in \cF_b(\cH)$ and the map $U:\phi \mapsto (\widetilde{\phi}_1,\widetilde{\phi}_{-1})$ is selfadjoint and unitary. Let $\eta\in \CC$, $\alpha\in \CC^{2n}$, $f\in \cH^{2n}$ and $\omega$ be a selfadjoint operator on $\cH$. Then
	\begin{align*}
	U H_\eta(\alpha,f,\omega) U^*=F_{\eta}(\alpha,f,\omega)\oplus F_{-\eta}(\alpha,f,\omega).
	\end{align*}
\end{prop}
\noindent In the remaining part of this section we will suppress $\alpha$, $f$ and $\omega$ from the notation as they are fixed in the initial part of each Theorem. The first result we present an HVZ type theorem about the location of the essential spectrum.

\begin{thm}[HVZ]\label{HVZ}
Let $\alpha\in \RR^{2n}$, $\eta\in \RR$, $f\in \cH^{2n}$ and $\omega$ be a selfadjoint operator on $\cH$.  Assume $(\alpha,f,\omega)$ satisfies Hypothesis \ref{Hyp1}, \ref{Hyp3} and \ref{Hyp4}. Then the following holds
\begin{align*}
&\inf\{ \sigma_{\textup{ess}}(F_\eta)  \}\geq \min \{ \cE_{-\eta}+m_{\textup{ess}},\cE_{\eta}+m+m_{\textup{ess}}   \}\\
&\bigcup_{q=1}^\infty \{ \cE_{(-1)^q\eta}+\lambda_1+\cdots+\lambda_q\mid \lambda_i\in \sigma_{\textup{ess}}(\omega) \}    \subset \sigma_{\textup{ess}}(F_\eta)\\
&\inf(\sigma_{\textup{ess}}(H_\eta))=E_\eta+m_{\textup{ess}}\\
&\bigcup_{q=1}^\infty \{ E_\eta+\lambda_1+\cdots+\lambda_q\mid \lambda_i\in \sigma_{\textup{ess}}(\omega) \}    \subset \sigma_{\textup{ess}}(H_\eta).
\end{align*}
In particular, $H_\eta$ has a ground state of finite multiplicity if $m_{\textup{ess}}>0$. We also have:
\begin{enumerate}
	\item[\textup{(1)}] Assume $m=m_{\textup{ess}}$, $[m_{\textup{ess}},3m_{\textup{ess}}]\subset \sigma_{\textup{ess}}(\omega)$ and $m_{\textup{ess}}$ is not isolated in $\sigma_{\textup{ess}}(\omega)$. Then $\sigma_{\textup{ess}}(F_\eta)=[\cE_{-\eta}+m_{\textup{ess}},\infty)$.
	
	\item[\textup{(2)}] Assume $[m_{\textup{ess}},2m_{\textup{ess}}]\subset \sigma_{\textup{ess}}(\omega)$ and $m_{\textup{ess}}$ is not isolated in $\sigma_{\textup{ess}}(\omega)$. Then $\sigma_{\textup{ess}}(H_\eta)=[E_\eta+m_{\textup{ess}},\infty)$.
	
	\item[\textup{(3)}] Assume in addition that $(f,\omega)$ satisfies Hypothesis 2. Then $\cE_{-\lvert \eta\lvert}\leq \cE_{\lvert \eta\lvert}$ and $ \inf(\sigma_{\textup{ess}}(F_{\lvert \eta\lvert}))= \cE_{-\lvert \eta\lvert}+m_{\textup{ess}}$. Furthermore, $\cE_{-\lvert \eta\lvert}= \cE_{\lvert \eta\lvert}$ if and only if $\eta=0$ or $m=0$.
\end{enumerate}
\end{thm}

\begin{figure}
	\includegraphics[width=\linewidth]{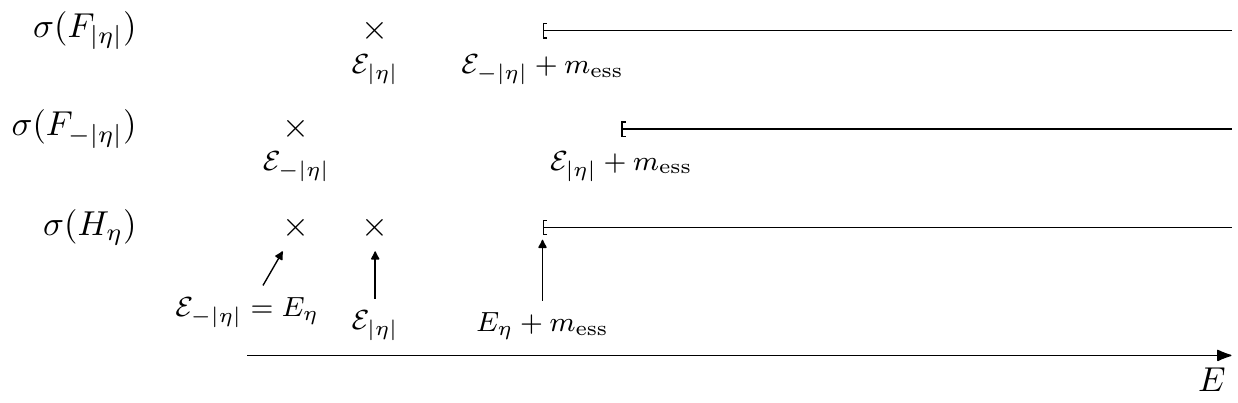}
	\caption{The picture established by Theorems \ref{HVZ} and \ref{unique} in the case $0<2\lvert \eta\lvert<m_{\textup{ess}}$, $m=m_{\textup{ess}}$ and $[m_{\textup{ess}},3m_{\textup{ess}}]\subset \sigma_{\textup{ess}}(\omega)$.}
\end{figure}

\noindent In the following result we single out which fiber operator is associated with the ground state and which fiber operator is associated with exited states.
\begin{thm}\label{unique}
Let $\alpha\in \RR^{2n}$, $\eta\in \RR$, $f\in \cH^{2n}$ and $\omega$ be a selfadjoint operator on $\cH$.  Assume $(\alpha,f,\omega)$ satisfies Hypothesis \ref{Hyp1}, \ref{Hyp2} and \ref{Hyp3}. Let $U$ be the map from Proposition \ref{Thm:Spectral Theory of decomposition}. Then the following holds:
\begin{enumerate}
		\item[\textup{(1)}] If $\eta\neq 0$ and $E_\eta$ is an eigenvalue of $H_\eta$ then $E_\eta$ is non degenerate. If $\psi$ is a ground state of $H_\eta$ then $U\psi=e_{-\textup{sign}(\eta)}\otimes \phi$ where $\phi$ is an eigenvector of $F_{-\lvert \eta\lvert}$ corresponding to the energy $E_\eta$.
		
		\item[\textup{(2)}] If $\cE_{-\lvert \eta\lvert}$ is an eigenvalue of $F_{-\lvert \eta\lvert}$ then $\cE_{-\lvert \eta\lvert}$ is non degenerate. In particular, if $E_0$ is an eigenvalue of $H_0$ then $E_0$ will have multiplicity two. Furthermore, if $\psi$ is a ground state of $H_0$ then $U\psi=e_{1}\otimes \phi_1+e_{-1}\otimes \phi_{-1}$ where $\phi_{i}$ is either 0 or an eigenvector of $F_{0}$ corresponding to the energy $E_0=\cE_0$.		
	
		\item[\textup{(3)}] Assume in addition that Hypothesis \ref{Hyp4} is satisfied. Then  $\cE_{-\lvert \eta\lvert}=E_\eta$ and $H_\eta$ has a ground state if and only if $F_{-\lvert \eta\lvert}$ has a ground state. Furthermore, if $m=0$ and $\eta\neq0$ then $F_{\lvert \eta\lvert}$ has no ground state.
		
		\item[\textup{(4)}] Assume in addition that Hypothesis \ref{Hyp4} is satisfied and $m,\eta \neq 0$. Then $H_\eta$ will have an exited state in $(E_\eta,E_\eta+m_{\textup{ess}}]$ if $F_{\lvert \eta\lvert}$ has a ground state. Furthermore, $F_{\lvert \eta\lvert}$ has a ground state if $2\lvert \eta\lvert <m_{\textup{ess}}$. 
\end{enumerate}
\end{thm}
\noindent Assuming weak infrared regularity one can prove the following theorem. Note that the assumptions imposed on $\omega$ are much weaker than in \cite{Gerard}.

\begin{thm}\label{Thm:Numberstructure_Massless}
Let $\alpha\in \RR^{2n}$, $\eta\in \RR$, $f\in \cH^{2n}$ and $\omega$ be a selfadjoint operator on $\cH$.  Assume $(\alpha,f,\omega)$ satisfies Hypothesis \ref{Hyp1}, \ref{Hyp2}, \ref{Hyp3}, \ref{Hyp4} and \ref{Hyp5}. Then the following holds:
\begin{enumerate}
	\item [\textup{(1)}] If $F_{-\lvert \eta\lvert}$ has a ground state $\psi$ and $H_\eta$ has a ground state $\phi$ then $\psi\in  \cD(N^{a})$ and $\phi \in \cD(1\otimes N^{a})$ for any $a>0$.
	
	\item [\textup{(2)}] Assume in addition that $\cH=L^2(\RR^\nu,\cB(\RR^\nu),\lambda^{\otimes \nu})$ and $\omega$ is a multiplication operator on $\cH$. Then $E_\eta$ is an eigenvalue for $F_{-\lvert \eta\lvert}$ and $H_\eta$.
\end{enumerate}
\end{thm}

\section{Important estimates}
In this section we prove series of estimates which will become useful later. We start with the following lemma
\begin{lem}\label{Lem:FundamentalLowerbound}
Let $\alpha\in \RR^{2n}$ and define 
\begin{equation*}
\cK(\alpha)=\{ f\in \cH^{2n}\mid  \text{$(\alpha,f)$ satisfies part (\ref{Hyp1:1}) of Hypothesis \ref{Hyp1}  }  \}.
\end{equation*}
There is a constant $C:=C(\alpha)$, such that for any collection $\{ A(v) \}_{v\in \cH}$ of selfadjoint operators and $f\in \cK(\alpha)$ we have
\begin{equation}
\sum_{j=2}^{2n} \alpha_jA(f_j)^j\geq C.
\end{equation}
\end{lem}
\begin{proof}
Let $K=\{ i\in \{2,4,\dots,2n\}\mid  \alpha_i>0 \}=\{ i_1,\dots,i_k \}$. For each $b\leq k$ we consider polynomials of the form
\begin{equation*}
\alpha_{i_b}X^{i_b}+\sum_{j=2}^{i_b-1}\widetilde{\alpha}_jX^j,  
\end{equation*}
where $\widetilde{\alpha}_j$ is either $0$ or $\alpha_j$. Since there are only finitely many choices of $b$ and $\widetilde{\alpha}_j$ we find a uniform lower bound $C_0<0$ of all these polynomials. Using the spectral theorem we find
\begin{equation}\label{eq:SimplePolyOperator}
\alpha_{i_b}A^{i_b}+\sum_{j=2}^{i_b-1}\widetilde{\alpha}_jA^j\geq C_0,
\end{equation}
for all selfadjoint operators $A$, $b\in \{ 1,\dots,k \}$ and choices of $\widetilde{\alpha}_j$ as either $0$ or $\alpha_j$. Let $f\in \cK(\alpha)$ and $\{ A(v) \}_{v\in \cH}$ be a collection of selfadjoint operators. Define
\begin{align*}
B:=\sum_{j=2}^{2n} \alpha_jA(f_j)^j=\sum_{j\in \cL(f)} \sum_{i=2}^{j} \widetilde{\alpha}_{i,j} A(f_j)^i
\end{align*}
where $\widetilde{\alpha}_{i,j}=\alpha_i$ if $f_i=f_j$ and $0$ otherwise. Note that either $\cL(f)\subset K$ or $\alpha_2=0$ and $\cL(f)\backslash \{2\}\subset K$. In any case we see $B$ is the sum of at most $k\leq n$ operators of the form given in equation (\ref{eq:SimplePolyOperator}). Hence $B\geq nC_0$ finishing the proof.
\end{proof}

\begin{lem}\label{Lem:Ulighednmindre2}
Let $\varepsilon>0$, $ r>0$ and $\omega$ be a selfadjoint, nonnegative and injective operator on $\cH$. There is $C:=C(r,\varepsilon)$ such that for all $v_1,v_2\in \cD(\omega^{-\frac{1}{2}})$ and $a,b\geq 0$ with $\max\{ \lVert (1+\omega^{-\frac{1}{2}}) v_1 \lVert,\lVert (1+\omega^{-\frac{1}{2}}) v_2 \lVert,a,b\}\leq r$ we have 
\begin{equation*}
2\textup{Re}(\langle a\varphi(v_1)^4  \psi, b\varphi(v_2)^2 \psi \rangle  )\geq -\varepsilon \lVert d\Gamma(\omega) \psi \lVert^2- C\lVert \psi \lVert^2
\end{equation*}
for all $\psi\in \cN\cap \cD(d\Gamma(\omega))$. 
\end{lem}
\begin{proof}
Let $\psi\in \cN\cap \cD(d\Gamma(\omega))$.  Using equation (\ref{eq:Commutation-Rel-phiphi}) we find
\begin{equation*}
\varphi(v_2)\varphi(v_1)^4\psi=\varphi(v_1)^4\varphi(v_2)\psi+4(2i\text{Im}(\langle v_2,v_1 \rangle))\varphi(v_1)^3\psi.
\end{equation*}
This implies
\begin{align*}
2\text{Re}(\langle a\varphi(v_1)^4  \psi, b\varphi(v_2)^2 \psi \rangle) &=2ab\lVert \varphi(v_1)^2\varphi(v_2) \psi \lVert^2\\&\quad +16ab\text{Im}(\langle v_2,v_1 \rangle)\text{Im}(\langle \varphi(v_1)^3  \psi, \varphi(v_2) \psi \rangle ).
\end{align*}
Now
\begin{equation*}
\text{Im}(\langle \varphi(v_1)^3  \psi, \varphi(v_2) \psi \rangle )=\frac{1}{2i}\langle [\varphi(v_2),\varphi(v_1)^3]  \psi,  \psi \rangle=-3\text{Im}(\langle v_2,v_1 \rangle)  \lVert \varphi(v_1)\psi \lVert^2.
\end{equation*}
Hence we find 
\begin{equation*}
2\text{Re}(a\langle \varphi(v_1)^4  \psi, b\varphi(v_2)^2 \psi \rangle)\geq -48r^6\lVert \varphi(v_1)\psi \lVert^2.
\end{equation*}
Using Cauchy-Schwarz inequality and Lemma \ref{Lem:FundamentalIneq} we find
\begin{align*}
\lVert \varphi(v_1)\psi \lVert^2\leq 4\lVert (\omega^{-\frac{1}{2}}+1)v_1 \lVert^2 (\langle \psi, d\Gamma(\omega) \psi \rangle+\lVert \psi\lVert^2)\leq 4r^2(\langle \psi, d\Gamma(\omega) \psi \rangle+\lVert \psi\lVert^2),
\end{align*}
so
\begin{align*}
2\text{Re}(\langle a\varphi(v_1)^4  \psi, b\varphi(v_2)^2 \psi \rangle)&\geq -196r^8\lVert  \psi\lVert \lVert  d\Gamma(\omega) \psi \lVert-196r^8\lVert \psi\lVert^2\\&\geq -\varepsilon\lVert d\Gamma(\omega)\psi \lVert^2-196r^8\lVert \psi\lVert^2-\frac{(196r^8)^2}{4\varepsilon}\lVert \psi\lVert^2,
\end{align*}
which finishes the proof.
\end{proof}

\begin{lem}\label{Lem:CommutatorDGamma}
Let $\varepsilon>0$, $r>0$, $ n\in \mathbb{N} $ and $\omega$ be a selfadjoint, nonnegative and injective operator on $\cH$. There is $C:=C(r,\varepsilon,n)$ such that for all $v \in  \cD(\omega^{\frac{1}{2}})$ and $a\geq 0$ with $\max\{\lVert \omega^{\frac{1}{2}} v \lVert,a\}\leq r$ we have  
\begin{equation*}
2\textup{Re}(\langle a\varphi(v)^{2n}  \psi, d\Gamma(\omega) \psi \rangle  )\geq - \varepsilon \lVert  a\varphi(v)^{2n} \psi \lVert^2- C\lVert \psi \lVert^2
\end{equation*}
for all $\psi \in \cN\cap \cD(d\Gamma(\omega))$.
\end{lem}
\begin{proof}
Let $\psi \in \cN\cap \cD(d\Gamma(\omega))$ and $P_\omega$ be the spectral measure of $\omega$. Define the bounded operator $\omega_k=\int_{\RR}\max\{ k,\lambda \}dP_\omega(\lambda)$. Using equation (\ref{eq:CommutatorPhi2ndQuantised}) we find \begin{align*}
\varphi(v)^n d\Gamma(\omega_k)\psi=d\Gamma(\omega_k)\varphi(v)^{n}\psi+i\sum_{j=0}^{n-1}\varphi(v)^{n-j-1}\varphi(i\omega_kv)\varphi(v)^{j}\psi.
\end{align*}
This yields
\begin{align*}
2\text{Re}(\langle a\varphi(v)^{2n}  \psi, d\Gamma(\omega_k) \psi \rangle  )&=2a\lVert d\Gamma(\omega_k)^{\frac{1}{2}}\varphi(v)^{n} \psi \lVert^2\\&\quad-2a\sum_{j=0}^{n-1}\text{Im}(\langle\varphi(v)^n \psi, \varphi(v)^{n-j-1}\varphi(i\omega_kv)\varphi(v)^{j}  \psi \rangle ).
\end{align*}
For each $j\in \{0, \dots ,n-1 \}$ we have
\begin{equation*}
\text{Im}(\langle  \varphi(v)^n \psi,\varphi(v)^{n-j-1}\varphi(i\omega_kv)\varphi(v)^{j}  \psi \rangle )=\frac{1}{2i}\langle [\varphi(v)^{n-j-1}\varphi(i\omega_kv)\varphi(v)^{j},\varphi(v)^n ]  \psi,  \psi \rangle.
\end{equation*}
Using equation (\ref{eq:Commutation-Rel-phiphi}) we may calculate
\begin{align*}
[\varphi(v)^{n-j-1}\varphi(i\omega_kv)\varphi(v)^{j},\varphi(v)^n ]\psi &=\varphi(v)^{n-j-1} [\varphi(i\omega_kv),\varphi(v)^n] \varphi(v)^{j}\psi \\&=n2i\text{Im}(-i\langle \omega_kv,v \rangle)\varphi(v)^{2(n-1)}\psi.
\end{align*}
Using the above equalities we find
\begin{align*}
2\text{Re}(\langle a\varphi(v)^{2n}  \psi, d\Gamma(\omega) \psi \rangle)&\geq -2an^2\lVert \omega^{\frac{1}{2}}_kv\lVert^2\lVert \varphi(v)^{n-1}\psi \lVert ^2\\&= -2a^{1/n}n^2\lVert  \omega^{\frac{1}{2}}_kv\lVert^2\lVert (a^{\frac{1}{2n}}\varphi(v))^{n-1}\psi \lVert ^2.
\end{align*}
 For any $\varepsilon'>0$ there is a constant $A$ depending only on $\varepsilon'$ and $n$ such that $x^{2(n-1)}\leq \varepsilon' x^{4n}+A$ for all $x\in \RR$. Pick $A$ corresponding to $\varepsilon'=2^{-1}n^{-2}r^{-2-1/n}\varepsilon$ and note that
\begin{equation*}
2\text{Re}(\langle a\varphi(v)^{2n}  \psi, d\Gamma(\omega_k) \psi \rangle)\geq -\varepsilon \lVert a\varphi(v)^{2n}\psi \lVert^2 - 2n^2Ar^{2+1/n}\lVert \psi\lVert^2 
\end{equation*}
since $\lVert \omega_k^{\frac{1}{2}}v\lVert^2\leq \lVert \omega^{\frac{1}{2}}v\lVert^2\leq r^2$ for all $k\in \NN$. Taking $k$ to $\infty$ finishes the proof.
\end{proof}
\begin{lem}\label{Lem:DominatingEstimate}
Let $r>0$, $\varepsilon\in (0,1)$, $n\in \mathbb{N}$ and $\omega$ be a selfadjoint, nonnegative and injective operator on $\cH$. Define
\begin{align*}
\cK&=\{ (\alpha,v)\in [0,\infty)\times (\cD(\omega^{1/2})\cap \cD(\omega^{1/2})  )\mid \max\{\alpha,\lVert (1+\omega^{-\frac{1}{2}} + \omega^{\frac{1}{2}}) v\lVert\}<r \}\\
\cA&=\begin{cases}
\cH^n  &  n\leq 2.\\
\{ v\in \cH^n \mid  \langle v_i,v_j \rangle\in \RR \,\,\, \forall \, i,j\in \{1,...,n\} \} & n>2.
\end{cases}
\end{align*}
There is a constant $C:=C(\varepsilon,r,n)$ such that for all $(\alpha_1,v_1),\dots,(\alpha_n,v_n)\in \cK$ with $v=(v_1,\dots,v_n)\in \cA$ we have 
\begin{equation*}
\lVert d\Gamma(\omega)\psi  \lVert^2+\sum_{j=1}^n \lVert  \alpha_j \varphi(v_j)^{2j}\psi \lVert^2 \leq  \frac{1}{1-\varepsilon}\biggl \lVert d\Gamma(\omega)\psi+\sum_{j=1}^n \alpha_j \varphi(v_j)^{2j}\psi \biggl \lVert^2+C \lVert \psi\lVert^2
\end{equation*}
for all $\psi \in \cN\cap \cD(d\Gamma(\omega))$.
\end{lem}
\begin{proof}
Let $\psi \in \cN\cap \cD(d\Gamma(\omega))$. First we note that	
\begin{align*}
\lVert d\Gamma(\omega)\psi  \lVert^2+\sum_{j=1}^n \lVert  \alpha_j \varphi(v_j)^{2j}\psi \lVert^2 &=  \biggl \lVert d\Gamma(\omega)\psi+\sum_{j=1}^n \alpha_j \varphi(v_j)^{2j}\psi \biggl \lVert^2\\&-\sum_{j=1}^{n}2\text{Re}(\langle \alpha_j\varphi(v_j)^{2j}  \psi, d\Gamma(\omega) \psi \rangle  )\\&-\sum_{j_1=1}^n\sum_{j_2=j_1+1}^n2\alpha_{j_2}\alpha_{j_1}\text{Re}(\langle \varphi(v_{j_1})^{2j_1}  \psi, \varphi(v_{j_2})^{2j_2} \psi \rangle.
\end{align*}
For $q\in \NN$ we let $\widetilde{C}(r,\varepsilon,q)$ be the constant from Lemma \ref{Lem:CommutatorDGamma}. Define $C_1=\widetilde{C}(r,\varepsilon,1)+\cdots+\widetilde{C}(r,\varepsilon,n)$ which depends only on $n,r$ and $\varepsilon$. Then we find
\begin{equation*}
-\sum_{j=1}^{n}2\text{Re}(\langle \alpha_j\varphi(v_j)^{2j}  \psi, d\Gamma(\omega) \psi \rangle  )\leq \sum_{j=1}^n \varepsilon\lVert  \alpha_j \varphi(v_j)^{2j}\psi \lVert^2+C_1\lVert \psi \lVert^2.
\end{equation*}
We now estimate the double sum. If $n\leq 2$ we have at most one term which can be estimated using Lemma \ref{Lem:Ulighednmindre2}. Therefore we find a constant $C_2>0$ such that
\begin{equation*}
-\sum_{j_1=1}^n\sum_{j_2=j_1+1}^n2\alpha_{j_2}\alpha_{j_1}\text{Re}(\langle \varphi(v_{j_1})^{2j_1}  \psi, \varphi(v_{j_2})^{2j_2} \psi \rangle)\leq  \varepsilon \lVert d\Gamma(\omega)\psi\lVert^2+ C_2\lVert \psi \lVert^2.
\end{equation*}
If $n>2$ then $\varphi(v_j)$ and $\varphi(v_i)$ commute on $\cN$ for all $i,j\in \{1, \dots, n \}$ so 
\begin{align*}
&-\sum_{j_1=1}^n\sum_{j_2=j_1+1}^n2\alpha_{j_2}\alpha_{j_1}\text{Re}(\langle \varphi(v_{j_1})^{2j_1}  \psi, \varphi(v_{j_2})^{2j_2} \psi \rangle)\\&=-\sum_{j_1=1}^n\sum_{j_2=j_1+1}^n2\alpha_{j_2}\alpha_{j_1}\lVert  \varphi(v_{j_1})^{j_1}  \varphi(v_{j_2})^{j_2} \psi \lVert^2\leq 0 \leq \varepsilon \lVert d\Gamma(\omega)\psi\lVert^2+ C_2\lVert \psi \lVert^2.
\end{align*}	
Combining the above inequalities we find the desired result with $C=\frac{C_1+C_2}{1-\varepsilon}$.
\end{proof}

\begin{lem}\label{Lem:TheRelativeBounds}
Let $r>0$, $\varepsilon\in (0,1)$, $n\in \mathbb{N}$ and $\omega$ be a selfadjoint operator on $\cH$. There is a constant $C:=C(r,\varepsilon,n)>0$ such that for all $f\in \cH^{2n}$, $\alpha\in \CC^{2n}$ and $\eta \in \CC$ such that $(\alpha,f,\omega)$ satisfies Hypothesis \ref{Hyp1}, Hypothesis \ref{Hyp3} and
\begin{equation*}
\max_{j\in \cL(f)\backslash \{2\}  }\{ \alpha_j^{-1}   \}+\lvert \eta \lvert + \lvert \alpha \lvert+\lVert (\omega^{-\frac{1}{2}}+1)f_1 \lVert+ \sum_{j=2}^{2n}\lVert (\omega^{-\frac{1}{2}}+1+\omega^{\frac{1}{2}})f_j \lVert<r
\end{equation*}
we have
\begin{align*}
\lVert \eta \Gamma(-1)\psi \lVert +\sum_{j\in \cL(f)^c} \lVert  \alpha_j \varphi(f_j)^{j}\psi \lVert&\leq \varepsilon \biggl \lVert  d\Gamma(\omega)\psi+\sum_{j\in \cL(f)} \alpha_j \varphi(f_j)^j\psi  \biggl \lVert+ C\lVert \psi \lVert \\
\max\{\lVert  \alpha_j \varphi(f_j)^{j}\psi \lVert,\lVert  d\Gamma(\omega)\psi \lVert \}&\leq (1-\varepsilon)^{-2} \left \lVert  F_\eta(\alpha,f,\omega) \psi \right \lVert+ C\lVert \psi \lVert. 
	\end{align*}
 for all $\psi\in \cN\cap \cD(d\Gamma(\omega))$ and $j\in \cL(f)$.
\end{lem}
\begin{proof}
Let $\psi\in \cN\cap \cD(d\Gamma(\omega))$. Pick $C_1>0$ depending only on $r$, $\varepsilon$ and $n$ such that
\begin{equation*}
r^2\sum_{j=1}^{2\ell-1} x^{2j}\leq \frac{\varepsilon^2}{16n^4r^2}  x^{4\ell}+C_1,
\end{equation*}
for all $\ell \in \{ 1,\dots,n \}$ and $x\in \RR$. For each $j\in \cL(f)^c\backslash \{1\}$ we find  $q\in \cL(f)\backslash\{2\}$ such that $f_j=f_{q}$ and $j< q$. Noting that $\alpha_q^{-1}\leq r$ and $\alpha_q\leq r$ we find 
\begin{equation*}
 \lVert \alpha_j \varphi(f_j)^{j}\psi \lVert\leq \frac{\varepsilon}{4n^2r} \lVert  \varphi(f_q)^q\psi \lVert+ \sqrt{C_1}\lVert \psi\lVert \leq  \frac{\varepsilon}{4n^2} \lVert \alpha_q \varphi(f_q)^q\psi \lVert+ \sqrt{C_1}\lVert \psi\lVert.
\end{equation*}
We know from Lemma \ref{Lem:FundamentalIneq} that
\begin{align*}
\lVert \alpha_1 \varphi(f_1)\psi \lVert&\leq 2r\lVert (\omega^{-\frac{1}{2}}+1)f_1 \lVert (\langle \psi, d\Gamma(\omega) \psi \rangle+\lVert \psi\lVert^2)^{\frac{1}{2}}\\&\leq \frac{\varepsilon}{4n} \lVert d\Gamma(\omega)\psi \lVert +  2r^2\lVert \psi \lVert+\frac{4r^4n}{\varepsilon}\lVert \psi \lVert.
\end{align*}
Combining the two inequalities above, we find $C_2>0$ depending only on $r$, $\varepsilon$ and $n$ such that
\begin{align*}
\sum_{j\in \cL(f)^c} \lVert  \alpha_j \varphi(f_j)^{j}\psi \lVert&\leq \frac{\varepsilon}{4n} \biggl(  \lVert d\Gamma(\omega)\psi \lVert+\sum_{j\in \cL(f)}   \lVert  \alpha_j \varphi(f_j)^j\psi \lVert   \biggl) + C_2\lVert \psi \lVert\\&\leq \frac{\varepsilon}{2} \biggl(  \lVert d\Gamma(\omega)\psi \lVert^2+\sum_{j\in \cL(f)}   \lVert  \alpha_j \varphi(f_j)^j\psi \lVert^2   \biggl) + C_2\lVert \psi \lVert.
\end{align*}
Applying Lemma \ref{Lem:DominatingEstimate} (with $\varepsilon=\frac{1}{2}$) there is a constant $C_3$ (depending only on $r$, $\varepsilon$ and $n$) such that
\begin{equation}\label{eq:FirstImpartantineq}
\sum_{j\in \cL(f)^c} \lVert  \alpha_j \varphi(f_j)^{j}\psi \lVert\leq \varepsilon \biggl \lVert  d\Gamma(\omega)\psi+\sum_{j\in \cL(f)} \alpha_j \varphi(f_j)^j\psi  \biggl \lVert+ C_3\lVert \psi \lVert.
\end{equation}
This proves the first inequality. To prove the next inequality we note that 
\begin{equation*}
\biggl \lVert  d\Gamma(\omega)\psi+\sum_{j\in \cL(f)} \alpha_j \varphi(f_j)^j\psi  \biggl \lVert\leq \lVert F_\eta(\alpha,f,\omega)\psi\lVert+\sum_{j\in \cL(f)^c} \lVert  \alpha_j \varphi(f_j)^{j}\psi \lVert+ \lvert \eta \lvert \lVert \psi\lVert.
\end{equation*}
Using equation (\ref{eq:FirstImpartantineq}) we obtain
\begin{equation*}
\biggl \lVert  d\Gamma(\omega)\psi+\sum_{j\in \cL(f)} \alpha_j \varphi(f_j)^j\psi  \biggl \lVert\leq \frac{1}{1-\varepsilon}\lVert F_\eta(\alpha,f,\omega)\psi\lVert+\frac{C_3+r}{1-\varepsilon}\lVert \psi\lVert. 
\end{equation*}
Combining this and Lemma \ref{Lem:DominatingEstimate} we find a constant $C_4$ such that
\begin{equation*}
\max\{\lVert \alpha_q\varphi(f_q)^q\psi \lVert,\lVert d\Gamma(\omega)\psi \lVert\}\leq  (1-\varepsilon)^{-2}\lVert F_\eta (\alpha,f,\omega)\psi\lVert+C_4\lVert \psi\lVert
\end{equation*} 
for all $q\in \cL(f)$. Taking $C=\max\{ C_3+r,C_4 \}$ finishes the proof.
\end{proof}

\section{Proof of Proposition \ref{Lem:BasicPropertiesSBmodel} and Proposition \ref{Thm:Spectral Theory of decomposition}}
\begin{lem}\label{Lem:TrasformationCanTrans}
The map $U$ defined in Proposition \ref{Thm:Spectral Theory of decomposition} is unitary with inverse $U^*=U$. Furthermore, for any $v\in \cH$ and selfadjoint operator $A$ on $\cH$ we have
\begin{align}\label{eq:111}
U (\sigma_x\otimes \varphi(v) )U^*&=\varphi(v)\oplus \varphi(v)=1\otimes \varphi(v)\\
U( 1\otimes d\Gamma(A)) U^*&=d\Gamma(A)\oplus d\Gamma(A)=1\otimes d\Gamma(A)\label{eq:112}\\
U( \sigma_z\otimes 1 )U^*&=\Gamma(-1)\oplus -\Gamma(-1)=\sigma_z\otimes \Gamma(-1).\label{eq:113}
\end{align}
Let $\alpha\in \CC^{2n}$, $f\in \cH^{2n}$, $\eta \in \CC$ and $\omega$ be a selfadjoint operator on $\cH$. Then
\begin{equation}\label{eq:Fuldtrans}
UH_\eta(\alpha,f,\omega)U^*=F_{\eta}(\alpha,f,\omega)\oplus F_{-\eta}(\alpha,f,\omega).
\end{equation}
\end{lem}
\begin{proof}
First we note that
\begin{align*}
\sum_{k=0}^{\infty} \lVert \widetilde{\psi}^{(k)}_1 \lVert^2+\sum_{k=0}^{\infty} \lVert \widetilde{\psi}^{(k)}_{-1} \lVert^2&=\sum_{k \,\, \textup{even}}\lVert \psi^{(k)}_1 \lVert^2+\sum_{k \,\, \textup{odd}}\lVert \psi^{(k)}_{-1} \lVert^2+\sum_{k \,\, \textup{even}}\lVert \psi^{(k)}_{-1} \lVert^2+\sum_{k \,\, \textup{odd}}\lVert \psi^{(k)}_1 \lVert^2\\&=\lVert \psi_1 \lVert^2+\lVert \psi_{-1} \lVert^2=\lVert e_1\otimes \psi_1+e_{-1}\otimes \psi_{-1} \lVert^2
\end{align*}
which shows that the $\widetilde{\psi}_i$ are elements in Fock space and $U$ gives rise to an isometric map from $\cF_b(\cH)^2$ to $\cF_b(\cH)^2$. Let $(\psi_1,\psi_{-1})\in \cF_b(\cH)^2$ and write $U^2(\psi_1,\psi_{-1})=U(\widetilde{\psi}_1,\widetilde{\psi}_{-1})=(\phi_1,\phi_{-1})$. For $j\in \{1,-1\}$ and $k\in \NN_0$ we have
\begin{align*}
\phi_j^{(k)}=\begin{cases}
\widetilde{\psi}_j^{(k)} & \text{k even}\\
\widetilde{\psi}_{-j}^{(k)} & \text{k odd}
\end{cases}=\begin{cases}
\psi_j^{(k)} & \text{k even}\\
\psi_j^{(k)} & \text{k odd}
\end{cases}=\psi_j^{(k)}.
\end{align*}
So $U$ is bijective with inverse $U^{-1}=U$. It is clear from the definition of $\widetilde{\psi}_j$ that the map $(\psi_1,\psi_{-1})\mapsto \widetilde{\psi}_j$ is linear and hence $U$ is also linear. We have thus proven that $U$ is unitary with $U=U^{-1}=U^*$.

It remains to prove equations (\ref{eq:111}), (\ref{eq:112}) (\ref{eq:113}) and (\ref{eq:Fuldtrans}). Equation (\ref{eq:Fuldtrans}) follows from the others so we are left with proving equations (\ref{eq:111}), (\ref{eq:112}) and (\ref{eq:113}). Define $\widetilde{\cJ}(\cD(A))$ as in Proposition \ref{Lem:BasicPropertiesSBmodel} with $\cD(A)$ instead of $\cD$. Then $\widetilde{\cJ}(\cD(A))$ spans a core for the operators on the left hand side of equations (\ref{eq:111}), (\ref{eq:112}) and (\ref{eq:113}). Hence it is enough to check equations (\ref{eq:111}), (\ref{eq:112}) and (\ref{eq:113}) hold on elements of the form $e_j\otimes \Omega$ and $e_j\otimes g_1\otimes_s\cdots\otimes_s g_k$ with $j\in \{ \pm 1 \}$ and $g_\ell\in \cD(A)$. Now
\begin{align*}
U^*(e_j\otimes \Omega)&=e_j\otimes \Omega\\
U^*(e_j\otimes (g_1\otimes_s\dots\otimes_s g_k))&=e_{(-1)^kj}\otimes (g_1\otimes_s\cdots\otimes_s g_k),
\end{align*}
which is in the domain of $\sigma_x\otimes \varphi(v),1\otimes d\Gamma(A)$ and $\sigma_z\otimes 1$. Using $\sigma_xe_j=e_{-j}$ and $\sigma_ze_j=je_{j}$ we find
\begin{align*}
\sigma_x\otimes \varphi(v)(e_j\otimes \Omega)&=e_{-j}\otimes v=U^*(1\otimes \varphi(v))( e_j\otimes \Omega)\\
\sigma_x\otimes \varphi(v)(e_{(-1)^kj}\otimes (g_1\otimes_s\cdots\otimes_s g_k))&=e_{(-1)^{k+1}j}\otimes a^{\dagger}(v)g_1\otimes_s\cdots\otimes_s g_k\\&\quad +e_{(-1)^{k-1}j}\otimes a(v)g_1\otimes_s\cdots\otimes_s g_k\\&=U^*(1\otimes \varphi(v))(e_j\otimes (g_1\otimes_s\cdots\otimes_s g_k))\\
1\otimes d\Gamma(A)(e_j\otimes \Omega)&=0=U^*(1\otimes d\Gamma(A))( e_j\otimes \Omega)\\
1\otimes d\Gamma(A)(e_{(-1)^kj}\otimes (g_1\otimes_s\cdots\otimes_s g_k))&=e_{(-1)^{k}j}\otimes d\Gamma(A)g_1\otimes_s\cdots\otimes_s g_k\\&=U^*(1\otimes d\Gamma(A))(e_j\otimes (g_1\otimes_s\cdots\otimes_s g_k))\\
\sigma_z\otimes 1(e_j\otimes \Omega)&=je_{j}\otimes \Omega=U^*(\sigma_z\otimes \Gamma(-1))
(e_j\otimes \Omega)  \\
\sigma_z\otimes 1(e_{(-1)^kj}\otimes (g_1\otimes_s\cdots\otimes_s g_k))&=(-1)^kj e_{(-1)^{k}j}\otimes (g_1\otimes_s\cdots\otimes_s g_k)\\&=je_{(-1)^kj}\otimes \Gamma(-1)g_1\otimes_s\cdots\otimes_s g_k   \\&=U^* (\sigma_z\otimes \Gamma(-1)) (e_j\otimes (g_1\otimes_s\cdots\otimes_s g_k)).
\end{align*}
This finishes the proof.
\end{proof}
\noindent Proposition \ref{Lem:BasicPropertiesSBmodel}  will follow as soon as we prove the statements for $F_{\eta}(\alpha,f,\omega)$. We start by proving the following lemma
\begin{lem}
The conclusions of Proposition \ref{Lem:BasicPropertiesSBmodel} hold under Hypothesis \ref{Hyp1}, Hypothesis \ref{Hyp3} and the assumption  
\begin{equation*}
L:=d\Gamma(\omega)+\sum_{j\in \cL(f)} \alpha_j\varphi(f_j)^j
\end{equation*}
is essentially selfadjoint on $\cN\cap \cD(d\Gamma(\omega))$.
\end{lem}
\begin{proof}
Combining the assumptions with Lemma \ref{Lem:DominatingEstimate} we see that $L$ is selfadjoint on
\begin{align*}
\cC=\cD(d\Gamma(\omega))\cap \bigcap_{j\in \cL(f)\backslash\{2\} }\cD( \varphi(f_j)^j ).
\end{align*}
Simple perturbation theory along with Lemma \ref{Lem:TheRelativeBounds} now shows that $F_{\eta}(\alpha,f,\omega)$ is closed on $\cC$ and any core for $L$ is a core for $F_{\eta}(\alpha,f,\omega)$. If $\alpha\in \RR^{2n}$ and $\eta\in \RR$, the Kato-Rellich Theorem shows that $F_{\eta}(\alpha,f,\omega)$ is selfadjoint and bounded below.

We now prove that $\cJ(\cD)$ spans a core for $L$. Any element $\psi\in \cN\cap \cD(\omega)$ can be approximated in $d\Gamma(\omega)$-norm by a sequence $\{\psi_j  \}_{j=1}^\infty\subset \textup{Span}(\cJ(\cD))$. Pick $c>0$ such that $1_{(-\infty,c)}(N)\psi=\psi$ and write $P=1_{(-\infty,c)}(N)$. Then $\{P\psi_j\}_{j=1}^\infty$ converges to $\psi$ in $\cD(d\Gamma(\omega))$-norm and in $N^{n}$ norm. It follows from Lemma \ref{Lem:FundamentalIneq} that $\{P\psi_j\}_{j=1}^\infty$ converges to $\psi$ in $L$-norm. Using that $\cN\cap \cD(d\Gamma(\omega))$ is a core for $L$ by assumption, we see $\cJ(\cD)$ spans a core for $L$.

If $\omega$ has compact resolvents then so does $d\Gamma(\omega)$ by Lemma \ref{Lem:SecondQuantisedProp}. That $F_{\eta}(\alpha,f,\omega)$ has compact resolvents will now follow from the equality
\begin{align*}
(F_{\eta}(\alpha,f,\omega)+i)^{-1}&=(d\Gamma(\omega)+i)^{-1}\\&+(d\Gamma(\omega)+i)^{-1}(F_{\eta}(\alpha,f,\omega)-d\Gamma(\omega))(F_{\eta}(\alpha,f,\omega)+i)^{-1}.
\end{align*}
This finishes the proof.
\end{proof}
 
\begin{proof}[Proof of Proposition \ref{Lem:BasicPropertiesSBmodel}]
It remains to prove that \begin{equation*}
L:=d\Gamma(\omega)+\sum_{j\in \cL(f)} \alpha_j\varphi(f_j)^j
\end{equation*}
is essentially selfadjoint on $\cN\cap \cD(d\Gamma(\omega))$ under Hypothesis \ref{Hyp1} and \ref{Hyp3}. The case $n\leq 2$ is simply done by appealing to \cite{Arai}. If $n>2$ one appeals to the theory of hypercontractive semigroups and obtains $L$ is essentially selfadjoint on $\cap_{n\in \mathbb{N}} \cD(d\Gamma(\omega)^n)$ (See Lemma \ref{Lem:ExsistenceOfNiceRealSubspace}, Theorem \ref{Thm:PropertiesOfQspace} and \cite[Theorem X.58]{RS2}). 

Recall that a vector $g\in \cH$ is said to be bounded for $\omega$ if $g\in \cap_{k\in \mathbb{N}} \cD(\omega^k)$ and there is $C>0$ such that $\lVert \omega^k g\lVert\leq C^k\lVert g\lVert $ for all $k\in \NN$. The set of vectors which are bounded for $\omega$ is dense in $\cH$ since
\begin{equation*}
g=\lim\limits_{\ell\rightarrow \infty}1_{[-\ell,\ell]}(\omega)g
\end{equation*}
for any $g\in \cH$. Let $g_1,\dots,g_q$ be bounded for $\omega$ and note that $g_1\otimes_s\cdots\otimes_s g_q\in\cap_{k\in \mathbb{N}} \cD(d\Gamma(\omega)^k)$ and
\begin{align*}
\lVert d\Gamma(\omega)^kg_1\otimes_s\cdots\otimes_s g_q\lVert &=\biggl\lVert \sum_{\alpha\in \NN_0^q, \lvert \alpha\lvert= k} \begin{pmatrix}
k \\ \alpha
\end{pmatrix} \omega^{\alpha_1}g_1\otimes_s\cdots\otimes_s \omega^{\alpha_q}g_q\biggl\lVert\\& \leq \sum_{\alpha\in \NN_0^q, \lvert \alpha\lvert= k} \begin{pmatrix}
k \\ \alpha
\end{pmatrix} C_1^{\alpha_1}\cdots C_q^{\alpha_n}\lVert g_1\lVert\cdots\lVert g_q\lVert\\&\leq (C_1+\cdots+C_q)^k \lVert g_1\lVert\cdots\lVert g_q\lVert.
\end{align*}
Hence $g_1\otimes_s\cdots\otimes_s g_q$ is an analytic vector for $d\Gamma(\omega)^{n}$ which implies
\begin{align*}
\{ \Omega \}\cup \bigcup_{q=1}^{\infty}\{ g_1\otimes_s\cdots\otimes_s g_q\mid g _i \,\, \text{is bounded for} \,\ \omega  \}\subset \cN\cap \cD(d\Gamma(\omega))
\end{align*}
will span a core for $d\Gamma(\omega)^{n}$ by Nelsons analytic vector theorem. Since $L$ is $d\Gamma(\omega)^{n}$ bounded by Lemma \ref{Lem:FundamentalIneq}, we find that elements from $\cN\cap \cD(d\Gamma(\omega))$ can approximate every element in $\cD(d\Gamma(\omega)^n)$ with respect to the graph norm of $L$. This finishes the proof because $L$ is essentially selfadjoint on $\cD(d\Gamma(\omega)^{n})$.
\end{proof}

\section{Lemmas for the HVZ theorem}
In this chapter we discuss some of the technical machinery needed to prove the HVZ theorem.
\begin{lem}\label{Lem:GeneralTransform}
Let $f\in \cH^{2n}$, $\alpha\in \RR^{2n}$, $\eta\in \RR$ and $\omega$ be a selfadjoint operator on $\cH$. Assume $(\alpha,f,\omega)$ satisfies Hypothesis \ref{Hyp1} and \ref{Hyp3}. If there is a unitary map $V:\cH\rightarrow \cH_1\oplus \cH_2$ such that $Vf_i=(\widetilde{f}_i,0)$ for all $i\in \{1,\dots,2n\}$ and $V\omega V^*=\omega_1\oplus \omega_2$ then $(\alpha,\widetilde{f},\omega_1)$ satisfies Hypothesis 1 and 3. Furthermore, there is a unitary map 
\begin{equation*}
U:\cF_b(\cH)\rightarrow \cF_b(\cH_1)\oplus \bigoplus_{k=1}^{\infty} \left (\cF_b(\cH_1)\otimes \cH_2^{\otimes_s k}\right)
\end{equation*}
such that
\begin{equation*}
UF_{\eta}(\alpha,f,\omega)U^*=F_{\eta}(\alpha,\widetilde{f},\omega_1)\oplus \bigoplus_{k=1}^\infty \left (F_{(-1)^k\eta}(\alpha,\widetilde{f},\omega_1)\otimes 1+1\otimes d\Gamma^{(k)}(\omega_2)\right).
\end{equation*}
In fact, $U=U_2U_1\Gamma(V)$ where $U_1$ is the unitary map from Theorem \ref{Thm:ISOTHM1} and $U_2$ is the unitary map from Theorem \ref{Thm:ISOTHM2}.
\end{lem}

\begin{proof}
It is easy to see that Hypothesis 1 and 3 are preserved under the isomorphism. Using Lemma \ref{Lem:SeconduantisedBetweenSPaces} one calculates
\begin{equation*}
\Gamma(V)F_{\eta}(\alpha,f,\omega)\Gamma(V)^*=\eta\Gamma(-1\oplus -1)+d\Gamma(\omega_1\oplus \omega_2 )+\sum_{i=1}^{2n} \alpha_i  \varphi((\widetilde{f}_i,0))^i.
\end{equation*}
	Let $U_1$ be the isomorphism from Theorem \ref{Thm:ISOTHM1}. Using Theorem \ref{Thm:SpectralPropTensor} part (6) and Theorem \ref{Thm:ISOTHM1} we see
	\begin{align*}
	U_1\Gamma(V)F_{\eta}(\alpha,f,\omega)\Gamma(V)^*U_1^*&=\eta\Gamma(-1)\otimes \Gamma(-1)+F_{0}(\alpha,\widetilde{f},\omega_1)\otimes 1+1\otimes d\Gamma(\omega_2).
	\end{align*}
 Let $U_2$ be the unitary transform from Theorem \ref{Thm:ISOTHM2}. Defining $U=U_2U_1\Gamma(V)$ we calculate
	\begin{align*}
	UF_{\eta}(\alpha,f,\omega)U^*&= \eta U_2\Gamma(-1)\otimes \Gamma(-1)U_2^*+U_2F_{0}(\alpha,\widetilde{f},\omega_1)\otimes 1U_2^*+U_21\otimes d\Gamma(\omega_2)U_2^*\\&=\left(\eta\Gamma^{(0)}(-1)\Gamma(-1)+F_{0}(\alpha,\widetilde{f},\omega_1)\right)\\&\oplus \bigoplus_{k=1}^\infty\left( \eta\Gamma(-1)\otimes \Gamma^{(k)}(-1)+F_{0}(\alpha,\widetilde{f},\omega_1)\otimes 1+1\otimes d\Gamma^{(k)}(\omega_2)\right).
	\end{align*}
	The fact that $\Gamma^{(k)}(-1)=(-1)^k$ finishes the proof.
\end{proof}	

\begin{lem}\label{Lem:FundamentalTransform}
Let $f\in \cH^{2n}$, $\alpha\in \RR^{2n}$, $\eta\in \RR$ and $\omega$ be a selfadjoint operator on $\cH$. Assume $(\alpha,f,\omega)$ satisfies Hypotheses \ref{Hyp1} and \ref{Hyp3}. Let $\cH_1,\cH_2\subset \cH$ be closed subspaces such that $\cH_1^\perp=\cH_2$ and let $P_i$ denote the orthogonal projection onto $\cH_i$. If $f\in \cH_1^{2n}$ and $\omega$ is reduced by $\cH_1$, then we may take $\omega_i=\omega\mid_{\cH_i}$ and $Vf=(P_1f,P_2f)$ in Lemma \ref{Lem:GeneralTransform}. Let $U$ be the corresponding map. For $g_1,\dots,g_q\in \cH_2$ we define
\begin{align*}
B&=\{ \Omega \} \cup \bigcup_{b=1}^\infty \{ h_1\otimes_s\cdots\otimes_s h_b\mid h_i\in \cH_1\cap \cD(\omega) \}\\
C&=\{ g_1\otimes_s \cdots \otimes_s g_q \} \cup \bigcup_{b=1}^\infty \{ h_1\otimes_s\cdots \otimes_s h_b\otimes_s g_1\otimes_s \cdots\otimes_s g_q\mid h_i\in \cH_1\cap \cD(\omega) \}.
\end{align*}
If $\psi\in \textup{Span}(B)$ then we may interpret $\psi$ as an element in both $\cF_b(\cH)$ and $\cF_b(\cH_1)$. Using this identification for $\psi$ we find that
\begin{align}\label{eq:SimpelTrans1}
U^*(\psi\otimes (g_1\otimes_s \cdots \otimes_s g_q))&\in \textup{Span}(C).\\
U^*(\psi)&=\psi.\label{eq:SimpelTrans2}\\
\lVert (F_\eta(\alpha,f,\omega)-\lambda)\psi\lVert&=\lVert (F_\eta(\alpha,f,\omega_1)-\lambda)\psi\lVert.\label{eq:SimpelTrans3}
\end{align}
for all $\lambda\in \CC$.
\end{lem}
\begin{proof}
$V$ is clearly unitary and satisfies the properties needed in Lemma \ref{Lem:GeneralTransform}. Let $j:\cH_1\rightarrow \cH_1\oplus \cH_2$ be the embedding $j(f)=(f,0)$ and define $Q=V^*j$. Then $Q$ is the inclusion map from $\cH_1$ into $\cH$. Lemma \ref{Thm:ISOTHM3} immediately yields equation $(\ref{eq:SimpelTrans1})$ and
\begin{align*}
\Gamma(Q)=U^*\mid_{\mathcal{F}_b(\mathcal{H}_1)}.
\end{align*}
This map acts as the identity on the set spanning $B$ proving equation $(\ref{eq:SimpelTrans2})$. To prove equation (\ref{eq:SimpelTrans3}) we note $\psi=UU^*\psi=U\psi$ and so
	\begin{equation*}
	\lVert (F_\eta(\alpha,f,\omega)-\lambda)\psi\lVert=\lVert U(F_\eta(\alpha,f,\omega)-\lambda)U^*U\psi\lVert=\lVert (F_\eta(\alpha,f,\omega_1)-\lambda)\psi\lVert.
	\end{equation*}
	This finishes the proof.
\end{proof}

\begin{lem}\label{Lem:UniformUpperBounds}
Let $\{ f^{k} \}_{k=1}^\infty \subset \cH^{2n}$, $\alpha\in \RR^{2n}$, $\eta\in \RR$ and $\omega$ be a selfadjoint operator on $\cH$. Assume $\cL(f^{k})=\cL(f^{1})$ for all $k\in \NN$ and $(\alpha,f^k,\omega)$ satisfies Hypothesis \ref{Hyp1} and \ref{Hyp3}. Assume furthermore that
\begin{equation*}
C:=\sup_{ k\in \mathbb{N},q\in \{2,\dots,2n\} } \{ \lVert f_q^k\lVert,\lVert \omega^{\pm \frac{1}{2}} f_q^k\lVert,\lVert \omega^{- \frac{1}{2}} f_1^k\lVert ,\lVert f_1^k\lVert \}<\infty.
\end{equation*}
For each $\lambda\in \RR$ there is $K<\infty$ such that
\begin{align*}
 \max\{\lVert \varphi(f^k_q)^j(F_\eta(\alpha,f^k,\omega)+\lambda\pm i)^{-1} \lVert,\lVert d\Gamma(\omega)(F_\eta(\alpha,f^k,\omega)+\lambda\pm i)^{-1}  \lVert \} \leq K
\end{align*}
for all $k\in \NN$, $q\in \{1,\dots,2n\}$ and $j\in \{ 1,\dots, q  \}$.
\end{lem}
\begin{proof}
Define
\begin{equation*}
r=  2n3C+\lvert \alpha\lvert+\lvert \eta \lvert+\left (\max_{q\in \cL(f^1)\backslash\{2\}}\alpha_q^{-1}\right )+1
\end{equation*}
and $\varepsilon=\frac{1}{2}$. By Lemma \ref{Lem:TheRelativeBounds} there is a $k$ independent constant $\widetilde{C}>0$ such that
\begin{equation*}
\max\{\lVert \alpha_q \varphi(f_q^k)^q\psi \lVert, \lVert d\Gamma(\omega)\psi \lVert\} \leq 4\lVert F_{\eta}(\alpha,f^k,\omega) \psi\lVert +\widetilde{C}\lVert \psi\lVert
\end{equation*}
for all $\psi \in \cN\cap \cD( d\Gamma(\omega) )$, $q\in \cL(f^1)\backslash \{2\}$ and $k\in \NN$. Now $\cD(d\Gamma(\omega))\cap \cN$ is a core for $F_{\eta}(\alpha,f^k,\omega)$ and so the inequality extends to all $\psi\in \cD(F_{\eta}(\alpha,f^k,\omega))$. Using 
\begin{align*}
\lVert F_{\eta}(\alpha,f^k,\omega)(F_{\eta}(\alpha,f^k,\omega)\pm i+\lambda)^{-1}\lVert\leq 2+\lvert \lambda\lvert 
\end{align*}
and $ \alpha_q^{-1}\leq r$ for all $q\in \cL(f^1)\backslash \{2\}$ we obtain the following bounds
\begin{align}\label{eq:fundamentalresest1}
\lVert \varphi(f_q^k)^q(F_{\eta}(\alpha,f^k,\omega)\pm i+\lambda)^{-1}\psi \lVert&\leq r(8+4\lvert \lambda\lvert +\widetilde{C})\lVert  \psi\lVert\\\label{eq:fundamentalresest2} \lVert d\Gamma(\omega)(F_{\eta}(\alpha,f^k,\omega)\pm i+\lambda)^{-1}\psi \lVert&\leq (8+4\lvert \lambda\lvert +\widetilde{C})\lVert  \psi\lVert
\end{align}
for all $\psi\in \cD(F_{\eta}(\alpha,f^k,\omega))$, $q\in \cL(f^1)\backslash \{2\}$ and $k\in \NN$. Let $q\in \{1,\dots,2n\}$, $j\in \{1,\dots, q  \}$ and $k\in \NN$. If $j\leq 2$ and $\psi\in \cD(F_{\eta}(\alpha,f^k,\omega))\subset \cD(d\Gamma(\omega))$ we apply Lemma \ref{Lem:FundamentalIneq} and obtain
\begin{equation*}
\lVert \varphi(f_q^k)^j\psi \lVert\leq 15(2r)^j \lVert (d\Gamma(\omega)+1)^{j/2}\psi \lVert\leq 60r^2 \lVert (d\Gamma(\omega)+1)\psi \lVert. 
\end{equation*}
Equation (\ref{eq:fundamentalresest2}) now gives $\lVert \varphi(f_q^k)^j(F_{\eta}(\alpha,f^k,\omega)\pm i+\lambda)^{-1} \lVert\leq 60r^2(8+4\lvert \lambda\lvert +\widetilde{C}+1)$ which is a uniform upper bound. If $j\geq 3$ we may find $p\in \cL(f^1)\backslash \{2\}$ such that $j\leq q\leq p$ and $f^k_p=f^k_q$. For $\psi\in \cD(F_{\eta}(\alpha,f^k,\omega))\subset \cD(\varphi(f^k_p)^p)$ we have
\begin{equation*}
\lVert \varphi(f_q^k)^j\psi \lVert\leq \lVert \varphi(f_p^k)^p\psi \lVert+\lVert \psi\lVert. 
\end{equation*}
Using equation (\ref{eq:fundamentalresest1}) we find $\lVert \varphi(f_q^k)^j(F_{\eta}(\alpha,f^k,\omega)\pm i+\lambda)^{-1} \lVert\leq r(8+4\lvert \lambda\lvert +\widetilde{C})+1$ which is a uniform upper bound.
\end{proof}

\noindent Next is a crucial result regarding convergence of operators.
\begin{lem}\label{Lem:ResolevntKonv}
Assume $\cH=L^2(\cM,\cF,\mu)$ where $(\cM,\cF,\mu)$ is $\sigma$-finite. Let  $\alpha\in \RR^{2n}$, $\eta\in \RR$, $\omega,\omega_1,\omega_2,\dots$ be a collection of multiplication operators on $L^2(\cM,\cF,\mu)$ and $f,f^1,f^2,\dots$ be a collection of elements from $\cH^{2n}$ such that $\cL(f)=\cL(f^{k})$ for all $k\in \NN$. Assume that $(\alpha,f,\omega)$ and $(\alpha,f^k,\omega_k)$ satisfies Hypothesis \ref{Hyp1} and \ref{Hyp3} for all $k\in \NN$. Assume also 
\begin{equation*}
\lim_{k\rightarrow \infty }\frac{\omega_k}{\omega}=1=\lim_{k\rightarrow \infty }\frac{\omega}{\omega_k}
\end{equation*}
in $L^{\infty }(\cM,\cF,\mu)$ and that
\begin{align}\label{eq:Konv1}
\lim_{k\rightarrow \infty }f^{k}_1=f_1 && \lim_{k\rightarrow \infty }\omega_k^{-\frac{1}{2}}f^{k}_1=\omega^{-\frac{1}{2}}f_1\\\label{eq:Konv2}
\lim_{k\rightarrow \infty }f^{k}_j=f_j && \lim_{k\rightarrow \infty }\omega_k^{\pm \frac{1}{2}}f^{k}_j=\omega^{\pm \frac{1}{2}}f_j
\end{align}
in $\cH$ for all $j\in \{2,3,\dots, 2n  \}$. If $n> 2$ we assume in addition that there is a function $h:\cM\rightarrow S^1\subset \CC$ such that $hf$ and $hf_k$ are almost surely $\RR^{2n}$-valued for all $k\in \NN$. Then $F_\eta(\alpha,f^{k},\omega_k)-\lambda_k$ converges to $F_\eta(\alpha,f,\omega)-\lambda$ in norm resolvent sense whenever $\{ \lambda_k \}_{k=1}^\infty\subset \RR$ converges to $\lambda$. 
\end{lem}
\begin{proof}
For convenience we will sometimes write $\omega=\omega_\infty$ or $f=f^{\infty}$ throughout this proof. We check convergence at the point $i$ in the resolvent set. Since $\omega_k/\omega$ and $\omega/\omega_k$ are essentially bounded functions we see $(\alpha,f^k,\omega)$ fulfils Hypothesis \ref{Hyp1} and \ref{Hyp3}. Furthermore, the limits in equations (\ref{eq:Konv1}) and (\ref{eq:Konv2}) also exist if we write $\omega$ instead of $\omega_k$ since $\omega_k/\omega$ and $\omega/\omega_k$ converges to 1 in $L^{\infty }(\cM,\cF,\mu)$. We now prove
\begin{equation}\label{eq:Resdiff}
(F_\eta(\alpha,f^k,\omega_k)+\lambda_k-i)^{-1}-(F_\eta(\alpha,f^k,\omega)+\lambda-i)^{-1}
\end{equation}
converges to 0 as $k$ tends to $\infty$ since this will reduce the problem to the case where $\omega_k=\omega$ and $\lambda_k=\lambda=0$ for all $k\in \NN$. For any $\psi \in \cF_b(\cH)$ and $k,k'\in \NN\cup \{ \infty \}$ have
\begin{align*}
\sum_{\ell=1}^{\infty}&\int_{\cM^\ell}(\omega_k(k_1)+\cdots+\omega_k(k_\ell))^2\lvert \psi^{(\ell)}(k_1,\dots,k_\ell)\lvert^2  d\mu^{\otimes \ell}(k_1,\dots,k_\ell)\\& \leq \left \lVert \frac{\omega_{k}}{\omega_{k'}}\right \lVert_{\infty}^2\sum_{\ell=1}^{\infty}\int_{\cM^\ell}(\omega_{k'}(k_1)+\cdots+\omega_{k'}(k_\ell))^2\lvert \psi^{(\ell)}(k_1,\dots,k_\ell)\lvert^2 d\mu^{\otimes \ell}(k_1,\dots,k_\ell).
\end{align*}
so $\cD(d\Gamma(\omega_k))=\cD(d\Gamma(\omega))$ for all $k\in \mathbb{N}$. On this set $\lVert (d\Gamma(\omega_k)-d\Gamma(\omega)) \psi \lVert^2$ is now estimated by
\begin{align*}
\sum_{\ell=1}^{\infty}&\int_{\cM^\ell}(\omega_k(k_1)-\omega(k_1)+\cdots+\omega_k(k_\ell)-\omega(k_\ell))^2\lvert \psi^{(\ell)}(k_1,\dots,k_\ell)\lvert^2  d\mu^{\otimes \ell}(k_1,\dots,k_\ell)\\& \leq \left \lVert \frac{\omega_{k}-\omega}{\omega}\right \lVert_{\infty}^2\sum_{\ell=1}^{\infty}\int_{\cM^\ell}(\omega(k_1)+\cdots+\omega(k_\ell))^2\lvert \psi^{(\ell)}(k_1,\dots,k_\ell)\lvert^2 d\mu^{\otimes \ell}(k_1,\dots,k_\ell).
\end{align*}
Defining $C_k=\lVert \frac{\omega_{k}-\omega}{\omega} \lVert_{\infty}$ we find 
\begin{align*}
\lVert (F_\eta(\alpha,f_k,\omega_k)+\lambda_n-i)^{-1}&-(F_\eta(\alpha,f_k,\omega)+\lambda-i)^{-1}\lVert\\& \leq \lvert \lambda_k-\lambda\lvert+C_k \lVert d\Gamma(\omega)(F_\eta(\alpha,f_k,\omega)+\lambda-i)^{-1}\lVert.
\end{align*}
$\lVert d\Gamma(\omega)(F_\eta(\alpha,f_k,\omega)+\lambda-i)^{-1}\lVert$ is uniformly bounded by Lemma \ref{Lem:UniformUpperBounds} and $C_k$ converges to 0. Thus the operator in equation (\ref{eq:Resdiff}) converges to 0 and so we have reduced to the case where $\omega_k=\omega$ and $\lambda_k=\lambda=0$ for all $k\in \NN$.

Assume $n>2$ and let $\cH_\RR$ be the real Hilbert space from Lemma \ref{Lem:ExsistenceOfNiceRealSubspace} corresponding to the elements $f^k_j$ for $k\in \NN\cup \{ \infty \}$ and $j\in \{1,\dots,2n\}$. Let $L^2(X,\cX,\QQ)$ be a Q-space corresponding to $\cH_\RR$ and $V$ be the unitary map from Theorem \ref{Thm:PropertiesOfQspace}. Define
\begin{equation}\label{eq:defnIK}
I(f^k)=\alpha_1\varphi(f_1^k)+\sum_{j=2}^{2n}\alpha_j\varphi(f_j^k)
\end{equation}
for all $k\in \NN\cup \{\infty \}$. By Theorem \ref{Thm:PropertiesOfQspace} we know that $Ve^{-t d\Gamma(\omega)}V^*$ is hypercontractive and the interaction terms $VI(f^k)V^*$ are a multiplication operators. Convergence in norm resolvent sense now follows from Theorem \ref{Thm:PropertiesOfQspace} and \cite[Theorem X.60]{RS2} if $\eta=0$. For $\eta \neq 0$ we apply Lemma \ref{Lem:KonvOfBoundedPer}.

Assume now $n\leq 2$ and define $I(f^k)$ as in equation (\ref{eq:defnIK}) for all $k\in \NN\cup \{\infty \}$. Write $F(f):=F_\eta(\alpha,f,\omega)$ and $F(f^k):=F_\eta(\alpha,f^k,\omega)$ for all $k\in \NN$. Define
\begin{align*}
C_k&=  \max_{b\in \{ 0,1 \}  }  \{ \lVert  \varphi(f-f^k)\varphi(f_k)^b(d\Gamma(\omega)+1)^{-1}\lVert, \lVert  \varphi(f-f^k)\varphi(f)^b(d\Gamma(\omega)+1)^{-1}\lVert  \}   \\
D&=\sup_{a\in \{ 0, \dots ,3 \},  k\in \NN\cup \{\infty \}}  \{ \lVert  \varphi(f_k)^a (F(f^k)\pm i)^{-1} \lVert, \lVert (d\Gamma(\omega)+1)(F(f^k)\pm i)^{-1}  \lVert  \}<\infty
\end{align*}
where $D<\infty$ follows from Lemma \ref{Lem:UniformUpperBounds}. On $\cN$ we may calculate
\begin{align*}
I(f^{k})-I(f)&=\alpha_1\varphi(f^k_1-f)+\alpha_2(\varphi(f^k_2)\varphi(f^{k}_2-f_2)+\varphi(f^{k}_2-f_2)\varphi(f_2))\\&+\alpha_3(\varphi(f^k_3)^2\varphi(f^{k}_3-f_3)+\varphi(f^k_3)\varphi(f^{k}_3-f_3)\varphi(f_3)+\varphi(f^{k}_3-f_3)\varphi(f_3)^2)\\&+\alpha_4\varphi(f^k_4)^3\varphi(f^{k}_4-f_4)+\alpha_4\varphi(f^k_4)^2\varphi(f^{k}_4-f_4)\varphi(f_4)\\&+\alpha_4\varphi(f^k_4)\varphi(f^{k}_4-f_4)\varphi(f_4)^2+\alpha_4\varphi(f^{k}_4-f_4)\varphi(f_4)^3.
\end{align*}
Let $k\in \mathbb{N}$ and define $\cA=(F(f)-i)(\cD(d\Gamma(\omega))\cap \cN) $. Then $\cA$ is a dense subspace of $\cF_b(\cH)$ since $\cD(d\Gamma(\omega))\cap \cN$ is a core for $F(f)$. Let $\phi\in \cF_b(\cH)$ and $\psi\in \cA$. Then $(F(f)-i)^{-1} \psi\in \cD(d\Gamma(\omega))\cap \cN\subset \cD(F(f^{k}))$ and so
\begin{align*}
\langle \phi, ((F(f^k)-i)^{-1}&-(F(f)-i)^{-1}) \psi \rangle\\&=\langle (F(f^k)+i)^{-1}\phi ,  (I(f)-I(f^{k}))(F(f)-i)^{-1} \psi \rangle.
\end{align*}
This is a sum of 10 terms of the form
\begin{align*}
&-\alpha_j \langle \varphi(f^k_j)^a (F(f^k)+i)^{-1}\phi,\varphi(f^k_j-f_j) \varphi(f_j)^b (F(f)-i)^{-1}\psi \rangle\\ &-\alpha_j\langle \varphi(f_j-f_j^k)\varphi(f^k_j)^b (F(f^k)+i)^{-1}\phi, \varphi(f_j)^a (F(f)-i)^{-1}\psi \rangle.
\end{align*}
with $a\in \{0,1,2,3\}$ and $b\in \{0,1 \}$. Hence we see that
\begin{equation*}
\lvert \langle \phi, ((F(f^k)-i)^{-1}-(F(f)-i)^{-1}) \psi \rangle \lvert\leq 10 \max\{ \lvert \alpha_1\lvert,\lvert \alpha_2\lvert,\lvert \alpha_3\lvert,\lvert \alpha_4\lvert \}D^2C_k\lVert \psi \lVert \lVert \phi \lVert. 
\end{equation*}
Using that $\cA$ is dense we find
\begin{align*}
\lVert (F(f^k)-i)^{-1}-(F(f)-i)^{-1} \lVert\leq 10 \max\{ \lvert \alpha_1\lvert,\lvert \alpha_2\lvert,\lvert \alpha_3\lvert,\lvert \alpha_4\lvert \}D^2C_k.
\end{align*}
$C_k$ converges to $0$ by Lemma \ref{Lem:FundamentalIneq} which finishes the proof.
\end{proof}

\begin{lem}\label{Lem:SimpleConvApplication}
Let $\cH=L^2(\cM,\cF,\mu)$ where $(\cM,\cF,\mu)$ is $\sigma$-finite, $\alpha\in \RR^{2n}$, $\eta\in \RR$, $f\in \cH^{2n}$ and $\omega:\cM\rightarrow \RR$ be measurable. Assume $(\alpha,f,\omega)$ satisfies Hypothesis \ref{Hyp1}, \ref{Hyp3} and \ref{Hyp4}. Let $\{ A_n \}_{n=1}^\infty$ be an increasing sequence of sets covering $\cM$ up to a zeroset and define $f^k=1_{A_k}f$. Then $(\alpha,f^k,\omega)$ satisfies Hypothesis \ref{Hyp1}, \ref{Hyp3} and \ref{Hyp4} for all $k\in \NN$. Furthermore, $F_\eta(\alpha,f^k,\omega)$ is uniformly bounded from below and converges to $F_\eta(\alpha,f,\omega)$ in norm resolvent sense. In particular
\begin{equation*}
\lim_{k\rightarrow \infty} \cE_\eta(\alpha,f^k,\omega)=\cE_\eta(\alpha,f,\omega)
\end{equation*}
and if $\{ \lambda_k \}_{k=1}^\infty\subset \RR$ converges to $\lambda$ where $\lambda_k\in \sigma_{\textup{ess}}(  F_\eta(\alpha,f^k,\omega)  )$ for all $k\in \NN$ then $\lambda\in  \sigma_{\textup{ess}}(  F_\eta(\alpha,f,\omega)  )$.
\end{lem}
\begin{proof}
$(\alpha,f^k,\omega)$ satisfies Hypothesis \ref{Hyp1} obviously. If $n\leq 2$ then Hypothesis \ref{Hyp3} and \ref{Hyp4} are automatically fulfilled. If $n>2$ the phase function for $f$ is also a phase function for $f^k$. Since Hypothesis \ref{Hyp4} parts (2) and (3) implies Hypothesis \ref{Hyp2}, we have proven Hypothesis 3 holds as well. In conclusion, we have proven that $(\alpha,f^k,\omega)$ satisfies Hypothesis \ref{Hyp1}, \ref{Hyp3} and \ref{Hyp4}.

Norm resolvent convergence follows directly from Lemma \ref{Lem:ResolevntKonv}. Write
\begin{equation*}
F_\eta(\alpha,f^k,\omega)=\eta\Gamma(-1)+d\Gamma(\omega)+\alpha_1\varphi(f_1^{k})+\sum_{j=2}^{2n}\alpha_j\varphi(f^{k}_j)^j.
\end{equation*}
Using Lemmas \ref{Lem:FundamentalIneq} and \ref{Lem:FundamentalLowerbound} we find a uniform lower bound of $F_\eta(\alpha,f^k,\omega)$. The remaining claims follows from standard spectral theory.
\end{proof}

\section{The HVZ theorem}
In this section we prove Theorem \ref{HVZ} except for part (3). Let $\eta \in \RR$, $\alpha\in \RR^{2n}$, $f\in \cH^{2n}$ and $\omega$ be a selfadjoint operator on $\cH$. Assume $(\alpha,f,\omega)$ satisfies Hypothesis \ref{Hyp1}, \ref{Hyp3} and \ref{Hyp4}. We introduce the notation $F_{(-1)^k}:=F_{(-1)^k \eta }(\alpha,f,\omega)$, $\cE_{(-1)^k}:=\cE_{(-1)^k \eta}(\alpha,f,\omega)$, $m:=m(\omega)$ and $m_{\textup{ess}}:=m_{\textup{ess}}(\omega)$.

Since spectral properties are conserved under unitary transformations we may  (using Lemmas \ref{Lem:SeconduantisedBetweenSPaces} and \ref{Lem:CannonicalSpace}) assume that $\cH=L^2(\cM,\cF,\mu)$ where $(\cM,\cF,\mu)$ satisfies the assumptions in Theorems \ref{Thm:EssentalPropertyCutSpaces} and $\omega$ is a multiplication operator on $\cH$ with $\omega>0$ almost everywhere.
\begin{lem}\label{Lem:PointsInTheEssentialSpectrum}
$ \{ \cE_{(-1)^q}+\lambda_1+\cdots+\lambda_q\mid \lambda_i\in \sigma_{\textup{ess}}(\omega)   \}\subset \sigma_{\textup{ess}}(F_1)$ for all $q\in \NN$.
\end{lem}
\begin{proof}
	Let $q\in \mathbb{N}$ and $\lambda_1,\dots,\lambda_q\in \sigma_{\textup{ess}}(\omega)$. By Theorem \ref{Thm:EssentalPropertyCutSpaces} we may if each $i\in \{1,\dots,q\}$ pick a collection of sets $\{A^i_k\}_{k=1}^\infty$ such that  $0<\mu(A^i_k)<\infty$, $\lvert \omega-\lambda_i\lvert\leq \frac{1}{k}$ on $A^{i}_k$, $A^{i}_k\cap A^{j}_\ell=\emptyset$ if $i\neq j$ or $k\neq \ell$ and
	\begin{equation*}
	\sum\limits_{k=1}^{\infty}\mu(A^i_k)<\infty.
	\end{equation*}
	Define for each and $k\in \mathbb{N}$ the set 
	\begin{equation*}
	B_k=\bigcup_{i=1}^q\bigcup_{j=k}^\infty A_{j}^i \,\, \Rightarrow \,\, \mu(B_k)= \sum_{i=1}^{q}\sum_{j=k}^{\infty}\mu(A_j^i)<\infty,
	\end{equation*}
	and note that $\mu(B_k)\downarrow 0$. Since $\{B_k\}_{k=1}^\infty$ is a decreasing collection of sets we find
	\begin{equation*}
	B=\bigcap_{k=1}^\infty B_k
	\end{equation*}  
	has measure 0. Define for each $\ell \in \mathbb{N}$ the subspace
	\begin{equation*}
	\cH_\ell=\{ f\in \cH \mid 1_{B_\ell^c}f=f\,\, \mu-a.e. \}=1_{B^c_\ell}\cH.
	\end{equation*}
	Assume first that $f \in \cH^{2n}_K$ for some $K$ and hence that $f \in \cH^{2n}_\ell$ for all $\ell \geq K$. Define the following subspaces for $\ell\geq K$
	\begin{equation*}
	\cA_\ell=\bigcup_{k=K}^\ell\cH_k\cap \cD(\omega)=\cH_\ell\cap \cD(\omega) \,\,\,\,\,\,\,\,\,\,\,\,\,	\cA_\infty=\bigcup_{k=K}^\infty\cH_k\cap \cD(\omega)
	\end{equation*}
	We now claim that $\cA_\infty$ is a core for $\omega$. If $\phi\in \cD(\omega)$ then $\phi_k=\phi 1_{B_k^c}\in \cA_\infty$ for all $k\geq K$ and using dominated convergence we find
	\begin{align*}
	\lim_{k\rightarrow \infty }\lVert \phi-\phi_k \lVert^2=0=\lim_{k\rightarrow \infty } \lVert \omega(\phi-\phi_k) \lVert^2&
	\end{align*}
	so $\cA_\infty$ is a core for $\omega$. Defining
	\begin{align*}
	\cJ(\cA_\infty)&=\{ \Omega \}\cup \bigcup_{k=1}^\infty \{ g_1\otimes_s\cdots\otimes_s g_k\mid g_i\in \cA_\infty  \}\\
	\cJ(\cA_\ell )&=\{ \Omega \}\cup \bigcup_{k=1}^\infty \{ g_1\otimes_s\cdots\otimes_s g_k\mid g_i\in \cA_\ell  \}
	\end{align*}
	we find that $\cJ(\cA_\infty)$ spans a core for $F_{\pm 1}$ by Proposition \ref{Lem:BasicPropertiesSBmodel}. 
	
	Let $g\in \textup{Span}(\cJ(\cA_\infty))$. Then $g=a\Omega+\sum_{j=1}^{k}\alpha_j f_j$ with $f_j\in \cJ(\cA_\infty)\backslash \{ \Omega \}$ for all $j\in \{  1, \dots, k \}$. Let $j\in \{  1, \dots, k \}$ and note that $f_j\in \cJ(\cA_{\ell(j)})$ for some $\ell(j)\in \NN$ by definition of $\cA_\infty$. Defining $u=\max_{j\in \{1,\dots,k  \}}\{ \ell(j) \}$ we see that $g\in \textup{Span}(\cJ(\cA_\ell))$ for any $\ell\geq u$. Hence we have now proven the following statements
	\begin{itemize}
	\item For any $g\in \textup{Span}(\cJ(\cA_\infty))$ there is $u\in \mathbb{N}$ with $u\geq K$ such that $g\in \textup{Span}(\cJ(\cA_\ell))$ for any $\ell\geq u$.
	\item $\textup{Span}(\cJ(\cA_\infty))$ is a core for $F_{\pm1}$.
	\end{itemize}
     For each $p\in \mathbb{N}$ we pick $\nu_p\in \textup{Span}(\cJ(\cA_\infty))$ such that $\lVert (F_{(-1)^q}-\cE_{(-1)^q})\nu_p \lVert\leq \frac{1}{p}$ and $\lVert \nu_p\lVert= 1$. Pick $u(p)\geq K$ such that $\nu_p\in \textup{Span}(\cJ(\cA_\ell))$ for any $\ell\geq u(p)$ and $u(p+1)>u(p)$ for all $p\in \NN$. For each $p\in \mathbb{N}$ and $i\in \{1,\dots,q\}$ we define
	\begin{equation*}
	g^p_i=\mu(A^i_{u(p)})^{-\frac{1}{2}}1_{A^i_{u(p)}}.
	\end{equation*}
	 Note $g^p_i\in \cD(\omega)$ since $\omega$ is bounded by $\lambda_i+\frac{1}{u(p)}$ on  $A^i_{u(p)}$. Note also, that $g^p_i\in \cH_{u(p)}^{\perp}$ since $A^i_{u(p)}\subset B_{u(p)}$ so $g^p_i$ and any element in $\cH_{u(p)}$ have disjoint support. Furthermore, the collection $\{g^p_i\}_{p\in \mathbb{N},i\in \{1,\dots,q\}}$ is orthogonal since the elements have disjoint support. Let $U_p$ be the unitary map from Lemma \ref{Lem:FundamentalTransform} corresponding to $\cH_{u(p)}$ which exists since $f\in \cH^{2n}_{K}\subset \cH^{2n}_{u(p)}$. Define
	\begin{equation*}
	\phi_p=\sqrt{q!}U^*_p (\nu_p\otimes g_1^p\otimes_s\cdots\otimes_s g_q^p).
	\end{equation*}
	We are done in the case $f\in \cH^{2n}_K$ for some $K$ if we can prove that $\{\phi_p\}_{p=1}^\infty$ is a Weyl sequence for $F_{1}$ corresponding to the value $\cE_{(-1)^q}+\lambda_1+\cdots+\lambda_q$. We check:
	\begin{enumerate}
	\item $\phi_p\in \cD(F_{1})$.
	
	\item $\lVert\phi_p\lVert=1$ for all $p\in \mathbb{N}$.
	
	\item $\phi_p$ is orthogonal to $\phi_r$ for $p\neq r$.
	
	\item  $\lVert (F_{(-1)^q}-(\cE_{(-1)^q}+\lambda_1+\cdots+\lambda_q) )\phi_p \lVert $ converges to 0.
	\end{enumerate}

	(1): Lemma \ref{Lem:FundamentalTransform} shows $\phi_p\in \cN\cap \cD(d\Gamma(\omega)) \subset  \cD(F_{1})$ for all $p\in \mathbb{N}$.

	(2): Note that $\{g_i^p \}_{i=1}^q$ is orthonormal for each $p\in \NN$. Let $\cS_q$ be the permutations of $\{ 1,\dots,q \}$. Then we find
	\begin{align*}
	\lVert g^p_1\otimes_s\cdots\otimes_s g^p_q \lVert^2&=\frac{1}{q!}\sum_{\sigma\in \cS_q} \langle g_{1}^p\otimes\cdots\otimes g_q^p,g_{\sigma(1)}^p\otimes\cdots\otimes g_{\sigma(q)}^p \rangle\\&=\frac{1}{q!}\sum_{\sigma\in \cS_q} \langle g_{1}^p,g_{\sigma(1)}^p\rangle\cdots\langle g_{q}^p,g_{\sigma(q)}^p\rangle\\& = \frac{1}{q!} \langle g_{1}^p,g_{1}^p\rangle\cdots\langle g_{q}^p,g_{q}^p\rangle=\frac{1}{q!}.
	\end{align*}

	(3): Define for all $p\in \mathbb{N}$ the set 
	\begin{equation*}
	C_p=\{g^p_1\otimes_s \cdots\otimes_s g^p_q \} \cup \bigcup_{\ell=1}^\infty \{ h_1\otimes_s\cdots\otimes_s h_\ell \otimes_s g^p_1\otimes_s \cdots\otimes_s g^p_q\mid h_i\in \cH_{u(p)}\cap \cD(\omega) \}
   \end{equation*}
	and let $r<p$. Then $\phi_r\in \textup{Span}(C_r)$ and $\phi_p\in \textup{Span}(C_p)$ by Lemma \ref{Lem:FundamentalTransform}, so we just need to see that every element in $C_p$ is orthogonal to every element in $C_r$. Let $\psi_1\in C_p$ and $\psi_2\in C_r$. Note $\psi_1$ has a factor $g_1^p$ and that this factor is orthogonal to $g_i^r$ for all $i\in \{1,\dots,q\}$ by construction. Furthermore, for any $h\in \cH_{u(r)}$ we see that $h$ is supported in $B_{u(r)}^c\subset B_{u(p)}^c\subset (A_{u(p)}^1)^c$ so $g_1^p$ is orthogonal to all elements in $\cH_{u(r)}$ by construction. Hence $\psi_1$ contains a factor which is orthogonal to all factors in $\psi_2$ so they are orthogonal. This proves (3).

	(4): Using Lemma \ref{Lem:FundamentalTransform} we find 
	\begin{align*}
	\lVert (F_{1}-\cE_{(-1)^q}&-\lambda_1-\cdots-\lambda_q)\phi_p \lVert\\&=\sqrt{q!}\lVert U_p(F_{1}-\cE_{(-1)^q}-\lambda_1-\cdots-\lambda_q)U_p^*\nu_p\otimes g^p_1\otimes_s\cdots\otimes_s g_q^p \lVert\\&\leq \sqrt{q!}\lVert (F_{(-1)^q}(\alpha,f,\omega_1)-\cE_{(-1)^q})\nu_p\otimes g^p_1\otimes_s\cdots\otimes_s g_q^p\lVert \\&+\sqrt{q!}\sum_{i=1}^{q}\lVert \nu_p\otimes g_1^p\otimes_s\cdots\otimes_s (\omega_2 g^p_i-\lambda_ig^p_i) \otimes_s\cdots\otimes_s g_q^p   \lVert
	\\&\leq \lVert (F_{(-1)^q}-\cE_{(-1)^q})\nu_p\lVert +\sqrt{q!}\sum_{i=1}^{q} \lVert (\omega-\lambda_i)g^p_i   \lVert\\&\leq \frac{1}{p}+\sqrt{q!}\sum_{i=1}^{q}\frac{1}{u(p)}
	\end{align*}
	which converges to 0. This finishes the case where $f \in \cH_K^{2n}$ for some $K$. To prove the general case let $f^k=1_{B_k^c}f$ and note that $\cE_{(-1)^q\eta}(\alpha,f^k,\omega)+\lambda_1+\cdots+\lambda_q\in \sigma_{\textup{ess}}(F_{\eta}(\alpha,f^k,\omega))$ for all $k\in \NN$. Applying Lemma \ref{Lem:SimpleConvApplication} finishes the proof.
\end{proof}
\begin{lem}\label{Lem:OutsideEssentialSPectrum} Define $\widetilde{m}=\min\{m_{\textup{ess}}+\cE_{-1},\cE_1+m_{\textup{ess}}+m \}$. 
	Then $(-\infty,\widetilde{m})\cap \sigma_{\textup{ess}}(F_1)=\emptyset$. 
\end{lem}
\begin{proof}
If $m=0$ then $m_{\textup{ess}}=0$ by injectivity of $\omega$ so the statement is trivial since $(-\infty,\widetilde{m})\cap \sigma(F_1)=\emptyset$. Hence we may assume $m>0$ so $\omega \geq m>0$ almost everywhere. If $m_{\textup{ess}}=\infty$ the conclusion will follow from Proposition \ref{Lem:BasicPropertiesSBmodel}. Hence we may assume $m_{\textup{ess}}<\infty$. Define 
\begin{align*}
\omega_k(x)=\sum_{j=0}^{\infty}2^{-k}(j+1)1_{\{ \omega\in (j2^{-k},(j+1)2^{-k}] \}}(x)
\end{align*}
From Lemma \ref{Lem:SpecialSequenceOfDispersions} one obtains $\omega/\omega_k$ and $\omega_k/\omega$ converges to 1 in $L^\infty(\cM,\cF,\mu)$. This implies $(\alpha,f,\omega_k)$ satisfies Hypothesis 1, 3 and 4. Furthermore,
\begin{align}
\lim_{k\rightarrow \infty }f_1=f_1 && \lim_{k\rightarrow \infty }\omega_k^{-\frac{1}{2}}f_1=\omega^{-\frac{1}{2}}f_1 \nonumber \\
\lim_{k\rightarrow \infty }f_j=f_j && \lim_{k\rightarrow \infty }\omega_k^{\pm \frac{1}{2}}f_j=\omega^{\pm \frac{1}{2}}f_j \nonumber
\end{align}
in $\cH$ for $j\in \{2, \dots , 2n \}$. So defining
\begin{align*}
F_{\pm 1,k}=F_\eta(\alpha,f,\omega_k)
\end{align*}
we see that $F_{\pm 1,k}$ converges to $F_{\pm 1}$ in norm resolvent sense as $k$ tends to infinity by Lemma \ref{Lem:ResolevntKonv}. Applying the bounds in Lemmas \ref{Lem:FundamentalIneq} and \ref{Lem:FundamentalLowerbound} along with the bound $\eta\Gamma(-1)\geq -\lvert \eta\lvert$ we see 
\begin{equation*}
F_{\pm 1,k}\geq -\lvert \eta\lvert- \lVert \omega_k^{-\frac{1}{2}}f_1 \lVert^2+C,
\end{equation*}
which is uniformly bounded below in $k$. Hence $\cE_{\pm 1,k}=\inf(\sigma(F_{\pm 1,k}))$ converges to $\cE_{\pm 1}$. Defining the masses $ m_k:=m(\omega_k)$ and $m_{\textup{ess},k}:=m_{\textup{ess}}(\omega_k)$ we see $\{m_k\}_{k=1}^\infty$ converges to $m$ and $\{m_{\textup{ess},k}\}_{k=1}^\infty$ converges to $m_{\textup{ess}}$ by Lemma \ref{Lem:SpecialSequenceOfDispersions}. Defining
\begin{equation*}
\widetilde{m}_k=\min\{m_{\textup{ess},k}+\cE_{-1,k},\cE_{1,k}+m_{\textup{ess},k}+m_k \}.
\end{equation*}
we see that $\{\widetilde{m}_k\}_{k=1}^\infty$ converges to $\widetilde{m}$. Assume we have proven that $h(F_{ 1,k})$ is compact for any $h\in C^{\infty}_c((-\infty,\widetilde{m}_k))$. Let $h\in C^{\infty}_c((-\infty,\widetilde{m}))$ and note that $h\in C^{\infty}_c((-\infty,\widetilde{m}_k))$ for $k$ large enough. Norm resolvent convergence implies that
\begin{equation*}
h(F_{ 1})=\lim_{k\rightarrow \infty}h(F_{1,k})
\end{equation*}
in norm so $h(F_{ 1})$ is compact. This would finish the proof, so it only remains to prove that $h(F_{ 1,k})$ is compact for any $h\in C^{\infty}_c((-\infty,\widetilde{m}_k))$.

Define the projection $P_{k,j}=1_{\{ \omega\in (j2^{-k},(j+1)2^{-k}] \}}$ for $j\in \NN_0$ and $k\in \NN$. Then $P_{k,j_1}P_{k,j_2}=0$ for $j_1\neq j_2$ so we may define
\begin{align*}
\cH_{k,j}&=\{ P_{k,j}f_i \mid i\in \{1,\dots, 2n \} \}\\
\cH_{k}&=\bigoplus_{j=0}^\infty \cH_{k,j}= \left \{  \sum_{j=0}^{\infty} g_j\biggl \lvert \,\, g_j\in \cH_{k,j} \,\,\, \text{and}\,\,\,  \sum_{j=0}^{\infty}\lVert g_j\lVert^2<\infty      \right  \} \subset \cH
\end{align*}
For $i\in \{1, \dots, 2n \}$ we note that
\begin{equation*}
f_i=\sum_{j=0}^{\infty} P_{k,j}f_i\in \cH_{k}
\end{equation*}
so $f\in \cH_k^{2n}$. Fix $\psi\in \cH_k$ and write
\begin{equation*}
\psi=\sum_{j=0}^{\infty} g_j
\end{equation*}
with $g_j\in \cH_{k,j}$. By definition we see $g_j\in \cD(\omega_k)$ and $\omega_kg_j=(j+1)2^{-k}g_j$ for all $j\in \NN_0$ so
\begin{equation*}
e^{it\omega_k}\psi=\sum_{j=0}^{\infty} e^{it(j+1)2^{-k}}g_j\in \cH_k.
\end{equation*}
This implies $\cH_k$ reduces $\omega_k$ by \cite[Theorem Theorem 7.39]{Weidmann}. Let $\omega_{k,1}$ denote the restriction of $\omega_k$ to $\cH_k$. We claim $\omega_{k,1}$ has compact resolvents. Let $Q_{k,j}:\cH_k\rightarrow \cH_{k,j}$ be the projection. Then $Q_{k,j}$ has finite dimensional range and $Q_{k,j_1}Q_{k,j_2}=0$ if $j_1\neq j_2$. Let $\xi\in \CC\backslash \RR$ and $K,q_1, q_2\in \NN$ with $K\leq q_1< q_2$. Then
\begin{align*}
\left \lVert \sum_{j=q_1}^{q_2}((j+1)2^{-k}+\xi)^{-1}Q_{k,j} \right \lVert&=\max_{j\in \{q_1,q_1+1,\dots , q_2\} }\lvert ((j+1)2^{-k}+\xi)^{-1}\lvert \\&\leq \sup_{j\in \NN,\, j\geq K  }\lvert ((j+1)2^{-k}+\xi)^{-1}\lvert
\end{align*}
which goes to 0 as $K$ tends to $\infty$. This shows
\begin{align*}
A=\lim_{K\rightarrow \infty}\sum_{j=0}^{K}((j+1)2^{-k}+\xi)^{-1}Q_{k,j}=\sum_{j=0}^{\infty}((j+1)2^{-k}+\xi)^{-1}Q_{k,j}.
\end{align*}
exists since the partial sums are Cauchy. Note $A$ is compact since the $Q_{k,j}$ are compact and the sum converges in norm. We claim $A$ is the inverse of $\omega_{k,1}+\xi$. By selfadjointness of $\omega_{k,1}$ we know $\omega_{k,1}+\xi$ is bijective from $\cD(\omega_{k,1})\cap \cH_k$ into $\cH_k$ so it is enough to see $A$ maps $\cH_k$ into $\cD(\omega_{k,1})\cap \cH_k$ and $(\omega_{k,1}+\xi)A=1$. Let $\psi\in \cH_k$ and write
\begin{align*}
\psi&=\sum_{j=0}^{\infty} g_j \,\,\,\, \Rightarrow \,\,\,\, A\psi=\sum_{j=0}^{\infty} ((j+1)2^{-k}+\xi)^{-1}g_j\\
\psi_K&=\sum_{j=0}^{K} g_j \,\,\,\, \Rightarrow \,\,\,\, A\psi_K=\sum_{j=0}^{K} ((j+1)2^{-k}+\xi)^{-1}g_j
\end{align*}
where $K\in \NN$. Note $A\psi_K\in \cD(\omega_{k,1})=\cH_k\cap \cD(\omega)$ and $(\omega_{k,1}+\xi)A\psi_K=\psi_K$. Taking $K$ to infinity we see $A\psi\in \cD(\omega_{k,1})$ and $(\omega_{k,1}+\xi)A\psi=\psi$.

Let $h\in C^{\infty}_c(  (-\infty,\widetilde{m}_k) )$. As previously noted we have $f \in \cH_k^{2n}$ and $\omega_k$ is reduced by $\cH_k$. Define $\omega_{k,2}=\omega_k\mid_{\cH_k^\perp}$. Applying Lemma \ref{Lem:GeneralTransform} we obtain a unitary map
\begin{equation*}
U_k:\cF_b(\cH)\rightarrow \cF_b(\cH_k)\oplus \bigoplus_{j=1}^{\infty}\left( \cF_b(\cH_k)\otimes (\cH_k^\perp)^{\otimes_s j}\right),
\end{equation*}
such that
\begin{equation*}
U_kF_{\pm 1,k}U_k^*=F_{\pm \eta}(\alpha,f,\omega_{k,1})\oplus \bigoplus_{j=1}^\infty \left (F_{\pm(-1)^j\eta}(\alpha,f,\omega_{k,1})\otimes 1+1\otimes d\Gamma^{(j)}(\omega_{k,2})\right ).
\end{equation*}
Thus we find $\cE_{\pm \eta}(\alpha,f,\omega_{k,1})\geq \cE_{\pm 1,k}$ and
\begin{equation*}
U_kh(F_{ 1,k})U_k^*=h(F_{\eta }(\alpha,f,\omega_{k,1}))\oplus \bigoplus_{j=1}^\infty h\left (F_{(-1)^j\eta}(\alpha,f,\omega_{k,1})\otimes 1+1\otimes d\Gamma^{(j)}(\omega_{k,2})\right).
\end{equation*}
$h(F_{\eta}(\alpha,f,\omega_{k,1}))$ is compact by Proposition \ref{Lem:BasicPropertiesSBmodel} since $\omega_{k,1}$ has compact resolvent. Write $C_j=F_{(-1)^j\eta}(\alpha,f,\omega_{k,1})\otimes 1+1\otimes d\Gamma^{(j)}(\omega_{k,2})$ for $j\in \NN$. Using Theorem \ref{Thm:SpectralPropTensor}, Proposition \ref{Lem:BasicPropertiesSBmodel} and Lemma \ref{Lem:SecondQuantisedProp} we find for $j\in \NN$ that
\begin{align*}
\inf(\sigma_{\textup{ess}}(C_j))&\geq \cE_{ (-1)^j\eta}(\alpha,f,\omega_{k,1})+(j-1)\inf(\sigma(\omega_{k,2}))+\inf(\sigma_{\textup{ess}}(\omega_{k,2}))\\&\geq \cE_{ (-1)^j,k} + (j-1)m_k+m_{\textup{ess},k}\geq \widetilde{m}_k\\
\inf(\sigma(C_j))&\geq \cE_{ (-1)^j\eta}(\alpha,f,\omega_{k,1})+j\inf(\sigma(\omega_{k,2}))\geq \cE_{ (-1)^j,k} + jm_k.
\end{align*}
This implies $h(C_j)$ is compact for all $j\in \NN$ and since $ m_k>0$ we find $h(C_j)=0$ for $j$ large enough. Hence $U_kh(F_{ 1,k})U_k^*$ is a direct sum of compact operators where only finitely many are nonzero. This shows $U_kh(F_{ 1,k})U_k^*$ is compact as desired.
\end{proof}

\noindent Combining Lemmas \ref{Lem:PointsInTheEssentialSpectrum} and \ref{Lem:OutsideEssentialSPectrum} with Proposition \ref{Thm:Spectral Theory of decomposition} proves the first part of Theorem \ref{HVZ}. Statements (1) and (2) will follow from the corollaries below.
\begin{cor}\label{Cor:HVZ1}
	Assume $m=m_{\textup{ess}}$, $[m,3m]\subset \sigma_{\textup{ess}}(\omega)$ and $m$ is not isolated in $\sigma_{\textup{ess}}(\omega)$. Then $\sigma_{\textup{ess}}(F_{1})=[ \cE_{-1}+m,\infty )$.
\end{cor}
\begin{proof}
If $m=m_{\textup{ess}}$ then $\cE_{1}\leq \cE_{-1}+m$ by Lemma \ref{Lem:PointsInTheEssentialSpectrum}. Hence the minimum in Lemma \ref{Lem:OutsideEssentialSPectrum} is $\cE_{-1}+m_{\textup{ess}}=\cE_{-1}+m$. If $m\neq 0$ then the result follows directly from Lemma \ref{Lem:PointsInTheEssentialSpectrum}, so we may assume $m=0$. Let $\varepsilon>0$ and $x\in [ \cE_{-1}+m,\infty )$. Since $0$ is not isolated in $\sigma_{\textup{ess}}(\omega)$ we find $\lambda\in \sigma_{\textup{ess}}(\omega)$ with $ \lambda\leq \varepsilon$. Note that $x-\cE_{-1}\geq 0$ so we may pick $q\in \NN_0$ such that
\begin{equation*}
\lvert x-\cE_{(-1)^{2q+1}}-(2q+1)\lambda\lvert\leq \varepsilon.
\end{equation*}
$\cE_{(-1)^{2q+1}}+(2q+1)\lambda \in \sigma_{\textup{ess}}(F_{1})$ by Lemma \ref{Lem:PointsInTheEssentialSpectrum} so $x\in \overline{\sigma_{\textup{ess}}(F_{1})}=\sigma_{\textup{ess}}(F_{1})$.
\end{proof}

\begin{cor}
Assume $[m_{\textup{ess}},2m_{\textup{ess}}]\subset \sigma_{\textup{ess}}(\omega)$ and $m_{\textup{ess}}$ is not isolated in $\sigma_{\textup{ess}}(\omega)$. Then $\sigma_{\textup{ess}}(H_{\eta}(\alpha,f,\omega))=[m_{\textup{ess}}+E_\eta(\alpha,f,\omega),\infty)$.
\end{cor}
\begin{proof}
Combining Proposition \ref{Thm:Spectral Theory of decomposition} and Lemma \ref{Lem:PointsInTheEssentialSpectrum} we see
\begin{equation*}
\{ E_\eta(\alpha,f,\omega)+\lambda_1+\cdots+\lambda_q\mid \lambda_i\in \sigma_{\textup{ess}}(\omega)  \}\subset \sigma_{\textup{ess}}(H_\eta(\alpha,f,\omega))
\end{equation*}	
for all $q\in \mathbb{N}$. The proof is now the same as for Corollary \ref{Cor:HVZ1}.
\end{proof}

\section{Proof of Theorem \ref{unique}}
In this section we prove Theorem \ref{unique}. Let $\alpha\in \RR^{2n}$, $f\in \cH^{2n}$ and $\omega$ be a selfadjoint operator on $\cH$. We assume $(\alpha,f,\omega)$ satisfies Hypothesis \ref{Hyp1}, \ref{Hyp2} and \ref{Hyp3}. For $\eta\in \RR$ we will throughout this section use the notation $F_{\eta}:=F_{ \eta}(\alpha,f,\omega)$, $\cE_{\eta}:=\cE_{ \eta}(\alpha,f,\omega)$, $H_\eta:=H_{ \eta}(\alpha,f,\omega)$, $m_{\textup{ess}}=m_{\textup{ess}}(\omega)$, $m=m(\omega)$ and $E_\eta:=E_{\eta}(\alpha,f,\omega)$. Let $\cH_{\RR}$ be the real Hilbert space from Lemma \ref{Lem:ExsistenceOfNiceRealSubspace} corresponding to $\{ f_i \}_{i=1}^{2n}$ and $L^2(Q,\cG,\PP)$ be the corresponding $Q$-space. 
\begin{lem}\label{ffs}
Define the unitary matrix
\begin{equation*}
A=\frac{1}{\sqrt{2}}\begin{pmatrix} 1& -1 \\ 1& 1
\end{pmatrix}.
\end{equation*}
Let $V$ be the Q-space isomorphism and define $U=A\otimes V$.  Then $U$ is a unitary map from $\CC^2\otimes \cF_b(\cH)$ to
\begin{equation*}
\CC^2 \otimes L^2(Q,\cG,\PP)=L^2(\{ \pm 1 \}\times Q , \cB( \{ \pm 1 \} )\otimes \cG,\tau\otimes \mathbb{P}):=L^2(X,\cX,\nu),
\end{equation*}
where $\tau$ is the counting measure. Here we use the tensor product
\begin{equation*}
((v_1,v_{-1})\otimes f)(a,x)=\delta_{1,a}v_1f(x)+\delta_{-1,a}v_{-1}f(x),
\end{equation*}
where $\delta_{i,j}$ is the Kronecker delta. For $v\in \cH_{\RR}$ we have
	\begin{align}
	U \sigma_x\otimes \varphi(v) U^*&=\Phi(v)\label{eq:211} \\
	U\sigma_z \otimes 1 U&=-\sigma_x\otimes 1 \label{eq:212} \\
	U1 \otimes d\Gamma(\omega) U^*&=1\otimes Vd\Gamma(\omega)V^*\label{eq:213},
	\end{align}
	where $\Phi(v)$ is a multiplication operator. Furthermore, $UH_\eta U^*$ generates a positivity improving semigroup if $\eta <0$.
\end{lem}
\begin{proof}
Recall that $V\varphi(v)V^*:=\widetilde{\varphi}(v)$ is a multiplication operator for all $v\in \cH_\RR$. We now prove equations (\ref{eq:211}), (\ref{eq:212}) and (\ref{eq:213}). Equation (\ref{eq:213}) is trivial. To prove the other two one calculates
\begin{equation*}
A\sigma_zA^*=\sigma_x \,\,\,\,\,\text{and}\,\,\,\,\, A\sigma_xA^*=-\sigma_z
\end{equation*}
so $U\sigma_x\otimes \varphi(v)U=-\sigma_z\otimes \widetilde{\varphi}(v)$ and $U\sigma_z\otimes \varphi(v)U=\sigma_x\otimes 1$. Now $-\sigma_z\otimes \widetilde{\varphi}(v)$ obviously acts like multiplication by the map $(\Phi(v))(a,x)=-a(\widetilde{\varphi}(v))(x)$ so we are done proving equations (\ref{eq:211}) and (\ref{eq:212}).

Any element $\psi\in L^2(X,\cX,\nu)$ is of the form
\begin{equation*}
\psi=e_1\otimes \psi_1+e_{-1}\otimes \psi_{-1}.
\end{equation*}
Hence $\psi$ is (strictly) positive if and only if $\psi_1$ and $\psi_{-1}$ are (strictly) positive. Using Theorem \ref{Thm:PropertiesOfQspace} we find that $1\otimes \exp(-tVd\Gamma(\omega)V^*)$ is positivity preserving for all $t\geq 0$. Furthermore, the map $\sigma_x\otimes 1$ is positivity preserving since it maps $e_1\otimes \psi_1+e_{-1}\otimes \psi_{-1}$ to $e_{-1}\otimes \psi_1+e_{1}\otimes \psi_{-1}$ and so $\exp(t1\otimes \sigma_x)$ is positivity preserving for all $t\geq 0$. It follows that
\begin{align*}
\exp(-tUH_\eta(0,0,\omega)U^*)=\exp(-t\eta 1\otimes \sigma_x) (1\otimes \exp(-tVd\Gamma(\omega)V^*))
\end{align*}
is positivity preserving for all $\eta <0$ and $t\geq 0$. We will need the following observation:
\begin{equation}\label{eq:positivGS}
U(e_1\otimes \Omega)=Ae_1\otimes V\Omega=\frac{1}{\sqrt{2}}(e_1\otimes 1+e_{-1}\otimes 1 )=\frac{1}{\sqrt{2}}.
\end{equation}
Let $\eta <0$ and note that $e_1\otimes \Omega$ spans the ground state eigenspace of  $H_\eta(0,0,\omega)=\eta\sigma_z\otimes 1+1\otimes d\Gamma(\omega)$ by Theorem \ref{Thm:SpectralPropTensor}, so equation (\ref{eq:positivGS}) shows that $\frac{1}{\sqrt{2}}$ spans the ground state eigenspace of $UH_\eta(0,0,\omega)U^*$. Hence $UH_\eta(0,0,\omega)U^*$ generates a positivity preserving semi group and the ground state eigenspace is spanned by a strictly positive vector. This implies that the semi group generated by $UH_\eta(0,0,\omega)U^*$ is ergodic by \cite[Theorem XIII.43]{RS4}. Write
\begin{equation*}
UH_\eta U^*=UH_\eta (0,0,\omega)U^*+\sum_{j=1}^{2n}\alpha_j \Phi(f_j)^j:=UH_\eta (0,0,\omega)U^*+B
\end{equation*}
and define
\begin{equation*}
B_k=\sum_{j=1}^{2n}\alpha_j \Phi(f_j)^j1_{\{\lvert \Phi(f_j) \lvert\leq k\}},
\end{equation*}
which is a bounded multiplication operator. Assume now that we have proven the following statements
\begin{enumerate}
\item If $u,v\geq 0$ and $\langle u,\exp(-tB_k)v \rangle=0$ then $\langle u,v \rangle=0$

\item $UH_\eta (0,0,\omega)U^*+B_k $ and $UH_\eta U^*-B_k$ are uniformly bounded below in $k$ and converge in strong resolvent sense to $UH_\eta U^*$ and $UH_\eta (0,0,\omega)U^*$ respectively. 
\end{enumerate}
Then we may appeal to the proof of \cite[Theorem 3]{Farris} to see that $UH_\eta (\alpha,f,\omega)U^*$ generates an ergodic semigroup, which by \cite[Theorem XIII.44]{RS4} will be positivity improving.

Statement (1) is trivial since $B_k$ is a multiplication operator and $\lvert B_k\lvert <\infty$ almost everywhere. To prove statement (2) we let $\psi\in \textup{Span}(\widetilde{\cJ}(\cD(\omega))) $ and note that $U\psi\in \cD(\Phi(f_j)^j)$ for all $j\in \{ 1,\dots, 2n  \}$ by Proposition \ref{Lem:BasicPropertiesSBmodel}, so dominated convergence implies
\begin{equation*}
\lim_{k\rightarrow \infty }B_kU\psi = BU\psi.
\end{equation*}
Since $U\widetilde{\cJ}(\cD(\omega))$ spans a core for $UH_\eta (0,0,\omega)U^*$ and $UH_\eta U^*$ by Proposition \ref{Lem:BasicPropertiesSBmodel} we find that $UH_\eta(0,0,\omega)U^*+B_k$ and $UH_\eta U^*-B_k$ converges to $UH_\eta U^*$ and $UH_\eta(0,0,\omega)U^*$ respectively in strong resolvent sense (see \cite[Theorem VIII.25]{RS1}).

It remains only to find a uniform lower bound. We calculate
\begin{align*}
UH_\eta U^*-B_k=&-\eta \sigma_x\otimes 1+1\otimes Vd\Gamma(\omega)V^*+\alpha_1 \Phi(f_1)1_{\{\lvert \Phi(f_1) \lvert> k\}}\\&+\sum_{j=2}^{2n}\alpha_j \Phi(f_j)^j1_{\{\lvert \Phi(f_j) \lvert> k\}}.\\
UH_\eta (0,0,\omega)U^*+B_k=&-\eta \sigma_x\otimes 1+1\otimes Vd\Gamma(\omega)V^*+\alpha_1 \Phi(f_1)1_{\{\lvert \Phi(f_1) \lvert\leq k\}}\\&+\sum_{j=2}^{2n}\alpha_j \Phi(f_j)^j1_{\{\lvert \Phi(f_j) \lvert\leq k\}}.
\end{align*}
In both expressions, the first term on the right hand side is bounded below by $-\lvert \eta\lvert $ and the sum is bounded below uniformly in $k$ by Lemma \ref{Lem:FundamentalLowerbound}. Now $\alpha_1\Phi(f_1)$ is infinitesimally $1\otimes Vd\Gamma(\omega)V^*$ bounded by Lemmas \ref{Lem:FundamentalIneq} and \ref{Lem:SmallObsTensor}. Hence there are $a\in (0,1)$ and $b\geq 1$ such that
\begin{equation*}
\lVert\alpha_1 \Phi(f_1) \psi \lVert \leq a\lVert (1\otimes Vd\Gamma(\omega)V^*)\psi \lVert+b\lVert \psi\lVert
\end{equation*}
 for all $\psi \in \cD(1\otimes Vd\Gamma(\omega)V^*)$. $1_{C }\alpha_1\Phi(f_1)$ will satisfies the same inequality for any $C\in \cX$ so \cite[Theorem 9.1]{Weidmann} provides a lower bound of $1_{C}\alpha_1\Phi(f_1)+1\otimes Vd\Gamma(\omega)V^*$ independent of $C$. This finishes the proof.
\end{proof}

\begin{lem}\label{Lem:UniquenessSB}
If $\eta\neq 0$ and $E_\eta$ is an eigenvalue of $H_\eta$ then $E_\eta$ is non degenerate. If $\psi$ is any ground state of $H_\eta$ then $U\psi=e_{-\textup{sign}(\eta)}\otimes \psi$ where $\psi$ is an eigenvector of $F_{-\lvert \eta \lvert}$ corresponding to the energy $E_\eta$. In particular, we can conclude that $E_\eta$ is not an eigenvalue of $F_{\lvert \eta \lvert}$.
\end{lem}
\begin{proof}
Let $\eta\neq 0$ and assume $E_\eta$ is an eigenvalue of $H_\eta$ with corresponding eigenvector $\psi$.	If $\eta <0$ then non degeneracy $E_\eta$ follows from Lemma \ref{ffs} and \cite[Theorem XIII.43]{RS4}. Furthermore, $\langle \psi, e_{1}\otimes \Omega\rangle\neq 0 $ since $e_{1}\otimes \Omega$ is mapped the positive element $\frac{1}{\sqrt{2}}$ under the map from Lemma \ref{ffs} (see equation (\ref{eq:positivGS})). If $\eta>0$ then $\sigma_x\otimes 1$ transforms $H_\eta$ into $H_{-\eta}$ showing that non degeneracy of $E_\eta$ holds in this case as well, but now $\psi$ will have nonzero inner product with $\sigma_x\otimes 1(e_1\otimes \Omega)=e_{-1}\otimes \Omega$. So all in all we can conclude that $E_{\eta}$ is a non degenerate eigenvalue of $H_\eta$ and $0 \neq \langle \psi, e_{-\text{sign}(\eta)}\otimes \Omega \rangle$.

Let $U$ be the unitary map from Proposition \ref{Thm:Spectral Theory of decomposition}. Then $U\psi=(\psi_{1},\psi_{-1})=e_1\otimes \psi_{1}+e_{-1}\otimes \psi_{-1}$ is a ground state of $F_{\eta}\oplus F_{-\eta}$ corresponding to the eigenvalue $E_\eta$. We see 
\begin{align*}
0 &\neq \langle \psi, e_{-\text{sign}(\eta)}\otimes \Omega \rangle= \langle (\psi_1,\psi_{-1}), Ue_{-\text{sign}(\eta)}\otimes \Omega \rangle\\&=\langle (\psi_1,\psi_{-1}), e_{-\text{sign}(\eta)}\otimes \Omega \rangle=\langle \psi_{-\text{sign}(\eta)},\Omega \rangle.
\end{align*}
This implies $\psi_{-\text{sign}(\eta)}\neq 0$. Hence $\psi_{-\text{sign}(\eta)}$ is an eigenvector of $F_{-\lvert \eta\lvert }$ corresponding to the eigenvalue $E_\eta$. If $E_\eta$ was an eigenvalue of $F_{\lvert \eta\lvert}$ then it would be an eigenvalue of $H_\eta$ and $F_{-\lvert \eta \lvert}$ as well. In particular, $E_\eta$ would be an eigenvalue of $H_\eta$ and the multiplicity would be $2$ or more which is a contradiction. 
\end{proof}
\begin{lem}\label{Uniquee=0}
	If $\cE_{-\lvert \eta\lvert}$ is an eigenvalue of $F_{-\lvert \eta\lvert}$ then $\cE_{-\lvert \eta\lvert}$ is non degenerate and every eigenvector will have nonzero inner product with $\Omega$. 
\end{lem}
\begin{proof}
We start with the case $\eta=0$. Let $V$ be the Q-space isomorphism from Theorem \ref{Thm:PropertiesOfQspace}. From Theorem \ref{Thm:PropertiesOfQspace} we know that $Vd\Gamma(\omega)V^*$ generates a positivity improving semigroup and $V\Omega=1$. We now prove that the semigroup of $VF_0V^*$ is positivity improving. Note
\begin{equation*}
VF_0V^*=VF_0(0,0,\omega)V^*+\sum_{j=1}^{2n}\alpha_j \widetilde{\varphi}(v_j)^j:=VH_\eta(0,0,\omega)V^*+A,
\end{equation*}
and define
\begin{equation*}
A_k=\sum_{j=1}^{2n}\alpha_j \widetilde{\varphi}(v_j)^j1_{\{\lvert \widetilde{\varphi}(v_j) \lvert\leq k\}},
\end{equation*}
which is now a bounded multiplication operator. With the exact same proof as in Lemma \ref{ffs} we check
\begin{enumerate}
\item If $u,v\geq 0$ and $\langle u,\exp(-tA_k)v \rangle=0$ then $\langle u,v \rangle=0$.

\item  $VF_0(0,0,\omega)V^*+A_k$ and $VF_0V^*-A_k$ are uniformly bounded below in $k$ and converge in strong resolvent sense to $VF_0V^*$ and $VF_0(0,0,\omega)V^*$ respectively.
\end{enumerate}
An appeal to the proof of \cite[Theorem 3]{Farris} along with \cite[Theorem XIII.43]{RS1} finishes the proof when $\eta=0$. For $\eta\neq 0$ one may combine Theorem \ref{Thm:PropertiesOfQspace} part (1) with \cite[Theorem 2]{Farris} to obtain the conclusion. 
\end{proof}

\noindent We can now prove some spectral properties of the fiber operators. In the remaining part of this section we will also assume Hypothesis \ref{Hyp4} is satisfied so we may use Theorem \ref{HVZ} except for part (3). However part (3) of Theorem \ref{HVZ} is proven in the next lemma
\begin{lem}\label{Fiberenergy}
$\cE_{-\lvert \eta\lvert}=E_\eta$ and $\cE_{-\lvert \eta\lvert}\leq \cE_{\lvert \eta\lvert}$. Furthermore, $\cE_{-\lvert \eta\lvert}< \cE_{\lvert \eta\lvert}$ if and only if $m> 0$ and $\eta\neq 0$. In particular, if $\eta\neq 0$ and $m=0$ then $F_{\lvert \eta\lvert}$ has no ground state.
\end{lem}
\begin{proof}
If $\eta=0$ then $\cE_{ \lvert \eta\lvert}= \cE_{- \lvert \eta\lvert}$ is trivial. If $m=0$ then $m_{\textup{ess}}=0$ by injectivity of $\omega$. Using Theorem \ref{HVZ} we obtain $\cE_{\pm \lvert \eta\lvert}\leq \cE_{\mp \lvert \eta\lvert}$ since $\cE_{\mp \lvert \eta\lvert}\in \sigma_{\textup{ess}}(F_{\pm \lvert \eta\lvert})$ so $\cE_{ -\lvert \eta\lvert}= \cE_{ \lvert \eta\lvert}$. The statement regarding absence of ground states for $F_{\lvert \eta\lvert}$ now follows from $\cE_{-\lvert \eta\lvert}= \cE_{\lvert \eta\lvert}$ and Lemma $\ref{Lem:UniquenessSB}$.

 Assume $m> 0$ and $\eta\neq 0$. $m>0$ implies that $E_\eta$ is an eigenvalue of $H_{\eta}$ by Theorem \ref{HVZ} and so $E_\eta=\cE_{-\lvert \eta\lvert}$ is an eigenvalue of $F_{-\lvert \eta\lvert}$ by Lemma \ref{Lem:UniquenessSB}. Since $\cE_{-\lvert \eta\lvert}=E_\eta\leq \cE_{\lvert \eta\lvert}$ we just have to prove that $\cE_{-\lvert \eta\lvert}=\cE_{\lvert \eta\lvert}$ is impossible. Assume $\cE_{-\lvert \eta\lvert}=  \cE_{\lvert \eta\lvert}$. Then Theorem \ref{HVZ} implies that
	\begin{equation*}
	\inf (\sigma_{\textup{ess}}(F_{\lvert \eta\lvert}))=\cE_{-\lvert \eta\lvert}+m_{\textup{ess}}>\cE_{\lvert \eta\lvert},
	\end{equation*}
	and so $\cE_{-\lvert \eta\lvert}=\cE_{\lvert \eta\lvert}=E_\eta$ would be an eigenvalue of $F_{\lvert \eta\lvert}$. However, this gives a contradiction with Lemma \ref{Lem:UniquenessSB}.
\end{proof}
\noindent Regarding exited states we deduce the following
\begin{lem}\label{Negative fiber has groundstates for positive mass}
If $\eta\neq 0$ and $\cE_{\lvert \eta\lvert}$ is an eigenvalue of $F_{\lvert \eta\lvert}$ then it is an eigenvalue of $H_{\eta}$ contained in $(E,E+m_{\textup{ess}}]$. Furthermore, if $2\lvert \eta\lvert <m_{\textup{ess}}$ then $\cE_{\lvert \eta\lvert}$ is an eigenvalue of $F_{\lvert \eta\lvert}$.
\end{lem}
\begin{proof}
Assume $\cE_{\lvert \eta\lvert}$ is an eigenvalue of $F_{\lvert \eta\lvert}$. Then $m_{\textup{ess}}\geq m>0$ by Lemma \ref{Fiberenergy}. Using Theorem \ref{HVZ} and Lemma \ref{Fiberenergy} we calculate
\begin{equation*}
E_\eta=\cE_{-\lvert \eta\lvert}< \cE_{\lvert \eta\lvert}\leq \cE_{-\lvert \eta\lvert}+m_{\textup{ess}}=E_\eta+m_{\textup{ess}}.
\end{equation*}	
Assume now $2\lvert \eta\lvert <m_{\textup{ess}}$. By Theorem $\ref{HVZ}$ it is enough to prove the inequality $\cE_{\lvert \eta\lvert}<\cE_{-\lvert \eta\lvert}+m_{\textup{ess}}$. For any $\varepsilon>0$ we may pick normalised $\psi\in \cD(F_{\lvert \eta\lvert})=\cD(F_{-\lvert \eta\lvert})$ such that
\begin{equation*}
\varepsilon+\cE_{\lvert \eta\lvert}-\cE_{-\lvert \eta\lvert}\leq \langle \psi, (F_{\lvert \eta\lvert}-F_{-\lvert \eta\lvert})\psi \rangle=2\lvert \eta\lvert \langle \psi, \Gamma(-1)\psi \rangle\leq 2\lvert \eta\lvert.
\end{equation*} 
This proves the desired inequality. 
\end{proof}
\noindent Theorem \ref{unique} is now a combination of all lemmas in this section.

\section{Proof of Theorem \ref{Thm:Numberstructure_Massless} part \textup{(2)} }
In this chapter we prove the last half of Theorem \ref{Thm:Numberstructure_Massless}. A proof of the first half can be found in Appendix D. First we shall need the following lemma.

\begin{lem}\label{Lem:SupportProp}
Assume $\cH=L^2(\cM,\cF,\mu)$ with $(\cM,\cF,\mu)$ $\sigma$-finite. Let $\eta \leq 0$, $\alpha\in \RR^{2n}$, $f\in \cH^{2n}$ and $\omega$ be a selfadjoint multiplication operator on $\cH$. Assume $(\alpha,f,\omega)$ satisfies Hypothesis \ref{Hyp1}, \ref{Hyp2}, \ref{Hyp3} and \ref{Hyp4}. Let $A=\bigcup_{i\leq 2n}\{ f_i\neq 0 \}$, $\cH_1=L^2(X,\cF,1_A\mu)$, $\cH_2=L^2(X,\cF,1_{A^c}\mu)$, $\omega_i$ be multiplication with $\omega$ on the space $\cH_i$ and define $\widetilde{f}_i\in \cH_1$ by $\widetilde{f}_i=f_i$ $1_A\mu$-almost everywhere. Then $(\alpha,\widetilde{f},\omega_1)$ satisfies Hypothesis \ref{Hyp1}, \ref{Hyp2}, \ref{Hyp3} and  4. We also have
\begin{enumerate}
	\item[\textup{(1)}] $\cE_{\eta }(\alpha,f,\omega)=\cE_{\eta}(\alpha,\widetilde{f},\omega_1)$ and $\cE_{\eta}(\alpha,f,\omega)$ is an eigenvalue of $F_{\eta}(\alpha,f,\omega)$ if and only if $\cE_{\eta}(\alpha,f,\omega)$ is an eigenvalue of $F_{\eta}(\alpha,\widetilde{f},\omega_1)$. In particular, if $\inf_{k\in A}\omega(k)>0$ then $\cE_{\eta}(\alpha,f,\omega)$ is an eigenvalue of $F_{\eta}(\alpha,\widetilde{f},\omega_1)$ and $F_{\eta}(\alpha,f,\omega)$.
	
	\item[\textup{(2)}] If $\psi=(\psi^{(k)})$ is a ground state of $F_{\eta}(\alpha,f,\omega_1)$ then $\psi=(1_{A^k}\psi^{(k)})$ is a ground state of $F_{\eta}(\alpha,f,\omega)$.
\end{enumerate}
\end{lem}
\begin{proof}
Define $P_i:\cH\rightarrow \cH_i$ by $P_1(f)=f$ $(1_A\mu)$-almost everywhere and  $P_2(f)=f$ $1_{A^c}\mu$-almost everywhere. Let $V:\cH\rightarrow \cH_1\oplus \cH_2$ be $V(f)=(P_1(f),P_2(f))$. Then we see $V$ is unitary with $V^*(f,g)=1_Af+1_{A^c}g$ $\mu$-almost everywhere. Clearly we have $Vf_i=(\widetilde{f}_i,0)$ along with $V\omega V^*=(\omega_1,\omega_2)$. The properties in Hypothesis \ref{Hyp1}, \ref{Hyp2}, \ref{Hyp3} and \ref{Hyp4} are easily checked. Using Lemma \ref{Lem:GeneralTransform} we find a unitary map
\begin{equation*}
U:\cF_b(\cH)\rightarrow \cF_b(\cH_1)\oplus \bigoplus_{j=1}^\infty \left( \cF_b(\cH_1)\otimes \cH_2^{\otimes_s j}\right)
\end{equation*}
such that
\begin{equation*}
U F_{\eta}(\alpha,f,\omega)U^*=F_{\eta}(\alpha,\widetilde{f},\omega_1) \oplus \bigoplus_{j=1}^\infty \left(F_{(-1)^j\eta}(\alpha,\widetilde{f},\omega_1)\otimes 1+1\otimes d\Gamma^{(j)}(\omega_2)\right).  
\end{equation*}
Define $C_j=F_{(-1)^j\eta}(\alpha,\widetilde{f},\omega_1)\otimes 1+1\otimes d\Gamma^{(j)}(\omega_2)$ for $j\in \NN$. Using Theorem \ref{Thm:SpectralPropTensor} and Lemma \ref{Lem:SecondQuantisedProp} we find   $\inf(\sigma(C_j))=\cE_{(-1)^j\eta}(\alpha,\widetilde{f},\omega_1)+j\inf(\sigma(\omega_2))$. Theorem \ref{HVZ} implies $\cE_{\eta}(\alpha,\widetilde{f},\omega_1)\leq \cE_{-\eta}(\alpha,\widetilde{f},\omega_1)$ so
\begin{equation*}
\cE_{\eta}(\alpha,f,\omega)=\min\left \{  \cE_{\eta}(\alpha,\widetilde{f},\omega_1),   \inf_{j\in \NN} \sigma(C_j)   \right \}=\cE_{\eta}(\alpha,\widetilde{f},\omega_1).
\end{equation*}
Assume $\cE_{\eta}(\alpha,f,\omega)$ is an eigenvalue of $F_{\eta}(\alpha,\widetilde{f},\omega_1)$. Then $\cE_{\eta}(\alpha,f,\omega)$ is obviously an eigenvalue of $F_{\eta}(\alpha,f,\omega)$. Assume now $\cE_{\eta}(\alpha,f,\omega)$ is an eigenvalue of $F_{\eta}(\alpha,f,\omega)$ and that $\cE_{\eta}(\alpha,f,\omega)$ is not an eigenvalue of $F_{\eta}(\alpha,\widetilde{f},\omega_1)$. Then there is $j\in \mathbb{N}$ such that $\cE_{\eta}(\alpha,f,\omega)$ is an eigenvalue of $C_j$. This implies
\begin{align*}
\cE_{\eta}(\alpha,\widetilde{f},\omega_1)=\cE_{\eta}(\alpha,f,\omega) \geq \inf(\sigma(C_j))= \cE_{(-1)^j\eta}(\alpha,\widetilde{f},\omega_1)+j\inf(\sigma(\omega_2))
\end{align*}
Using $\cE_{\eta}(\alpha,\widetilde{f},\omega_1)\leq \cE_{(-1)^j\eta}(\alpha,\widetilde{f},\omega_1)$ we find $\inf(\sigma(\omega_2))=0$ and  $\cE_{(-1)^j\eta}(\alpha,\widetilde{f},\omega_1)=\cE_{\eta}(\alpha,f,\omega)$. Injectivity of $\omega_2$ and Lemma \ref{Lem:SecondQuantisedProp} shows 0 is not an eigenvalue of $d\Gamma^{(j)}(\omega_2)$. By Theorem \ref{Thm:SpectralPropTensor}  we find that $\cE_{\eta}(\alpha,f,\omega)=\cE_{(-1)^j\eta}(\alpha,\widetilde{f},\omega_1)+0$ is not an eigenvalue of $C_j=F_{(-1)^j\eta}(\alpha,f,\omega_1)\otimes 1+1\otimes d\Gamma^{(j)}(\omega_2)$  which is a contradiction. Hence $\cE_{\eta}(\alpha,f,\omega)$ is an eigenvalue of $F_{\eta}(\alpha,\widetilde{f},\omega_1)$. The last part of statement (1) follows from $m(\omega_1)>0$ and Theorem \ref{HVZ}.

To prove statement (2) we let $j:\cH_1\rightarrow \cH_1\oplus \cH_2$ be the embedding $j(f)=(f,0)$ and define $Q=V^*j$. Now $U^*\psi=\Gamma(Q)\psi$ by Lemma \ref{Thm:ISOTHM3} and $U^*\psi$ is the desired eigenvector of $F_{\eta}(\alpha,f,\omega)$. Noting $Q(f):=V^*j(f)=1_{A}f$ we see $U^*\psi=\Gamma(Q)\psi=(1_{A^k}\psi^{(k)})$ as desired.
\end{proof}

We now prove part \textup{(2)} of Theorem \ref{Thm:Numberstructure_Massless}. Let $\eta \in \RR$, $\alpha\in \RR^{2n}$, $f\in \cH^{2n}$ and $\omega$ be a selfadjoint operator on $\cH$. Assume $\cH=L^2(\RR^\nu,\cB(\RR^\nu),\lambda^{\otimes \nu})$, $\omega$ is a multiplication operator on $\cH$ and $(\alpha,f,\omega)$ satisfies Hypothesis \ref{Hyp1}, \ref{Hyp2}, \ref{Hyp3}, \ref{Hyp4} and \ref{Hyp5}.

Define $B_\ell=\{ \omega\geq \ell^{-1} \}$ and $f^\ell=1_{B_\ell}f$. Then $F_{\pm 1,\ell}:=F_{\pm \lvert \eta \lvert }(\alpha,f^\ell,\omega)$ converges in norm resolvent sense to $F_{\pm 1}:=F_{\pm \lvert \eta \lvert }(\alpha,f,\omega)$ by Lemma \ref{Lem:SimpleConvApplication} and $\cE_{\ell}=\cE_{ -\lvert \eta \lvert }(\alpha,f^\ell,\omega)$ converges to $\cE:=\cE_{- \lvert \eta \lvert }(\alpha,f,\omega)$. Furthermore, $F_{ -1,\ell}$ has a normalised ground state $\psi_\ell$ for all $\ell\in \NN$ by Lemma \ref{Lem:SupportProp}. After taking a subsequence we may assume that $\psi_\ell$ converges weakly to a vector $\psi$. The last half of Theorem $\ref{Thm:Numberstructure_Massless}$ will be proven by \cite[Lemma 4.9]{HirokawaArai} if we can prove that $\lVert \psi \lVert=1$. First a few observations which we will summarise in Lemma \ref{Lem:Basicstuff} below. In order to formulate Lemma \ref{Lem:Basicstuff} we need to use pointwise annihilation operators which are defined in the discussion after Lemma \ref{Lem:LiftToIntegralOperators}. 
\begin{lem}\label{Lem:Basicstuff}
The following holds:
\begin{enumerate}
	\item[\textup{(1)}] Let $A_1$ be the pointwise annihilation operator of order 1. We have
	\begin{equation*}
	(A_1\psi_\ell)(k)=-\sum_{j=1}^{2n}  f^{\ell}_j(k) (F_{1,\ell}-\cE_\ell+\omega(k))^{-1} j\alpha_j\varphi(f_j^{\ell})^{j-1}\psi_\ell .
	\end{equation*}
	\item[\textup{(2)}] There is a constant $C$ independent of $\ell$ and $j$ such that $\lVert \alpha_j\varphi(f_j^{\ell})^{j-1}\psi_\ell\lVert\leq C$
	
	\item[\textup{(3)}] $\psi_\ell\in \cD(N)$ and $\langle \psi_\ell,N\psi_\ell \rangle$ is uniformly bounded in $\ell$. Hence $A_1\psi_\ell\in L^2(\RR^\nu,\cB(\RR^\nu),\lambda^{\otimes \nu},\cF_b(\cH))$ for all $\ell\in \NN$ and the sequence $\{A_1\psi_\ell\}_{\ell=1}^\infty$ is bounded in this space.
	
	\item[\textup{(4)}] We have
	\begin{equation}\label{eq:diff}
	A_1\psi_\ell+\sum_{j=1}^{2n}  f_j(k) (F_{1}-\cE+\omega(k))^{-1} j\alpha_j\varphi(f_j^{\ell})^{j-1}\psi_\ell
	\end{equation}
	converges to 0 in $L^2(\RR^\nu,\cB(\RR^\nu),\lambda^{\otimes \nu},\cF_b(\cH))$. 
\end{enumerate}
\end{lem}
\begin{proof}
Statement (1) follows directly from Theorem \ref{Thm:ThepulllThrough}. To prove statement (2), we note that $j\alpha_j \varphi(f^\ell_j)^{j-1} (F_{-1,\ell}-i)^{-1}$ is uniformly bounded for $\ell\in \NN$ and $j\in \{1, \dots, 2n \}$ by Lemma \ref{Lem:UniformUpperBounds}. Let $\widetilde{C}$ be the bounding constant. Then
\begin{equation*}
\lVert j\alpha_j\varphi(f_j^{\ell})^{j-1}\psi_\ell\lVert \leq \widetilde{C} \lVert (\cE_\ell-i)\psi_\ell \lVert.
\end{equation*}
This proves statement (2) since $\{\cE_\ell\}_{\ell=1}^\infty$ is convergent and $\lVert \psi_\ell \lVert=1$ for all $\ell\in \NN$.

To prove statement (3) we note that $\psi_\ell\in \cD(N)$ by part (1) of Theorem \ref{Thm:Numberstructure_Massless}. Using (1), (2) and $\cE_{\ell}\leq \cE_{\lvert \eta\lvert}(\alpha,f^\ell,\omega)$ by Theorem \ref{HVZ} we estimate
\begin{equation}\label{eq:Dominator}
\lVert (A_1\psi_\ell)(k)\lVert^2\leq C^2\left(\sum_{j=1}^{2n} \frac{\lvert f^\ell_j(k)\lvert }{\omega(k)}\right)^2\leq 2nC^2\sum_{j=1}^{2n} \frac{\lvert f_j(k)\lvert^2 }{\omega(k)^2}.
\end{equation}
Integrating and appealing to Theorem \ref{Thm:CalculatingSecondQuantisedUsingAnihilation} yields the result.

To prove statement (4), note that $(F_{1,\ell}-\cE_\ell+\omega(k))^{-1}-(F_{1}-\cE+\omega(k))^{-1}$ converges to 0 in norm by Lemma \ref{Lem:ResolevntKonv}. Since $j\alpha_j\varphi(f_j^{\ell})^{j-1}\psi_\ell$ is uniformly bounded, we see that the function in equation (\ref{eq:diff}) converges to 0 pointwise. The conclusion now follows by dominated convergence and estimates similar to those in equation (\ref{eq:Dominator}).
\end{proof}

In the proof of the next lemma, we will write $B(\cF_b(\cH))$ for the set bounded linear maps from $\cF_b(\cH)$ into $\cF_b(\cH)$.
\begin{lem}
Let $G\in C_0^\infty(\RR^\nu)$ such that $G(0)=1$ and $0\leq G\leq 1$. Define $G_R=G(x/R)$ and let $A$ be either $x=-i\nabla_k$ or $k$. For any $\varepsilon>0$ there is $\ell'\in \NN$ and $R'>0$ such that $\lVert (1-\Gamma(G_R(A)))\psi_\ell\lVert\leq \varepsilon$ for all $R>R'$, $\ell>\ell'$ and $A\in \{ -i\nabla_k,k  \}$.
\end{lem}
\begin{proof}
Note that
\begin{equation*}
(1-\Gamma(G_R(A)))^2=1-\Gamma(G_R(A))+\Gamma(G_R(A))(\Gamma(G_R(A))-1).
\end{equation*}
On $j$ particle vectors we see that $\Gamma(G_R(A))(\Gamma(G_R(A))-1)$ acts like a negative multiplication operator in position/momentum space depending on the choice of $A$. Hence
\begin{equation*}
(1-\Gamma(G_R(A)))^2\leq 1-\Gamma(G_R(A)).
\end{equation*}
On $j$ particle vectors in position/momentum space (depending on $A$) we find that $1-\Gamma(G_R(A))$ acts like multiplication by
\begin{equation*}
1-G_R(k_1)G_R(k_2)\cdots G_R(k_j)=\sum_{i=1}^{j}(1-G_R(k_i))G_R(k_{i+1})\cdots G_R(k_j)\leq \sum_{i=1}^{j}(1-G_R(k_i)).
\end{equation*}
Hence $1-\Gamma(G_R(A))\leq d\Gamma(1-G_R(A))$ so it is enough to prove that
\begin{equation*}
\langle \psi_\ell,d\Gamma(1-G_R(A))\psi_\ell \rangle
\end{equation*}
goes to 0 for $R$ and $\ell$ tending to $\infty$. Note that $\psi_\ell\in \cD(N)\subset \cD(d\Gamma(1-G_R(A)))$ by Lemma \ref{Lem:Basicstuff} so $\langle \psi_\ell,d\Gamma(1-G_R(A))\psi_\ell \rangle$ is well defined. Using Theorem \ref{Thm:CalculatingSecondQuantisedUsingAnihilation} we see that
\begin{equation*}
\langle \psi_\ell,d\Gamma(1-G_R(A))\psi_\ell \rangle=\langle A_1\psi_\ell,((1-G_R(A))\otimes 1) A_1\psi_\ell \rangle.
\end{equation*}
Define
\begin{equation*}
q_\ell(k)= -\sum_{j=1}^{2n}  f_j(k) (F_{1}(f)-\cE+\omega(k))^{-1} j\alpha_j\varphi(f_j^{\ell})^{j-1}\psi_\ell.
\end{equation*}
By Lemma \ref{Lem:Basicstuff} we know $A_1\psi_\ell-q_\ell$ converges to 0 in $L^2(\RR^\nu,\cB(\RR^\nu),\lambda^{\otimes \nu},\cF_b(\cH))$ and $\{A_1\psi_\ell\}_{\ell=1}^\infty$ is bounded in $L^2(\RR^\nu,\cB(\RR^\nu),\lambda^{\otimes \nu},\cF_b(\cH))$. This implies $\{q_\ell\}_{\ell=1}^\infty$ is bounded in $L^2(\RR^\nu,\cB(\RR^\nu),\lambda^{\otimes \nu},\cF_b(\cH))$. Combing this observation with the fact that $\lVert (1-G_R(A))\otimes 1\lVert=1$ for all $R>0$ and $A_1\psi_\ell-q_\ell$ converges to 0 in $L^2(\RR^\nu,\cB(\RR^\nu),\lambda^{\otimes \nu},\cF_b(\cH))$ we find $\ell'\in \NN$ such that 
\begin{equation*}
\langle A_1\psi_\ell,(1-G_R(A))\otimes 1 A_1\psi_\ell \rangle\leq \frac{\varepsilon}{3}+\langle q_\ell,((1-G_R(A))\otimes 1) q_\ell \rangle
\end{equation*}
for all $R>0$ and $\ell>\ell'$. Write
\begin{equation*}
\widetilde{q}_{j}(t)=  -f_j(k) (F_{1}-\cE+\omega(k))^{-1} 
\end{equation*}
and note that $\widetilde{q}_j\in L^2(\RR^\nu,\cB(\RR^\nu),\lambda^{\otimes \nu},B(\cF_b(\cH)))$. Hence there is a sequence $\{ \widetilde{q}_{j,p}  \}_{p=1}^\infty$ of elements in $L^2(\RR^\nu,\cB(\RR^\nu),\lambda^{\otimes \nu},B(\cF_b(\cH)))$ such that $\{ \widetilde{q}_{j,p}  \}_{p=1}^\infty$ converges to $\widetilde{q}_{j}$ in $L^2(\RR^\nu,\cB(\RR^\nu),\lambda^{\otimes \nu},B(\cF_b(\cH)))$ and each $\widetilde{q}_{j,p}$ is a linear combination of functions of the form $k\mapsto g(k)C$ where $C\in B(\cF_b(\cH))$ and $g\in \cH$. Define
\begin{equation*}
q_{\ell,p}:=\sum_{j=1}^{2n}  \widetilde{q}_{j,p} j\alpha_j\varphi(f_j^{\ell})^{j-1}\psi_\ell.
\end{equation*}
Since $j\alpha_j \varphi(f^\ell_j)^{j-1}\psi_\ell$ is uniformly bounded in $\ell$ we see that 
\begin{align*}
\lim_{p\rightarrow \infty} \sup_{\ell\in \NN} \lVert q_\ell-q_{\ell,p}  \lVert=0.
\end{align*}
In particular, $\{q_{\ell,p}\}_{\ell,p=1}^\infty$ is bounded in $L^2(\RR^\nu,\cB(\RR^\nu),\lambda^{\otimes \nu},\cF_b(\cH))$ since $\{q_{\ell}\}_{\ell=1}^\infty$ is bounded in $L^2(\RR^\nu,\cB(\RR^\nu),\lambda^{\otimes \nu},\cF_b(\cH))$. Picking $p$ large enough we may thus estimate
\begin{equation*}
\langle A_1\psi_\ell,(1-G_R(A))\otimes 1 A_1\psi_\ell \rangle\leq \frac{2\varepsilon}{3}+\langle q_{\ell,p},(1-G_R(A))\otimes 1 q_{\ell,p} \rangle
\end{equation*}
for all $\ell>\ell'$ and $R>0$. Now each of the terms in $q_{\ell,p}$ is of the form $g\otimes v_\ell$ where $v_\ell$ is uniformly bounded in $\ell$ and $g$ is independent of $\ell$. Furthermore, the number of terms in $q_{\ell,p}$ is also independent of $\ell$ (it depends only on $p$ by construction). Since $1-G_R(A)$ converges to 1 strongly by the functional calculus, we find that
\begin{align*}
\lim_{R\rightarrow \infty} \sup_{\ell\in \NN} \lVert ((1-G_R(A))\otimes 1) q_{\ell,p}  \lVert=0.
\end{align*}
Picking $R$ larger than some $R'$ and $\ell>\ell'$ we find that
\begin{equation*}
\langle A_1\psi_\ell,((1-G_R(A))\otimes 1 )A_1\psi_\ell \rangle\leq \varepsilon.
\end{equation*}
This finishes the proof.
\end{proof}
\noindent The following lemma finishes the proof of Theorem \ref{Thm:Numberstructure_Massless}.
\begin{lem}
$\lVert \psi\lVert=1$.
\end{lem}
\begin{proof}
Let $\varepsilon>0$. Pick $R'>0$ and $\ell '\in \NN$ such that $\lVert (1-\Gamma(G_R(A)))\psi_\ell\lVert\leq \frac{\varepsilon}{3}$ when $R>R'$, $\ell >\ell'$ and $A\in \{  -i\nabla_k,k \}$. Using Lemma \ref{Lem:Basicstuff} we see $\langle \psi_\ell ,N\psi_\ell \rangle$ is uniformly bounded by a constant $C$ and so we find
\begin{equation*}
\lVert (1-1_{ [0,p] }(N) )\psi_\ell \lVert=\lVert 1_{  (p,\infty)   }(N)\psi_\ell\lVert \leq \frac{1}{\sqrt{p}}\lVert 1_{  (p,\infty)   }(N)N^{\frac{1}{2}}\psi_\ell  \lVert\leq \frac{\sqrt{C}}{\sqrt{p}}  
\end{equation*}
Hence we may pick $p$ so large that $\lVert (1-1_{ [0,p]  }(N))\psi_\ell \lVert\leq \frac{\varepsilon}{3}$ uniformly in $\ell$. We now find
\begin{align*}
1=&\lVert \psi_\ell \lVert\\\leq& \lVert (1-\Gamma(G_{R}(k)))\psi_\ell \lVert +\lVert \Gamma(G_{R}(k))(1-\Gamma(G_{R}(-i\nabla_k)))\psi_\ell \lVert\\& +\lVert \Gamma(G_{R}(k))\Gamma(G_{R}(-i\nabla_k))(1-1_{ [0,p]  }(N))\psi_\ell \lVert\\&+\lVert \Gamma(G_{R}(k))\Gamma(G_{R}(-i\nabla_k))1_{ [0,p]  }(N)\psi_\ell \lVert\\\leq& \varepsilon +\lVert \Gamma(G_{R}(k))\Gamma(G_{R}(-i\nabla_k))1_{ [0,p]  }(N)\psi_\ell \lVert.
\end{align*}
Since $\Gamma(G_{R}(k))\Gamma(G_{R}(-i\nabla_k))1_{\{ [0,p] \} }(N)$ is compact, we may take $\ell$ to $\infty$ and find 
\begin{equation*}
1-\varepsilon\leq \lVert \Gamma(G_{R}(k))\Gamma(G_{R}(-i\nabla_k))1_{\{ [0,p] \} }(N)\psi \lVert\leq \lVert \psi\lVert \leq \liminf_{\ell\rightarrow \infty} \lVert \psi_\ell\lVert=1.
\end{equation*}
This finishes the proof.
\end{proof}
\appendix

\section{Measure Theory.}
In this section, we introduce the necessary measure theoretic tools to prove the HVZ theorem. Throughout this section, we will consider a $\sigma$-finite measure space $(\cM,\cF,\mu)$. If $f\colon \cM\rightarrow \RR$ is a measurable map and $M_f$ is the corresponding multiplication operator then it is easy to see that
\begin{align*}
\sigma(M_f)&=\{ \lambda\in \RR\mid \mu( (\lambda-\varepsilon,\lambda+\varepsilon)  )>0 \,\, \text{for all $\varepsilon>0$}  \}:=\text{essran}(f,\mu)\\
\sigma_{\textup{ess}}(M_f)&=\{ \lambda\in \RR\mid \text{Dim$(1_{  \{  \lambda-\varepsilon<\omega<\lambda+\varepsilon \} } L^2(\cM,\cF,\mu))=\infty$ for all $\varepsilon>0$}  \}.
\end{align*}
Here $\text{essran}(f,\mu)$ is called the essential range of $f$ under $\mu$.

\begin{lem}\label{Lem:SpecialSequenceOfDispersions}
Let $f:\cM\rightarrow \RR$ be measurable and assume $f\geq m>0$ almost everywhere. Define
\begin{equation*}
f_n=\sum_{k=0}^{\infty}\frac{k+1}{2^n}1_{\left \{f\in  \left [\frac{k}{2^n},\frac{k+1}{2^n}\right)  \right \}}.
\end{equation*}
Then $f_n/f$ and $f/f_n$ converges to 1 in $L^{\infty}(\cM,\cF,\mu)$ and
\begin{equation*}
\lim_{n\rightarrow \infty}\inf(\sigma(M_{f_n}))=\inf(\sigma(M_{f})).
\end{equation*}
Furthermore, $\sigma_{\textup{ess}}(M_{f_n})\neq \emptyset$ for all $n\in \NN$ if $\sigma_{\textup{ess}}(M_{f})\neq \emptyset$ and in this case
\begin{equation*}
\lim_{n\rightarrow \infty}\inf(\sigma_{\textup{ess}}(M_{f_n}))=\inf(\sigma_{\textup{ess}}(M_{f})).
\end{equation*}
\end{lem}
\begin{proof}
The sum defining $f_n$ is pointwise finite, so it defines a measurable and nonnegative function. Note that
\begin{equation}\label{eq:UniformConv}
\lvert f(x)-f_n(x)\lvert\leq \frac{1}{2^n},
\end{equation}
almost everywhere and that $f_n\geq f$ almost everywhere. Hence the following calculation is true almost everywhere
\begin{equation*}
\left \lvert \frac{f}{f_n}-1\right \lvert  =\left \lvert \frac{f-f_n}{f_n}\right \lvert \leq  \frac{1}{2^nm}.
\end{equation*}
This implies $ f/f_n-1\in L^{\infty}(\cM,\cF,\mu)$ and converges to 0 in this topology. A similar argument works for $f_n/f-1$. 

Equation (\ref{eq:UniformConv}) shows $\cD(M_f)=\cD(M_{f_n})$ for all $n$ and on this set we have the inequality $\lVert (M_f-M_{f_n}) \psi \lVert\leq  2^{-n}\lVert \psi \lVert$ which shows $M_{f_n}$ converges to $M_{f}$ in norm resolvent sense (see \cite[Theorem VIII.25]{RS1}). Since the operators $M_{f_n}$ are uniformly bounded below by 0, we conclude
\begin{equation*}
\lim_{n\rightarrow \infty}\inf(\sigma(M_{f_n}))=\inf(\sigma(M_{f})).
\end{equation*}
Assume $\lambda=\inf(\sigma_{\textup{ess}}(M_{f}))$ is finite. Then $1_{\{f\in (\lambda-q^{-1},\lambda+q^{-1})\}}$ is an infinite dimensional projection for every $q\in \mathbb{N}$. Fix $n\in \mathbb{N}$ and pick $k\in \NN_0$ such that 
 \begin{equation*}
 \lambda\in [2^{-n}k,2^{-n}(k+1)).
 \end{equation*}
 Note that either $2^{-n}(k+1)$ or $2^{-n}k$ belongs to the essential spectrum of $M_{f_n}$ since it is an eigenvalue of infinite dimension. In particular, $\sigma_{\textup{ess}}(M_{f_n})$ contains a point $\lambda_n$ such that $\lvert \lambda-\lambda_n\lvert\leq 2^{-n}$.

 Hence we have now proven, that for each $n\in \mathbb{N}$ there is a $\lambda_n \in \sigma_{\textup{ess}}(M_{f_n})$ such that $\{\lambda_n\}_{n=1}^\infty$ converges to $\lambda$. In particular, $\mu_n=\inf(\sigma_{\textup{ess}}(M_{f_n}))$ is finite. Note that $\{ \mu_n \}_{n=1}^\infty$ is bounded since $\mu_n\leq \lambda_n$ for all $n\in \NN$ and $\{ \lambda_n \}_{n=1}^\infty$ is convergent.

 Let $\mu$ be any limit point of $\{ \mu_n \}_{n=1}^\infty$. Then $\mu\in \sigma_{\textup{ess}}(M_f)$ by elementary properties of norm resolvent convergence so $\lambda\leq \mu$. On the other hand we have $\mu_n\leq \lambda_n$ for all $n\in \NN$ so $\mu\leq \lambda$. This implies $\mu=\lambda$, and so $\lambda$ is the only accumulation point of the bounded sequence $\{ \mu_n \}_{n=1}^\infty$. This implies $\{ \mu_n \}_{n=1}^\infty$ converges to $\lambda$.
\end{proof}

\noindent We have the following definition:
\begin{defn}
	Write $\mathbb{R}_+=[0,\infty)$. A continuous resolution for the measure space $(\mathcal{M},\mathcal{F},\mu)$ is a collection $(A_x)_{x\in \mathbb{R}_+}\subset \mathcal{F}$ such that $\mu(A_0)=0$, $A_{x}\subset A_{y}$ when $x\leq y$, $\mu(A_x)<\infty$ for all $x\in \mathbb{R}^+$, $1_{A_x}\rightarrow 1_{A_y}$ $\mu$-a.e for $x\rightarrow y$ and 
	\begin{equation*}
	\bigcup_{x\geq 0}A_x=\mathcal{M}.
	\end{equation*} 
\end{defn}

\begin{lem}\label{Lem:Cutting the space}
	Assume $(\mathcal{M},\mathcal{F},\mu)$ allows a continuous resolution $(A_x)_{x\in \mathbb{R}_+}$ and let $A\in \mathcal{F}$ with $0<\mu(A)$. For every $\lambda\in (0,\mu(A))$ there is $B\subset A$ with $B\in \cF$ and $\mu(B)=\lambda$. Furthermore, there is a partition of $A$ into disjoint measurable sets $\{B_n\}_{n=1}^\infty$ such that $0<\mu(B_n)<\infty$ for all $n\in \NN$.
\end{lem}
\begin{proof}
	Define $f:[0,\infty) \rightarrow [0,\infty)$ by
	\begin{equation*}
	f(x)=\int_{\mathcal{M}} 1_{A}(y)1_{A_x}(y)d\mu(y).
	\end{equation*}
	Then $f$ is increasing and continuous by the dominated convergence theorem. Furthermore, $f(0)=0$ and by monotone convergence we find
	\begin{equation*}
	\mu(A)=\lim\limits_{x\rightarrow \infty}f(x).
	\end{equation*} 
	Let $\lambda\in (0,\mu(A))$. The intermediate value theorem now gives $x_0\in [0,\infty)$ such that $\lambda=f(x_0)$ implying $B=A_{x_0}\cap A$ has the properties claimed in the lemma. We now prove that the subdivision of $A$ exists. For each $n\in \mathbb{N}$ pick $x_n\in [0,\infty)$ such that
	\begin{equation*}
	f(x_n)=\begin{cases}
	(1-2^{-n})\mu(A) & \mu(A)<\infty \\ n & \mu(A)=\infty
	\end{cases}
	\end{equation*}
	Since $f$ is increasing and $f(x_n)<f(x_{n+1})$ we find that $x_n<x_{n+1}$. Define now $E_n=A\cap A_{x_n}$, $\widetilde{B}_1=E_1$ and $\widetilde{B}_n=E_n\backslash E_{n-1}$ for $n\geq 2$. Then
	\begin{equation*}
	\mu(\widetilde{B}_1)=\mu(E_1)=f(x_1)
	\end{equation*}
	so $0<\mu(\widetilde{B}_1)<\infty$. For $n\geq 2$ we see
	\begin{equation*}
	\mu(\widetilde{B}_n)=\mu(E_n\backslash E_{n-1})=f(x_n)-f(x_{n-1}),
	\end{equation*}
    so $0<\mu(\widetilde{B}_n)<\infty$. Furthermore, $\widetilde{B}_n\cap \widetilde{B}_m=\emptyset$ for $n\neq m$ by construction. Define
	\begin{equation*}
	x=\lim\limits_{n\rightarrow \infty}x_n.
	\end{equation*}
	If $x=\infty$ then
	\begin{equation*}
	\bigcup_{n=1}^{\infty}B_n=\bigcup_{n=1}^{\infty}E_n=A\cap \bigcup_{n=1}^{\infty}A_{x_n}=A,
	\end{equation*}
	so we may pick $B_n=\widetilde{B}_n$ in this case. If $x<\infty$ we note that
	\begin{equation*}
	\infty >\int_{\mathcal{M}} 1_{A}(y)1_{A_x}(y) d\mu(y)=f(x)=\lim\limits_{n\rightarrow \infty} f(x_n)=\mu(A),
	\end{equation*}
	so $\mu(A)<\infty$ and $f(x)=\mu(A)$. Furthermore, 
	\begin{equation*}
	\mu \biggl( \bigcup_{n=1}^\infty \widetilde{B}_n \biggl)=\mu \biggl( \bigcup_{n=1}^\infty E_n \biggl)=\lim\limits_{n\rightarrow \infty}\mu(E_n)=\mu(A).
	\end{equation*}
	Let $B=A\backslash \bigcup_{n=1}^\infty \widetilde{B}_n $ and note that $\mu(B)=0$. Define $B_1=\widetilde{B}_1\cup B$ and $B_n=\widetilde{B}_n$ for $n\geq 2$. Then $\mu(B_n)=\mu(\widetilde{B}_n)\in (0,\infty)$ for all $n\in \NN$ and $B_n\cap B_m=\emptyset $ for $n\neq m$. Furthermore,
	\begin{equation*}
	\bigcup_{n=1}^{\infty}B_n= \bigcup_{n=1}^\infty \widetilde{B}_n \cup\biggl ( A\backslash \bigcup_{n=1}^\infty \widetilde{B}_n\biggl )=A
	\end{equation*}
	which finishes the proof.
\end{proof}
\noindent Using Lemma \ref{Lem:Cutting the space} we prove
\begin{lem}\label{Lem:Important property of continouos resolutions}
	Assume $(\mathcal{M},\mathcal{F},\mu)$ allows a continuous resolution and $f:\mathcal{M}\rightarrow \mathbb{R}$ is measurable. For every $z\in \sigma(M_f)$ there is a collection of sets $\{A_n\}_{n=1}^\infty\subset \mathcal{F}$ such that $\lvert f(x)-z\lvert\leq \frac{1}{n}$ on $A_n$, $A_n\cap A_m=\emptyset$ if $m\neq n$, $\mu(A_n)>0$ and
	\begin{equation*}
	\mu \left( \bigcup_{n=1}^{\infty} A_n \right)=\sum\limits_{n=1}^{\infty}\mu(A_n)<\infty.
	\end{equation*}
\end{lem}
\begin{proof}
Let $z\in \sigma(M_f)$ and define
	\begin{equation*}
	B_n=\{ f\in (z-n^{-1},z+n^{-1}) \}.
	\end{equation*}
	There are now several cases. Assume first that $\mu(B_n)=\infty$ for all $n\in \mathbb{N}$. Then we define $A_n$ recursively as follows: By Lemma \ref{Lem:Cutting the space} we may pick $A_1\subset B_1$ such that $\mu(A_1)=1$. Assume now we have constructed disjoint sets $A_1,\dots,A_{n-1}$ such that $A_j\subset B_j$ for all $j\in \{1,\dots,n-1\}$ and $\mu(A_j)=\frac{1}{j^2}$. Then
	\begin{equation*}
	\infty=\mu(B_n)\leq \mu(B_n\backslash (A_1\cup\dots\cup A_{n-1}))+\sum\limits_{j=1}^{n-1}\frac{1}{j^2}.
	\end{equation*}
	so $\mu(B_n\backslash (A_1\cup\dots\cup A_{n-1}))=\infty$. By Lemma \ref{Lem:Cutting the space} there is  $A_n\subset B_n\backslash (A_1\cup\dots\cup A_{n-1})$ such that $\mu(A_n)=\frac{1}{n^2}$ for all $n\in \NN$. Hence we have now constructed a sequence of disjoint sets $\{A_n\}_{n=1}^\infty\subset \mathcal{F}$ such that $A_n\subset B_n$ and $\mu(A_n)=\frac{1}{n^2}$ for all $n\in \NN$. Since 
	\begin{equation*}
	\mu \left( \bigcup_{n=1}^{\infty} A_n \right)=\sum\limits_{n=1}^{\infty}\mu(A_n)=\sum\limits_{n=1}^{\infty}\frac{1}{n^2}<\infty
	\end{equation*}
	we are done. Assume now that there is an $N\in \mathbb{N}$ such that $\mu(B_N)<\infty$. Define
	\begin{equation*}
	C_n=B_{N+n}.
	\end{equation*}
	Since the $B_n$ are decreasing we find that $C_n\subset B_n$ and
	\begin{equation*}
	\lim_{n\rightarrow \infty } \mu(C_n)=\mu(\{ f=z \})<\infty. 
	\end{equation*}
	If $\mu(\{ f=z \})>0$ we apply Lemma \ref{Lem:Cutting the space} and obtain a disjoint subdivision $\{A_n\}_{n=1}^\infty$ of $\{ f=z \}$. Since $A_n\subset \{ f=z \}$ for all $n$ and
	\begin{equation*}
	\mu \biggl( \bigcup_{n=1}^{\infty} A_n \biggl)=\sum\limits_{n=1}^{\infty}\mu(A_n)=\mu(\{ f=z \})<\infty,
	\end{equation*}
	we are finished.
	
	What remains is the case $\mu(\{ f=z \})=0$. We know that $\mu(C_n)>0$ for all $n$ since $z\in \text{essran}(f,\mu)$. We now claim that there are natural numbers $n_1<n_2<n_3<\dots$ such that $\mu(C_{n_{k}}\backslash C_{n_{k+1}})>0$. Define $n_1=1$ and assume that $n_1<n_2<\dots<n_{k}$ has been constructed. Define
	\begin{equation*}
	\cA=\{ n\in \mathbb{N}\mid \mu(C_{n_k}\backslash C_{n})>0 \,\,\, \text{and}\,\,\,  n>n_k  \}.
	\end{equation*}
	Then $\cA\neq \emptyset$ because if $\cA=\emptyset$ then
	\begin{equation*}
	\mu(C_{n_k})=\mu(C_n)
	\end{equation*} 
	for all $n>n_{k}$ implying that $\mu(\{ f=z \})=\mu(C_{n_{k}})>0$ which is a contradiction. Since $\cA\subset \NN$ we can now define $n_{k+1}=\min (\cA)$. Let
	\begin{equation*}
	A_{k}=C_{n_k}\backslash C_{n_{k+1}}.
	\end{equation*}
	Then $A_{k}\subset C_{n_k}\subset C_{k}\subset B_k$ so $\lvert f(x)-z\lvert \leq \frac{1}{k}$ holds on $A_k$. Furthermore, $0<\mu(A_k)\leq \mu(C_{n_k})<\infty$ for all $k$ and $\{A_k\}_{k=1}^\infty$ is disjoint by construction. Note also
	\begin{equation*}
	\sum\limits_{k=1}^{\infty}\mu(A_k)=\mu \left( \bigcup_{k=1}^{\infty} A_k \right)\leq \mu(C_1)<\infty.
	\end{equation*}
This finishes the proof.
\end{proof}
\noindent This leads to the following theorem which we shall need. The reader is reminded that singletons are sets of the form $\{x\}$ for some $x\in \cM$.

\begin{thm}\label{Thm:EssentalPropertyCutSpaces}
Assume $(\cM,\cF,\mu)$ is $\sigma$-finite and that all singletons are measurable. Then $A=\{x\in \cM\mid \mu(\{ x \})>0  \}$ is countable and therefore measurable. Let $\mu_{A^c}$ denote the measure $\mu_{A^c}(B)=\mu(A^c\cap B)$ and assume that $(\cM,\cF,\mu_{A^c})$ has a continuous resolution. Let $f:\cM\rightarrow \RR$ be measurable, $B=\textup{essran}(f,\mu_{A^c})$ and define
\begin{equation*}
C=\{ \lambda \in \RR \mid \exists \{ \lambda_n \}_{n=1}^\infty\subset A, \lambda_n\neq \lambda_m \,\, \forall n\neq m \,\, \text{and}\,\, \lvert f(\lambda_n)-\lambda\lvert\leq n^{-1}   \}.
\end{equation*}
Then
\begin{equation*}
\sigma_{\textup{ess}}(M_f)=B\cup C.
\end{equation*}
Assume in addition that $\mu(A)<\infty$. If $z_1,\dots,z_k\in \sigma_{\textup{ess}}(M_f)$ then there are $k$ collections of sets $\{A^1_n\}_{n=1}^\infty,\dots, \{A^k_n\}_{n=1}^\infty\subset \mathcal{F}$ such that $\lvert f(x)-z_i\lvert\leq \frac{1}{n}$ on $A^i_n$, $A^i_n\cap A^j_m=\emptyset$ if $i\neq j$ or $m\neq n$, $\mu(A^i_n)>0$ and
\begin{equation*}
\mu \left( \bigcup_{n=1}^{\infty} A^i_n \right)=\sum\limits_{n=1}^{\infty}\mu(A^i_n)<\infty.
\end{equation*}
\end{thm}
\begin{proof}
By $\sigma$-finiteness of $(\cM,\cF,\mu)$ we know $\cM$ can be divided into countably many disjoint subsets of finite measure. Each of these sets can only contain countably many elements from $A$ and these elements must all have finite measure. Hence $A$ must be countable and all singletons must have finite measure.

Let $\lambda \in B\cup C$. If $\lambda\in B$ then Lemma \ref{Lem:Important property of continouos resolutions} gives a sequence of disjoint elements $\{ E_n \}_{n=1}^\infty\subset \cF$ such that $0<\mu_{A^c}(E_n)<\infty$ and $\lvert f-\lambda\lvert\leq \frac{1}{n}$ on $E_n$ for all $n\in \NN$. Let $A_n=E_n\backslash A$ and note $\{\mu(A_n)^{-\frac{1}{2}}1_{A_n}\}_{n=1}^\infty$ will be a Weyl sequence for $\lambda$, so $\lambda\in \sigma_{\textup{ess}}(M_f)$. If $\lambda \in C$ and $\{ \lambda_n \}_{n=1}^\infty\subset A$ is the corresponding sequence then $\{\mu(\{\lambda_n\})^{-\frac{1}{2}}1_{\{\lambda_n \}}\}_{n=1}^\infty$ will be a Weyl sequence for $\lambda$, so $\lambda\in \sigma_{\textup{ess}}(M_f)$.

Let $\lambda\in (B\cup C)^c$. If $\mu( \{  \lvert f-\lambda\lvert  < n^{-1} \} \cap A^c)>0$ or $A\cap \{  0<\lvert f-\lambda\lvert  < n^{-1} \} \neq \emptyset $ for all $n\in \mathbb{N}$ then $\lambda\in B\cup C$ which would be a contradiction. Hence there is an $N\in \mathbb{N}$ such that $A^c\cap \{ \lvert f-\lambda\lvert < N^{-1}\}$ is a null set and $\{ \lvert f-\lambda\lvert<N^{-1} \}\cap A=\{ \lvert f-\lambda\lvert=0 \}\cap A $. In particular, the spectral projection of $f$ onto $(\lambda-N^{-1},\lambda+N^{-1})$ is given by $1_{\{ \lvert f-\lambda\lvert=0 \}\cap A}$. We note that $\{ \lvert f-\lambda\lvert=0 \}\cap A$ is finite as we would otherwise have $\lambda\in C$. Hence $1_{\{ \lvert f-\lambda\lvert=0 \}\cap A}$ has finite-dimensional range and so $\lambda\in \sigma_{\textup{ess}}(M_f)^c$. Hence we have now established
\begin{equation*}
\sigma_{\textup{ess}}(M_f)=B\cup C.
\end{equation*}
It only remains to prove the last part of the theorem so we assume $\mu(A)<\infty$. Let $z_1,\dots,z_k\in \sigma_{\textup{ess}}(M_f)$. We start by dealing with the special case where $z_i\neq z_j$ when $i\neq j$.

 For each $i\in \{1,\dots,k\}$ we either have $z_i\in B$ or $z_i\in C$. If $z_i\in B$ then Lemma \ref{Lem:Important property of continouos resolutions} gives a sequence of disjoint elements $\{ \widetilde{A}_n^i \}_{n=1}^\infty\subset \cF$ such that $\lvert f-z_i\lvert\leq \frac{1}{n}$ on $\widetilde{A}^i_n$ for all $n\in \NN$, $0<\mu_{A^c}(\widetilde{A}_n^i)\leq \mu(\widetilde{A}_n^i)$ for all $n\in \NN$ and
\begin{equation*}
\mu \left( \bigcup_{n=1}^{\infty} \widetilde{A}^i_n \right)=\sum\limits_{n=1}^{\infty}\mu(\widetilde{A}^i_n)\leq \sum\limits_{n=1}^{\infty}\mu_{A^c}(\widetilde{A}^i_n)+\mu(A)<\infty
\end{equation*}
If $z_i \in C$ and $\{ \lambda^i_n \}_{n=1}^\infty$ is the corresponding sequence then we define $\widetilde{A}_n^i=\{\lambda^i_n\}$. Note that $\{ \widetilde{A}_n^i \}_{n=1}^\infty\subset \cF$ is disjoint and $\lvert f-\lambda\lvert\leq \frac{1}{n}$ on $\widetilde{A}^i_n$ for all $n\in \NN$. Furthermore, $0<\mu(\widetilde{A}_n^i)$ for all $n\in \NN$ and
\begin{equation*}
\mu \left( \bigcup_{n=1}^{\infty} \widetilde{A}^i_n \right)=\sum\limits_{n=1}^{\infty}\mu(\widetilde{A}^i_n)\leq \mu(A)<\infty
\end{equation*}
by assumption. Pick $N$ so large that $2N^{-1}<\max_{i\neq j}\lvert z_i-z_j\lvert$ and define $A^i_n=\widetilde{A}_{N+n}^i$ for all $n\in \NN$ and $i\in \{1,\dots, k\}$. Then
\begin{align*}
A^i_n&\subset \left\{ \lvert f-z_i\lvert<\frac{1}{N+n}  \right\}\subset \left\{ \lvert f-z_i\lvert<\frac{1}{n} \right \}\\
\mu \left( \bigcup_{n=1}^{\infty} A^i_n \right)&=\sum\limits_{n=1}^{\infty}\mu(A^i_n)\leq \sum\limits_{n=1}^{\infty}\mu(\widetilde{A}^i_n)<\infty
\end{align*}
and $0<\mu(A^i_n)$ for all $i\in \{1, \dots,k \}$ and $n\in \NN$. If $x\in A^i_n\cap A^j_m$ for $i\neq j$ we would have $\lvert z_i-z_j\lvert\leq  \lvert z_i-f(x)\lvert +\lvert f(x)-z_j\lvert<\frac{2}{N}$ which is a contradiction. So $A^i_n\cap A^j_m=\emptyset$ if $i\neq j$. If $i=j$ and $n\neq m$ we find $A^i_n\cap A^j_m=\widetilde{A}_{N+n}^i\cap \widetilde{A}_{N+m}^i=\emptyset$. Hence we are now finished in the case where $z_1,\dots,z_k$ are different.

Assume now $z_1,\dots,z_k$ are not all different and let $\lambda_1,\dots,\lambda_\ell$ be the different elements in $\{z_1,\dots, z_k \}$. For each $i\in \{ 1,\dots , k \}$ there is $h(i)\in \{1,\dots, \ell  \}$ such that $z_i=\lambda_{h(i)}$. Pick sequences $\{\widetilde{A}^1_n\}_{n=1}^\infty,\dots,\{\widetilde{A}^\ell_n\}_{n=1}^\infty \subset \cF$ as in the theorem for the collection $\lambda_1,\dots,\lambda_\ell\in \sigma_{\textup{ess}}(M_f)$ and define $A^i_n=\widetilde{A}^{h(i)}_{i+kn}$. Observe that
\begin{align*}
A^i_n&\subset \left\{ \lvert f-\lambda_{h(i)}\lvert<\frac{1}{i+kn} \right \}\subset \left\{ \lvert f-z_i\lvert<\frac{1}{n} \right \}\\
\mu \left( \bigcup_{n=1}^{\infty} A^i_n \right)&=\sum\limits_{n=1}^{\infty}\mu(A^i_n)\leq \sum\limits_{n=1}^{\infty}\mu(\widetilde{A}^{h(i)}_n)<\infty
\end{align*}
and $0<\mu(A^i_n)$ for all $i\in \{1, \dots, k \}$ and $n\in \NN$. If $i\neq j$ or $n\neq m$ then $j+mk\neq i+nk$ since $1\leq i,j\leq k$ and so
\begin{align*}
A^i_n\cap A^j_m=\widetilde{A}_{i+kn}^{h(i)}\cap \widetilde{A}_{j+km}^{h(j)}=\emptyset.
\end{align*}
This finishes the proof.
\end{proof}
\noindent The following two lemmas show that Theorem \ref{Thm:EssentalPropertyCutSpaces} can be applied to a wide range of $L^2$-spaces.
\begin{lem}\label{Lem:BredFamafrum}
Let $A\subset \ZZ$ and let $\mu$ be some measure on $(A\times \RR^\nu, \cB(A\times \RR^\nu) )$ which is finite on compact sets. Then the assumptions of Theorem \ref{Thm:EssentalPropertyCutSpaces} are satisfied if $\mu(B)<\infty$ where $B=\{ x\in A\times \RR^\nu\mid \mu(\{ x \})>0 \}$ and  $\mu_{B^c}$ is zero on sets of the form $\{i\}\times C$ with $C= \{ x\in \RR^\nu \mid \lvert x\lvert=c  \}$.
\end{lem}
\begin{proof}
$A\times \RR^\nu$ is a metric space that can be covered by countably many compact sets so $(A\times \RR^\nu, \cB(A\times \RR^\nu) ,\mu)$  is $\sigma$-finite and singletons are measurable. In particular, $B$ is countable and therefore measurable. 

Define $U_x=\{ y\in \RR^{\nu+1}\mid \lvert y\lvert\leq x  \}\cap ( A\times \RR^\nu)$. Then $1_{U_x}$ converges to $1_{U_y}$ pointwise for $x\rightarrow y$ except at points in $\partial U_y$. Note that $\partial U_y$ is a finite union of sets of the form $\{i \}\times \{ x\in \RR^\nu \mid \lvert x\lvert=c  \}$ with $c\geq 0$ and $i\in \ZZ$. Hence $\partial U_y$ is a $\mu_{B^c}$ null set proving that $\{U_x\}_{x\in [0,\infty)}$ defines a continuous resolution for $\mu_{B^c}$.
\end{proof}
\noindent The following lemma is central to the spectral analysis.
\begin{lem}\label{Lem:CannonicalSpace}
Let $\cH$ be a separable Hilbert space and $A$ be a selfadjoint operator on $\cH$. Then there is a measure space $(\cM,\cF,\mu )$ that fulfil the conditions in Theorem \ref{Thm:EssentalPropertyCutSpaces} and a unitary map $U:\cH\rightarrow L^2(\cM,\cF,\mu )$ such that $UAU^*$ is a multiplication operator on $L^2(\cM,\cF,\mu )$.
\end{lem}
\begin{proof}
We follow the construction found in \cite{Teschl}. Let $P_A$ be the spectral measure of $A$ and $\{ \psi_n \}_{n\in B}$ (where $B\subset \NN$) be a complete collection of normed cyclic vectors. Define the measure $\mu_n$ by the expression $\mu_n(C)=\langle \psi_n,P_A(C)\psi_n\rangle$. By \cite{Teschl} we find a unitary map 
\begin{equation*}
U_1:\cH \rightarrow \cK_1=\bigoplus_{n \in B} L^2(\RR,\cB(\RR),\mu_n), 
\end{equation*}
such that $U_1AU_1^*$ acts like multiplication by the identity map $f(x)=x$ on each of the component spaces. By standard measure theory of kernels, there is measure $\widetilde{\mu}$ on $\cB(B\times \RR)$ defined by
\begin{equation*}
\int_{B\times \RR}f(n,k)d\widetilde{\mu}(n,k)=\sum_{n\in B}\int f(n,x)d\mu_n(x)
\end{equation*}
for any non negative and measurable map $f:B\times \RR\rightarrow \RR$. Hence we find a unitary map
\begin{equation*}
U_2:\cK_1\rightarrow \cK_2=L^2\left (B\times \RR,\cB(B\times \RR),  \widetilde{\mu} \right )
\end{equation*}
such that $U_2U_1AU_1^*U_2^*$ acts like multiplication by the map $\omega(n,x)=f(x)=x$. Note that each $\{n\} \times \RR$ has measure 1 for all $n\in \NN$, so $(B\times \RR,\cB(B\times \RR),  \widetilde{\mu} )$ is $\sigma$-finite. Hence there is a strictly positive measurable map $f$ on $B\times \RR$ which integrates to 1. Define $\mu=f\widetilde{\mu}$ and note that $\mu$ is a probability measure. Let
\begin{equation*}
U_3:\cK_2\rightarrow \cK_3=L^2(B\times \RR,\cB(B\times \RR),  \mu )
\end{equation*}
be multiplication by $f^{-\frac{1}{2}}$. Then $U_3$ is unitary map and $U_3U_2U_1AU_1^*U_2^*U_3^*$ acts like multiplication by $\omega$ on $\cK_3$. Furthermore, it is clear that $(B\times \RR,\cB(B\times \RR),  \mu )$ satisfies the conditions in Lemma \ref{Lem:BredFamafrum} since $\mu(B\times \RR)=1$ and sets of the form $\{i\}\times \{ x\in \RR\mid \lvert x\lvert=c \}$ only contain finitely many points. This finishes the proof.
\end{proof}

\section{ Spectral Theory of tensor products}
In this section we list a few results regarding the tensor product of operators. A good reference for these results are \cite{Schmudgen}. Throughout this section,  $\mathcal{H}_1,\dots,\mathcal{H}_n$ will denote be a finite collection of Hilbert spaces. For $V_j\subset \mathcal{H}_j$ subspaces, we define the algebraic tensor product
\begin{equation*}
V_1\widehat{\otimes}\cdots\widehat{\otimes} V_n=\textup{Span}\{ x_1\otimes\cdots\otimes x_n \mid x_j\in V_j \}.
\end{equation*}
Most of the content in the following theorem can be found in \cite{Schmudgen}. The remaining items can easily be deduced.
\begin{thm}\label{Thm:BasicPropTensor}
	Let $T_j$ be an operator on $\mathcal{H}_j$ for all $j\in \{1,\dots,n\}$. Then there is a unique linear map $T=T_1\widehat{\otimes}\cdots\widehat{\otimes} T_n$ defined on $\mathcal{D}(T_1)\widehat{\otimes}\cdots\widehat{\otimes} \mathcal{D}(T_n)$ such that
	\begin{equation*}
	T_1\widehat{\otimes}\cdots\widehat{\otimes} T_n(x_1\otimes\cdots\otimes x_n)=T_1x_1\otimes\cdots\otimes T_nx_n,
	\end{equation*}
	when $x_j\in \mathcal{D}(T_j)$ for all $j\in \{1,\dots,n\}$. We also have the following:
	\begin{enumerate}
	\item[\textup{(1)}] If $T_j$ is densely defined for all $j\in \{1, \dots, n \}$ then $T$ is densely defined and $T_1^*\widehat{\otimes}\cdots\widehat{\otimes} T_n^* \subset T^* $.
	
	\item[\textup{(2)}]  If $T_j$ is closable for all $j\in \{1, \dots, n \}$ then $T$ is closable. We will then write $\overline{T}=T_1\otimes\cdots\otimes T_n$. Furthermore, the following identities hold
		\begin{align*}
		T_1\otimes\cdots\otimes T_n&=\overline{T}_1\otimes\cdots\otimes \overline{T}_n\\
		T_1^*\otimes\cdots\otimes T_n^*&= (T_1\otimes\cdots\otimes T_n)^*.
		\end{align*}

		\item[\textup{(3)}] If $T_j$ is symmetric (selfadjoint, unitary, a projection) for all $j\in \{1, \dots, n \}$ then $T$ is symmetric (selfadjoint, unitary, a projection).
		
		\item[\textup{(4)}] If $T_j\geq 0$ for all $j\in \{1,\dots,n\}$ then $T\geq 0$.

		\item[\textup{(5)}] If $T_j$ is bounded for all $j\in \{1, \dots, n \}$ then $T$ is bounded and
			\begin{equation*}
			\lVert T\lVert=\lVert T_1\lVert\cdots\lVert T_n\lVert=\lVert T_1\otimes\cdots \otimes T_n \lVert. 
			\end{equation*}
			\end{enumerate}
	
\end{thm}
\noindent The following result is important.

\begin{thm}\label{Thm:SpectralPropTensor}
	Let $T_j$ be a selfadjoint operator on $\mathcal{H}_j$ for all $j\in \{1,\dots,n\}$ and define 
	\begin{align*}
	H_j&=1\otimes\cdots\otimes T_j \otimes\cdots\otimes 1,\\
	H&=H_1+H_2+\cdots+H_n.
	\end{align*}
	Then
	\begin{enumerate}
	\item[\textup{(1)}]  $(H_1,\dots,H_n)$ is a tuple of strongly commuting selfadjoint operators with $\sigma(H_j)=\sigma(T_j)$. The joint spectrum is $\sigma(T_1)\times\cdots \times \sigma(T_n)$ and if $f:\RR\rightarrow \CC$ is Borel measurable then $f(H_j)=1\otimes\cdots\otimes f(T_j)\otimes\cdots\otimes 1$.
	
	\item[\textup{(2)}]  $H$ is essentially selfadjoint with
	\begin{equation*}
	e^{it\overline{H}}=e^{itT_1}\otimes\cdots\otimes e^{itT_n}\,\,\,\,\, t\in \mathbb{R}.
	\end{equation*}

	\item[\textup{(3)}] If $V_j$ is a core for $T_j$ then $V_1\widehat{\otimes}\cdots\widehat{\otimes} V_n$ is a core for $\overline{H}$.
	
	\item[\textup{(4)}] Assume $T_j$ is semibounded for all $j\in \{1, \dots, n \}$ and define $\lambda_j=\inf(\sigma(T_j))$. Then $H$ is selfadjoint and semibounded with $\inf(\sigma(H))=\lambda:=\lambda_1+\cdots +\lambda_n$. Let $P_B$ denote the spectral measure for an operator $B\in\{H,T_1,\dots,T_n\}$. Then
	\begin{align*}
	e^{-tH}&=e^{-tT_1}\otimes\cdots\otimes e^{-tT_n}\,\,\,\,\, t\geq 0\\
	P_{H}(\{\lambda\})&=P_{T_1}(\{\lambda_1\})\otimes\cdots\otimes P_{T_n}(\{\lambda_n\}).
	\end{align*}
	In particular, $\textup{Dim}(P_{H}(\{\lambda\}))=\textup{Dim}(P_{T_{1}}(\{\lambda_1\})) \cdots  \textup{Dim}(P_{T_{n}}(\{\lambda_n\})) $. Define $\mu_j=\inf(\sigma_{\textup{ess}}(T_j))$ which may be $\infty$. Then
	\begin{equation*}
	\inf(\sigma_{\textup{ess}}(\overline{H}))\geq \min_{j} \left \{ \mu_j+\sum_{\ell\neq i}\lambda_\ell \right \}:=m.
	\end{equation*}
	
	\item[\textup{(5)}] Assume $B_j$ is selfadjoint on $\cH_j$. If $\cD(T_j)\subset \cD(B_j)$ for some $j\in \{1,\dots ,n\}$ then $\cD(H_j)\subset \cD(1\otimes\cdots\otimes B_j \otimes\cdots\otimes 1)$.

	\item[\textup{(6)}] Assume $B_j$ is selfadjoint on $\cH_j$ and $T_j+B_j$ is selfadjoint. Then
	\begin{align*}
	H_j+1\otimes\cdots\otimes B_j \otimes\cdots\otimes 1=1\otimes\cdots\otimes (T_j+B_j) \otimes\cdots\otimes 1:=S_j.
	\end{align*}
	\end{enumerate}
\end{thm}
\begin{proof}
	Statements (1)-(3) can more or less be found in either \cite{Schmudgen} or \cite{Weidmann}. It is proven in \cite{Schmudgen} that $f(H_j)=1\otimes\cdots\otimes f(T_j)\otimes\cdots\otimes 1$ holds for $f(x)=(x\pm i)^{-1}$. From there one may simply use standard approximation arguments. We now prove part $(4)$. Let $P$ be the joint spectral measure of $(H_1, \dots ,H_n)$. The joint spectrum is $\sigma(T_1)\times\cdots\times \sigma(T_n)$ and 
	\begin{align*}
	H=\int_{\sigma(T_1)\times\cdots\times \sigma(T_n)}(x_1+\dots +x_n)dP(x_1,\dots, x_n)
	\end{align*}
	so we see $H$ is selfadjoint and $\lambda=\inf(\sigma(H))$. The formula for $e^{-tH}$ is immediate from the spectral theorem. We also find	
	\begin{align*}
	P_{H}(\lambda)&=P( \{ x_1+\cdots+x_n=\lambda\}\cap \sigma(T_1)\times\cdots\times \sigma(T_n) )\\&=P( \{ x_1=\lambda_1 \}\times\cdots\times  \{ x_n=\lambda_n \} )\\&=P_{T_1}(\{\lambda_1\})\otimes\cdots\otimes P_{T_n}(\{\lambda_n\}).
	\end{align*}
	To finish the proof of part (4) we must prove that $f(H)$ is compact for any $f\in C_c^{\infty}( (-\infty,m) )$. Let $f\in C_c^{\infty}( (-\infty,m) )$ and note that
	\begin{equation*}
	f(H)=\int_{\sigma(T_1)\times\cdots\times \sigma(T_n)}f(x_1+\cdots+x_n)dP(x_1,\dots,x_n).
	\end{equation*} 
	Define
	\begin{align*}
	\cA=\{ (x_1,\dots,x_n)\in \sigma(T_1)\times\cdots\times \sigma(T_n)\mid  \, f(x_1+\cdots+x_n)\neq 0 \}
	\end{align*}
	Let $(x_1,\dots,x_n)\in \cA$. There is $\varepsilon>0$ such that $f$ is supported in $(-\infty,m-\varepsilon)$ and hence $x_j<\mu_j-\varepsilon$ for all $j\in \{1, \dots, n \}$. Therefore
	\begin{align*}
	\cA\subset (\sigma(T_1)\cap (-\infty ,\mu_1-\varepsilon))\times \cdots \times (\sigma(T_n)\cap (-\infty ,\mu_n-\varepsilon)).
	\end{align*}
	Let $j\in \{ 1, \dots , n \}$ and note $\sigma(T_j)\cap (-\infty ,\mu_j-\varepsilon)$ is either empty or contains only finitely many eigenvalues of finite multiplicity. This implies $f(H)$ is a finite linear combination of operators of the form 
	\begin{equation*}
	P_{T_1}(\{x_1\})\otimes\cdots\otimes P_{T_n}(\{x_n\})
	\end{equation*}
	with $x_j$ in the discrete spectrum of $T_j$. The above projection has finite rank and is therefore compact so $f(H)$ is compact.

	To prove part (5) we note $B_j(T_j+i)^{-1}$ is bounded and	
	\begin{equation*}
	(1\otimes\cdots\otimes B_j \otimes\cdots\otimes 1)(H_j+i)^{-1}=1\otimes\cdots\otimes B_j(T_j+i)^{-1} \otimes\cdots\otimes 1
	\end{equation*}
	holds on the span of simple tensors. Thus $(1\otimes\cdots\otimes B_j \otimes\cdots\otimes 1)(H_j+i)^{-1}$ extends to a bounded operator implying the claim.

	To prove part (6), note that the equality holds on $\cH_1\widehat{\otimes}\cdots \widehat{\otimes} \cD(T_j+B_j) \widehat{\otimes} \cdots \widehat{\otimes}\cH_n$ which is a core for $S_j$. Therefore
	\begin{equation}\label{eqvig}
	S_j = \overline{H_j+1\otimes\cdots\otimes B_j \otimes\cdots\otimes 1}.
	\end{equation}
	By part (5) we note $\cD(S_j)\subset \cD(H_j)\cap \cD(1\otimes\cdots\otimes B_j \otimes\cdots\otimes 1)$ so the closure on the right side of (\ref{eqvig}) is not necessary.
\end{proof}

\begin{lem}\label{Lem:SmallObsTensor}
Let  $A$ and $B$ be selfadjoint on $\cH_2$. If $B$ is $A$-bounded with bound $a$, and $C\in B(\cH_1)$ then $C\otimes B$ is $1\otimes A$ bounded with relative bound $a \lVert C \lVert$.
\end{lem}
\begin{proof}
On simple tensors we see that
\begin{equation*}
C\otimes B(1\otimes A-i\lambda)^{-1}=C\otimes B(A-i\lambda)^{-1}
\end{equation*}
which is bounded. Hence $\cD(C\otimes B)\subset \cD(1\otimes A)$ and the above identity extends to the full tensor product. Calculating norms and taking $\lambda$ to $\infty$ gives the $1\otimes A$-bound by \cite[Lemma 6.3]{Teschl}. 
\end{proof}

\noindent We now wish to consider second quantised observables. Let $\omega$ be selfadjoint on the Hilbert space $\cH$. By standard theory of reducing subspaces we have
\begin{align}
\sigma_p(d\Gamma^{(n)}(\omega))\subset& \sigma_p\left( \sum_{k=1}^{n} (1\otimes)^{k-1}\omega (\otimes 1)^{n-k} \right)\label{Redlign1}\\
\label{Redlign2}\sigma_{\textup{ess}}(d\Gamma^{(n)}(\omega))\subset& \sigma_{\textup{ess}}\left( \sum_{k=1}^{n} (1\otimes)^{k-1}\omega (\otimes 1)^{n-k} \right).
\end{align}

\begin{lem}\label{Lem:SecondQuantisedProp}
Let $\omega$ be a selfadjoint and nonnegative operator on the Hilbert space $\cH$. Write $m=\inf(\sigma(\omega))$ and $m_{\textup{ess}}=\inf(\sigma_{\textup{ess}}(\omega))$. For $n\geq 1$ we have
\begin{align*}
\sigma(d\Gamma^{(n)}(\omega) )&= \overline{ \{  \lambda_1+\cdots+\lambda_n\mid \lambda_j\in \sigma(\omega) \}}\\
\inf(\sigma(d\Gamma^{(n)}(\omega)))&= n m.
\end{align*}
Assume in addition that $\omega$ is injective. Then 
\begin{enumerate}
	\item[\textup{(1)}] $0$ is an eigenvalue for $d\Gamma(\omega)$ with multiplicity 1. The eigenspace is spanned by $\Omega$.

	\item[\textup{(2)}] $\inf(\sigma_{\textup{ess}}(d\Gamma^{(n)}(\omega))  )\geq m_{\textup{ess}} +(n-1)m$.

    \item[\textup{(3)}] $d\Gamma(\omega)$ has compact resolvents if and only if this is the case for $\omega$.
    
\end{enumerate}

\end{lem}
\begin{proof}
The statements regarding the spectrum is easy and can be found in e.g, \cite{Lecture}. We prove the statements (1), (2) and (3). 

To prove statement (1), we note that $\Omega$ is an eigenvector as desired. Assume that there exists an eigenvector $\psi$ orthogonal to $\Omega$. We may then assume that there is $n\in \NN$ such that $\psi\in \cH^{\otimes_s n}$ and $\psi$ is an eigenvector for $d\Gamma^{(n)}(\omega)$ with eigenvalue 0. Since $d\Gamma^{(n)}(\omega)\geq nm$ we find $m=0$ and thus $0 \in \sigma(\omega)$ but is not an eigenvalue. By Theorem \ref{Thm:SpectralPropTensor} and equation (\ref{Redlign1}) we find $0$ is not an eigenvalue for $d\Gamma^{(n)}(\omega)$, reaching a contradiction.

Statement (2) follows from Theorem \ref{Thm:SpectralPropTensor} and equation (\ref{Redlign2}).

If $d\Gamma(\omega)$ has compact resolvents then projection onto the one particle subspace shows that $\omega$ has compact resolvents. If $\omega$ has compact resolvents, then $m_{\textup{ess}}=\infty$ and so $m>0$ by injectivity of $\omega$. Statement (2) now gives that $d\Gamma^{(n)}(\omega)$ has compact resolvents for all $n\in \mathbb{N}$. Observe $\lVert (d\Gamma^{(n)}(\omega)+i)^{-1}\lVert\leq \frac{1}{nm}$ which converges to 0 as $n$ tends to $\infty$. Hence we find 
\begin{align*}
(d\Gamma(\omega)+i)^{-1}=\bigoplus_{n=0}^{ \infty }(d\Gamma^{(n)}(\omega)+i)^{-1}
\end{align*}
is compact.
\end{proof}
\section{Isomorphism theorems}
\noindent Let $\cH_1$ and $\cH_2$ be Hilbert spaces. Then
\begin{equation*}
\mathcal{F}_b(\mathcal{H}_1\oplus \mathcal{H}_2)\approx \mathcal{F}_b(\mathcal{H}_1)\otimes \mathcal{F}_b(\mathcal{H}_2)\approx \bigoplus_{n=0}^\infty \left (\mathcal{F}_b(\mathcal{H}_1)\otimes \mathcal{H}_2^{\otimes_s n}\right) .
\end{equation*}
In this chapter we investigate these isomorphisms. See also \cite{DerezinskiGerard} and \cite{Parthasarathy}.
\begin{thm}\label{Thm:ISOTHM1}
	There is a unique isomorphism $U\colon\mathcal{F}_b(\mathcal{H}_1\oplus \mathcal{H}_2)\rightarrow  \mathcal{F}_b(\mathcal{H}_1)\otimes \mathcal{F}_b(\mathcal{H}_2)$ such that $U(\epsilon(f\oplus g))=\epsilon(f)\otimes \epsilon(g)$. If $f_1,\dots,f_j\in \mathcal{H}_1$ and $g_1,\dots,g_\ell\in \mathcal{H}_2$ then
	\begin{align*}
	U((f_1,0)\otimes_s\cdots\otimes_s (f_j,0)\otimes_s &(0,g_1)\otimes_s\cdots\otimes_s (0,g_\ell)) \\&=\left(\frac{j!\ell!}{(j+\ell)!} \right)^{1/2} (f_1\otimes_s \cdots\otimes_s f_j)\otimes  (g_1\otimes_s\cdots\otimes_s g_\ell)
	\end{align*}
	Furthermore, if $A_i$ is selfadjoint on $\mathcal{H}_i$, $V_i$ is unitary on $\mathcal{H}_i$, $f\in \cH_1$ and $g\in \cH_2$ then 
	\begin{align}
	UW(f\oplus g,V_1\oplus V_2)U^*&=W(f,V_1)\otimes W(g,V_2) \label{eq:Trans1}\\
	Ud\Gamma (A_1\oplus A_2)U^*&=\overline{d\Gamma (A_1)\otimes 1 +1\otimes d\Gamma ( A_2)}\label{eq:Trans2}\\
	U\varphi(f,g)U^*&=\overline{\varphi(f)\otimes 1+1\otimes \varphi(g)}\label{eq:Trans3}\\
	Ua(f,g)U^*&=\overline{a(f)\widehat{\otimes} 1+1\widehat{\otimes} a(g)}\label{eq:Trans4}\\
	Ua^{\dagger}(f,g)U^*&=\overline{a^{\dagger}(f)\widehat{\otimes} 1+1\widehat{\otimes} a^{\dagger}(g)}.\label{eq:Trans5}
	\end{align}
\end{thm}
\begin{proof}
The set of exponential vectors is total. This implies, that at most one linear and bounded map can satisfy $U(\epsilon(f\oplus g))=\epsilon(f)\otimes \epsilon(g)$ for all $f\in \cH_1$ and $g\in \cH_2$. By the linear independence of exponential vectors we may define $U(\epsilon(f\oplus g))=\epsilon(f)\otimes \epsilon(g)$ and extend by linearity. Note that the image of $U$ is dense and $U$ conserves the inner product since
\begin{align*}
\langle \epsilon(h_1\oplus h_2),\epsilon(f_1\oplus f_2) \rangle =e^{\langle h_1\oplus h_2,f_1\oplus f_2\rangle}=\langle \epsilon(h_1)\otimes \epsilon(h_2), \epsilon(f_1)\otimes  \epsilon(f_2)\rangle.
\end{align*}
Hence it extends by continuity to a unitary map. To prove equation (\ref{eq:Trans1}) it is enough to check the set of exponential vectors. We calculate
\begin{align*}
UW(f\oplus g,V_1\oplus V_2)U^*  \epsilon(h_1)\otimes \epsilon (h_2)&=U\epsilon(V_1h_1\oplus V_2h_2+f\oplus g)\\&=\epsilon(V_1h_1+f)\otimes \epsilon(V_2h_2+g)\\&=W(f,V_1)\otimes W(g,V_2)(\epsilon(h_1)\otimes \epsilon (h_2)).
\end{align*}
 This also proves equations (\ref{eq:Trans2}) and (\ref{eq:Trans3}) since both sides generate the same unitary group. To prove equations (\ref{eq:Trans4}) and (\ref{eq:Trans5}) we define
 \begin{equation*}
 \cC=\{  \epsilon(f_1)\otimes \epsilon(f_2)\mid  f_i\in \cH_i \}=U(\{ \epsilon(f_1\oplus f_2)\mid f_i\in \cH_i  \}).
 \end{equation*}
For $\psi=\epsilon(f_1)\otimes \epsilon(f_2)\in \cC$ we have
\begin{align*}
Ua(f\oplus g)U^*\epsilon(f_1)\otimes \epsilon(f_2)&=Ua(f\oplus g)\epsilon(f_1\oplus f_2)\\&=\langle f\oplus g,f_1\oplus f_2\rangle \epsilon(f_1)\otimes \epsilon(f_2)\\&=(a(f)\widehat{\otimes} 1+1\widehat{\otimes} a(g))\epsilon(f_1)\otimes \epsilon(f_2)
\end{align*}
showing equation (\ref{eq:Trans4}) is true on $\cC$. For $\phi=\epsilon(g_1)\otimes \epsilon(g_2)\in \cC$ we find
\begin{align*}
\langle \phi ,Ua^\dagger(f\oplus g)U^*\psi  \rangle= \langle Ua(f\oplus g)U^* \phi ,\psi  \rangle=\langle  \phi ,(a^\dagger(f)\widehat{\otimes} 1+1\widehat{\otimes} a^\dagger(g))\psi  \rangle.
\end{align*}
$\cC$ is total in $\cF_b(\cH_1)\otimes \cF_b(\cH_2)$ so $(a^\dagger(f)\widehat{\otimes} 1+1\widehat{\otimes} a^\dagger(g))\psi=Ua^\dagger(f\oplus g)U^*\psi$. This proves equation (\ref{eq:Trans5}) is true on $\cC$. We can now conclude that equations (\ref{eq:Trans4}) and (\ref{eq:Trans5}) hold on $\cC$. Let $\sharp$ denote either $\dagger$ or nothing. Exponential vectors span a core for both creation and annihilation operators (see \cite{Lecture}) so
\begin{align*}
Ua^\sharp(f\oplus g)U^*=\overline{a^\sharp(f)\widehat{\otimes} 1+1\widehat{\otimes} a^\sharp(g)\mid_{\text{Span($\cC$)} }  }
\end{align*}
It is not hard to see that $\overline{a^\sharp(f)\widehat{\otimes} 1+1\widehat{\otimes} a^\sharp(g)\mid_{\text{Span($\cC$)} }  }$ is extends $a^\sharp(f)\widehat{\otimes} 1+1\widehat{\otimes} a^\sharp(g)$ so $\overline{a^\sharp(f)\widehat{\otimes} 1+1\widehat{\otimes} a^\sharp(g)\mid_{\text{Span($\cC$)} }  }=\overline{a^\sharp(f)\widehat{\otimes} 1+1\widehat{\otimes} a^\sharp(g)  }$ proving equations (\ref{eq:Trans4}) and (\ref{eq:Trans5}). We now calculate
\begin{align*}
U(f_1,0)&\otimes_s\cdots\otimes_s (f_j,0)\otimes_s (0,g_1)\otimes_s\cdots\otimes_s (0,g_\ell) \\&=\left(\frac{1}{(j+\ell)!} \right)^{1/2} U a^\dagger(f_1,0)\cdots  a^{\dagger}(f_j,0) a^{\dagger}(0,g_1)\cdots a^{\dagger}(0,g_\ell)\Omega\\&=\left(\frac{\ell!j!}{(j+\ell)!} \right)^{1/2}(f_1\otimes_s \cdots\otimes_s f_j)\otimes(  g_1\otimes_s\cdots\otimes_s g_\ell)
\end{align*}
finishing the proof.
\end{proof}

\noindent The following result is obvious.
\begin{thm}\label{Thm:ISOTHM2}
There is a unique isomorphism
	\begin{equation*}
	U:\mathcal{F}_b(\mathcal{H}_1)\otimes \mathcal{F}_b(\mathcal{H}_2)\rightarrow \mathcal{F}_b(\mathcal{H}_1)\oplus \bigoplus_{n=1}^{\infty} \mathcal{F}_b(\mathcal{H}_1)\otimes \mathcal{H}_2^{\otimes_s n}
	\end{equation*}
	such that
	\begin{equation*}
	U(w \otimes \{\psi_2^{(n)} \}_{n=0}^\infty)=\psi^{(0)}w\oplus \bigoplus \limits_{n=1}^{\infty} w \otimes \psi_2^{(n)}.
	\end{equation*}
	Let $A$ be a selfadjoint operator on $\mathcal{F}_b(\mathcal{H}_1)$ and $B$ be a selfadjoint operator on $\mathcal{F}_b(\mathcal{H}_2)$ such that $B$ is reduced by all of the subspaces $\mathcal{H}_2^{\otimes_s n}$. Write $B^{(n)}=B\mid_{\mathcal{H}_2^{\otimes_s n}}$. Then
	\begin{align*}
	U(A\otimes 1+1\otimes B)U^*&=A+B^{(0)}\oplus \bigoplus_{n=1}^\infty (A\otimes 1+1\otimes B^{(n)}) \\ U (A\otimes B) U^*&= B^{(0)} A \oplus \bigoplus_{n=1}^\infty (A\otimes  B^{(n)}).
	\end{align*}
\end{thm}

\begin{lem}\label{Thm:ISOTHM3}
Let $\cH$ be a Hilbert space and assume there is a unitary map $V:\cH\rightarrow \cH_1\oplus \cH_2$. Let $U_1$ be the map from Theorem \ref{Thm:ISOTHM1}, $U_2$ be the map from Theorem $\ref{Thm:ISOTHM2}$ and $j_i:\cH_i\rightarrow \cH_1\oplus \cH_2$ be the embedding defined by $j_1(f)=(f,0)$ or $j_2(g)=(0,g)$. Define the maps $U=U_2U_1\Gamma(V)$ and $Q_i=V^*j_i$. Then
\begin{equation}\label{eq:SimpelTrans12}
\Gamma(Q_1)=U^*\mid_{\mathcal{F}_b(\mathcal{H}_1)}.
\end{equation}
Let $\cK \subset \cH_1 $ be a subspace and $g_1,\dots,g_q\in \cH_2$. Define  
\begin{align*}
B=& \{ \Omega \}\cup \bigcup_{b=1}^\infty \{ h_1\otimes_s\cdots\otimes_s h_b\mid h_i\in \cK \}\\
C=& \{ Q_2g_1\otimes_s \cdots\otimes_s Q_2g_q \}\\&\cup \bigcup_{b=1}^\infty \{ Q_1h_1\otimes_s\cdots \otimes_s Q_1h_b\otimes_s Q_2g_1\otimes_s \cdots\otimes_s Q_2g_q\mid h_i\in \cK \}.
\end{align*}
Let $\psi\in \textup{Span}(B)$. Then
\begin{equation}\label{eq:SimpelTrans11}
U^*(\psi\otimes g_1\otimes_s \cdots \otimes_s g_q)\in \textup{Span}(C).
\end{equation}
\end{lem}
\begin{proof}
It is enough to prove equation (\ref{eq:SimpelTrans12}) on elements of the form $\epsilon(f)$ for $f\in \cH_1$. We calculate using Theorems \ref{Thm:ISOTHM1} and \ref{Thm:ISOTHM2}:
\begin{equation*}
U^*\epsilon(f)=\Gamma(V)^*U_1^*\epsilon(f)\otimes \Omega=\Gamma(V)^*\epsilon(f,0)=\epsilon(V^*j_1f)=\epsilon(Q_1f)=\Gamma(Q_1)\epsilon(f).
\end{equation*}
By linearity it is enough to prove equation (\ref{eq:SimpelTrans11}) when $\psi\in B$. Using Theorems \ref{Thm:ISOTHM1} and \ref{Thm:ISOTHM2} along with Lemma \ref{Lem:SeconduantisedBetweenSPaces} we find
\begin{align*}
U^*(\Omega \otimes (g_1\otimes_s\cdots\otimes_s g_q))&=\Gamma(V)^*U_1^*(\Omega \otimes (g_1\otimes_s\cdots\otimes_s g_q))\\&=\Gamma(V)^*(g_1\otimes_s\cdots\otimes_s g_q)\\&=Q_2g_1\otimes_s\cdots\otimes_s Q_2g_q\in C
\end{align*}
and
\begin{align*}
U^*((h_1\otimes_s&\cdots\otimes_s h_b) \otimes (g_1\otimes_s\cdots\otimes_s g_q))\\=&\Gamma(V)^*U_1^*((h_1\otimes_s\cdots\otimes_s h_b) \otimes (g_1\otimes_s\cdots\otimes_s g_q))\\=&\left (\frac{(b+q)!}{q!b!}\right)^{1/2}\Gamma(V)^*((h_1,0)\otimes_s\cdots\otimes_s (h_b,0)\otimes_s(0,g_1)\otimes_s\cdots\otimes_s (0,g_q))\\=&\left (\frac{(b+q)!}{q!b!}\right)^{1/2}Q_1h_1\otimes_s\cdots\otimes_s Q_1h_b\otimes_s Q_2g_1\otimes_s\cdots\otimes_s Q_2g_q\in \textup{Span}(C).
	\end{align*}
This finishes the proof.
\end{proof}

\section{Pointwise annihilation operators}
\noindent In this appendix we define pointwise annihilation operators and show associated pull through formulas. Furthermore, we will prove part (1) of Theorem \ref{Thm:Numberstructure_Massless}. Let $\cH=L^2(\cM,\cE,\mu)$, where $(\cM,\cE,\mu)$ is assumed to be $\sigma$-finite. We define the extended symmetric Fock space to be the product
\begin{equation*}
\cF_{+}(\cH)=\bigtimes_{n=0}^{\infty}\cH^{\otimes_s n}
\end{equation*}
with coordinate projections $P_n:\cF_{+}(\cH)\rightarrow \cH^{\otimes_s n}$. For elements $(\psi^{(n)}),(\phi^{(n)})\in \cF_+(\cH)$ we define
\begin{equation*}
d((\psi^{(n)}),(\phi^{(n)}))=\sum_{n=0}^{\infty} \frac{1}{2^n}\frac{\lVert \psi^{(n)}-\phi^{(n)} \lVert  }{1+\lVert \psi^{(n)}-\phi^{(n)} \lVert }
\end{equation*} 
where $\lVert \cdot \lVert$ is the Fock space norm. This makes sense since $P_n(\cF_{+}(\cH))\subset \cF_b(\cH)$. Standard theorems from measure theory and topology now gives the following lemma.
\begin{lem}\label{Lem:BasicTopologyExtSpace}
The map $d$ defines a metric on $\cF_{+}(\cH)$ and turns this space into a complete separable metric space and a topological vector space. Both the topology and the Borel $\sigma$-algebra are generated by the projections $\{P_n \}_{n=0}^\infty$. If a sequence $\{\psi_n\}_{n=1}^\infty \subset \cF_b(\cH)$ is convergent/Cauchy then it is also convergent/Cauchy with respect to $d$ so any total/dense set in $\cF_{b}(\cH)$ will be total/dense in $\cF_{+}(\cH)$.
\end{lem}
\noindent Define
\begin{equation}\label{Defn:A}
\cA=\{ \Omega \}\cup \bigcup_{n=1}^\infty \{ g^{\otimes n}\mid g\in \cH \}.
\end{equation}
Then $\cA$ is total in $\cF_b(\cH)$ since the span of $\cA$ can approximate any exponential vector. By Lemma \ref{Lem:BasicTopologyExtSpace} we find $\cA$ is total in $\cF_{+}(\cH)$ as well. For each $a\in \RR$ we define
\begin{equation*}
\lVert \cdot \lVert_{a,+}=\lim_{n\rightarrow \infty}\left (\sum_{k=0}^{n} (k+1)^{2a} \lVert P_k(\cdot) \lVert^2   \right )^{\frac{1}{2}}.
\end{equation*}
which is measurable from $\cF_+(\cH)$ into $[0,\infty]$. Let
\begin{align*}
\cF_{a,+}(\cH)=\{ \psi\in \cF_{+}(\cH)\mid \lVert \psi \lVert_{a,+}<\infty  \}.
\end{align*}
Note $\cF_{a,+}(\cH)$ is a Hilbert space and $\cF_{a,+}(\cH)=\cD((N+1)^a)$ for $a\geq 0$. We summarise:
\begin{lem}\label{Lem:DenseSetMeasNorm}
$\lVert \cdot \lVert_{a,+}$ defines measurable map from $\cF_+(\cH)$ to $[0,\infty]$ and restricts to a norm on $\cF_{a,+}(\cH)$. Furthermore, $\cF_{a,+}(\cH)$ is a Hilbert space and the set $\cA$ from equation (\ref{Defn:A})  is total in $\cF_+(\cH)$.
\end{lem}
\noindent  The point of defining a metric on $\cF_+(\cH)$ and finding a total subset is that most of the operators we will encounter in this chapter are continuous on $\cF_{+}(\cH)$. Therefore many operator identities can be proven by simply checking the identity holds on $\cA$. Let $v\in \cH$ and define the following maps on $\cF_{+}(\cH)$
\begin{align*}
a_+(v) ( \psi^{(n)} ) &= ( a_{n}(v)\psi^{(n+1)} )\\
a^{\dagger}_+(v) ( \psi^{(n)} ) &= (0, a^\dagger_0(v)\psi^{(0)},a_1^\dagger(v)\psi^{(1)},\dots )\\
\varphi_+(v)  &= a_+(v)+a_+^\dagger(v)
\end{align*}
where $a_n(v)$ is the annihilation operator from $\cH^{\otimes_s (n+1)}$ to $\cH^{\otimes_s n}$ and $a^\dagger_n(f)$ is the creation operator from $\cH^{\otimes_s n}$ to $\cH^{\otimes (n+1)}$ which are both continuous. The following lemma is almost automatic.
\begin{lem}\label{Lem:AnihilCreaFieldOnF_+}
The maps $a_+(v)$, $a^{\dagger}_+(v)$ and $\varphi_+(v)$ are continuous. For $B\in \{ a,a^{\dagger},\varphi \}$ we have
\begin{equation}\label{eq:extphi}
B_+(v)\psi=B(v)\psi \,\,\, \,\,\,  \forall \,\psi\in \cD(B(v)).
\end{equation}
\end{lem}
\begin{proof}
The topology on $\cF_+(\cH)$ is generated by the projections $P_n$ so continuity of $a_+(v)$, $a^{\dagger}_+(v)$ and $\varphi_+(v)$ follows from continuity of 
\begin{align*}
P_na_+(v)&= a_n(v) P_{n+1}  \\   P_na^\dagger_+(v)&= a^\dagger_{n-1}(v) P_{n-1}\,\,\,\, n\geq 1\\ P_0a^\dagger_+(v)&=0.
\end{align*}
Equation (\ref{eq:extphi}) holds for $B\in \{ a,a^{\dagger} \}$ simply by definition. We now prove equation (\ref{eq:extphi}) for $B=\varphi$. The relation
\begin{equation*}
\varphi(v)\psi=\varphi_+(v)\psi
\end{equation*}
is easily seen to hold on the span of $\cA$ where $\cA$ is the set from equation (\ref{Defn:A}). For $\psi\in \cD(\varphi(v))$ we may pick as sequence $\{\psi_n\}_{n=1}^\infty\subset \textup{Span}(\cA)$ that converges to $\psi$ in $\varphi(f)$ norm (use e.g. \cite[Corollary 20.5]{Parthasarathy}). Continuity of $\varphi_+(v)$ together with Lemma \ref{Lem:BasicTopologyExtSpace} now yields the desired result.
\end{proof}
Let $U$ be unitary on $\cH$ and $\omega=(\omega_1,\dots,\omega_p)$ be a tuple of strongly commuting selfadjoint operators on $\cH$. We then define
\begin{align*}
d\Gamma(\omega)&=(d\Gamma(\omega_1),\dots,d\Gamma(\omega_p))\\
d\Gamma^{(n)}(\omega)&=(d\Gamma^{(n)}(\omega_1),\dots,d\Gamma^{(n)}(\omega_p))
\end{align*}
which are tuples of strongly commuting selfadjoint operators (it is easily checked that the unitary groups commute). Let $f:\RR^p\rightarrow \CC$ be a measurable map and define
\begin{align*}
f(d\Gamma_+(\omega))&=\bigtimes_{n=0}^{\infty}f(d\Gamma^{(n)}(\omega)) \,\,\,\, \cD(f(d\Gamma_+(\omega)))=\bigtimes_{n=0}^{\infty}\cD(f(d\Gamma^{(n)}(\omega)))  \\
\Gamma_+(U)&=\bigtimes_{n=0}^{\infty} \Gamma^{(n)}(U).
\end{align*}
If $\omega:\cM\rightarrow \RR^p$ is measurable then we may identify $\omega$ as a tuple of strongly commuting selfadjoint operators. In this case, $f(d\Gamma^{(n)}(\omega))$ is multiplication by the map $f(\omega(k_1)+\cdots+\omega(k_n))$. The following lemma is obvious.
\begin{lem}\label{Lem:2ndQuantisedF_+}
The map $\Gamma_+(U)$ is an isometry on $\cF_+(\cH)$. Furthermore,
\begin{align*}
f(d\Gamma_+(\omega))\psi&=f(d\Gamma(\omega))\psi \,\,\,\,\, \forall \, \psi\in \cD(f(d\Gamma(\omega)))  \\
\Gamma_+(U)\psi&=\Gamma(U)\psi \,\,\,\,\,\,\,\,\,\,\,\,\,\,\,\,\, \forall \, \psi\in \cF_{b}(\cH)
\end{align*}
\end{lem}
\noindent We will now consider a class of linear functionals on $\cF_+(\cH)$. For each $n \in \mathbb{N}$ we let $Q_n:\cF_{+}(\cH)\rightarrow \cN$ be the linear projection which preserves the first $n$ entries of $(\psi^{(k)})$ and projects the rest of them to 0. For $\psi\in \cN$ there is $K\in \mathbb{N}$ such that $Q_n\psi=\psi$ for $n\geq K$. For $\phi\in \cF_+(\cH)$ we may thus define the pairing
\begin{equation}\label{eq:extinner}
\langle \psi,\phi \rangle_+:=\langle \psi,Q_n\phi \rangle=\sum_{k=0}^{K} \langle \psi^{(k)}, \phi^{(k)} \rangle,
\end{equation}
where $n\geq K$.
\begin{lem}\label{Lem:SeperatingForm}
The map $Q_n$ is linear and continuous into $\cF_b(\cH)$. The pairing $\langle \cdot, \cdot \rangle_+$ is sesquilinear and continuous in the second entry. If $\phi\in \cF_{a,+}(\cH)$ then $\psi\mapsto \langle \psi,\phi \rangle_+$ is continuous with respect to $\lVert \cdot \lVert_{-a,+}$. Furthermore, the collection of maps $\{\langle \psi, \cdot \rangle_+\}_{\psi\in \cN}$ will separate points of $\cF_+(\cH)$.
\end{lem}
\begin{proof}
The pairing $\langle \cdot, \cdot \rangle_+$ is trivially sesquilinear. Let $\{ \psi_j \}_{j=1}^\infty$ converge to $\psi$ in $\cF_+(\cH)$. Then $\{\psi^{(k)}_j\}_{j=1}^\infty$ will converge to $\psi^{(k)}$ for all $k\in \NN$ so
\begin{equation*}
\lVert Q_n(\psi_j-\psi) \lVert^2=\sum_{k=0}^{n}\lVert \psi_j^{(k)}-\psi^{(k)} \lVert^2  
\end{equation*}
converges to 0. Hence $Q_n$ is continuous from $\cF_+(\cH)$ into $\cF_b(\cH)$. This also shows continuity in the second entry of $\langle\cdot,\cdot\rangle_+$. If $\phi\in \cF_{a,+}(\cH)$ and $\psi\in \cN$ we find some $K\in \mathbb{N}$ such that
\begin{equation*}
\lvert \langle \psi,\phi \rangle_+\lvert\leq \sum_{k=0}^{K} (k+1)^a\lVert \phi^{(k)} \lVert (k+1)^{-a}\lVert \psi^{(k)} \lVert\leq \lVert \phi\lVert_{a,+}\lVert \psi\lVert_{-a,+}
\end{equation*}
showing the desired continuity. Let $\phi=(\phi^{(k)})\in \cF_+(\cH)$ and assume that $\langle \psi, \phi \rangle_+=0$ for all $\psi\in \cN$. Then $\langle  \psi,\phi^{(k)} \rangle=0$ for all $\psi \in \cH^{\otimes_s k}$ and $k\in \NN_0$ showing $\phi=0$.
\end{proof}
\begin{cor}\label{Cor:seperatingDense}
Let $a\leq 0$, $\phi\in \cF_{a,+}(\cH)$, $\cD\subset \cN$ be dense in $\cF_b(\cH)$ and assume $\langle \psi, \phi \rangle_+=0$ for all $\psi\in \cD$. Then $\phi=0$.
\end{cor}
\begin{proof}
Note $\cD$ consists of elements which are analytic for $(N+1)^{-a}$ so $\cD$ is a core for $(N+1)^{-a}$. Let $\psi$ in $\cN$ and pick $\{\psi_n\}_{n=1}^\infty\subset \cD $ converging to $\psi$ in $(N+1)^{-a}$-norm. Using Lemma \ref{Lem:SeperatingForm} we see $\langle \psi, \phi \rangle_+=\lim_{n\rightarrow \infty}\langle \psi_n, \phi \rangle_+=0$ and thus $\phi=0$ by Lemma \ref{Lem:SeperatingForm}.
\end{proof}

\begin{lem}\label{Lem:FormalAdjoints}
Let $\psi\in \cN$, $\phi\in \cF_+(\cH)$, $v\in \cH$ and $U$ be a unitary operator on $\cH$. Then
\begin{align*}
\langle a^{\dagger}(v)\psi,\phi \rangle_+&=\langle \psi,a_+(v) \phi \rangle_+,\,\,\,\,\,\,\,\,\,\,\,\,\,\,\, \langle a(v)\psi,\phi \rangle_+=\langle \psi,a_+^{\dagger}(v) \phi \rangle_+,\\ \langle \varphi(v)\psi,\phi \rangle_+&=\langle \psi,\varphi_+(v) \phi \rangle_+,\,\,\,\,\,\,\,\,\,\,\,\,\,\,\, \langle \Gamma(U)\psi,\phi \rangle_+=\langle \psi,\Gamma_+(U^*)\phi \rangle_+.
\end{align*}
Let $\omega=(\omega_1,\dots,\omega_p)$ be a tuple of commuting selfadjoint operators, $f:\RR^p\rightarrow \CC$ be measurable, $\psi \in \cN\cap \cD( f(d\Gamma(\omega))  )$ and $\phi\in \cD(\overline{f}(d\Gamma_+(\omega))  )$. Then we have
\begin{equation*}
\langle f(d\Gamma(\omega))\psi,\phi \rangle_+=\langle \psi,\overline{f}(d\Gamma_+(\omega))\phi \rangle_+.
\end{equation*}
\end{lem}
\begin{proof}
Since $\psi\in \cN$ we may pick $K$ such that $\psi^{(n)}=0$ for all $n\geq K$. We calculate
\begin{align*}
\langle a^{\dagger}(v)\psi,\phi \rangle_+&=\langle \psi,a(v)Q_{K+1}\phi \rangle=\langle \psi,Q_{K}a_+(v)\phi \rangle=\langle \psi,a_+(v)\phi \rangle_+\\
\langle a(v)\psi,\phi \rangle_+&=\langle \psi,a^{\dagger}(v)Q_{K-1}\phi \rangle=\langle \psi,Q_{K}a^{\dagger}_+(v)\phi \rangle=\langle \psi,a^{\dagger}_+(v)\phi \rangle_+\\
\langle \varphi(v)\psi,\phi \rangle_+&=\langle \psi,a_+(v)\phi \rangle_++\langle \psi,a^{\dagger}_+(v)\phi \rangle_+=\langle \psi,\varphi_+(v)\phi \rangle_+
\\
\langle \Gamma(U)\psi,\phi \rangle_+&=\langle \psi,\Gamma(U^*)Q_K\phi \rangle=\langle \psi,Q_K\Gamma_+(U^*)\phi \rangle=\langle \psi,\Gamma_+(U^*)\phi \rangle_+
\end{align*}
Assume now that $\psi \in \cN\cap \cD( f(d\Gamma(\omega))  )$ and $\phi\in \cD( \overline{f}(d\Gamma_+(\omega))  )$. Then $Q_K\phi \in \cD(\overline{f}(d\Gamma(\omega))) $ and
\begin{equation*}
\langle f(d\Gamma(\omega))\psi,\phi \rangle_+=\langle \psi,\overline{f}(d\Gamma(\omega))Q_K\phi \rangle=\langle \psi,Q_K\overline{f}(d\Gamma_+(\omega))\phi \rangle=\langle \psi,\overline{f}(d\Gamma_+(\omega))\phi \rangle_+.
\end{equation*}
This finishes the proof.
\end{proof}
\noindent We now consider functions with values in $\cF_+(\cH)$. Let $(X,\cX,\nu)$ be a $\sigma$-finite measure space. Define the quotient
\begin{equation*}
\cM(X,\cX,\nu)=\{f:X\rightarrow  \cF_+(\cH)\mid f \,\, \text{is}\,\, \cX-\cB(\cF_+(\cH))\,\, \text{mesurable}  \}/\sim,
\end{equation*}
where $f\sim g\iff f=g$ almost everywhere. We are interested in the subspace
\begin{equation*}
\cC(X,\cX,\nu)=\{ f\in \cM(X,\cX,\nu)\mid x\mapsto P_nf(x)\in L^2(X,\cX,\nu,\cH^{\otimes_s n}) \,\, \forall \, n\in \NN_0 \}.
\end{equation*}
Lemma \ref{Lem:DenseSetMeasNorm} shows that $x\mapsto \lVert f(x)\lVert_{a,+}$ is measurable for all $f\in \cC(X,\cX,\nu)$ and so
\begin{equation*}
\int_{X}\lVert f(x)\lVert_{a,+}^2 d\nu(x)
\end{equation*}
always makes sense as an element in $[0,\infty]$. If $a=0$ then the integral is finite if and only if $f\in L^2(X,\cX,\nu,\cF_{b}(\cH))$. We write $f\in \cC(X,\cX,\nu)$ as $(f^{(n)})$ where $f^{(n)}=x\mapsto P_nf(x)$ is an element in the Hilbert space $L^2(X,\cX,\nu,\cH^{\otimes_s n})$. For $f,g\in \cC(X,\cX,\nu)$ we define
\begin{equation*}
d(f,g)=\sum_{n=0}^{\infty} \frac{1}{2^n}\frac{\lVert f^{(n)}-g^{(n)} \lVert  }{1+\lVert f^{(n)}-g^{(n)} \lVert  }.
\end{equation*} 
and $\cP_nf=f^{(n)}$ for $n\in \NN_0$. The following result is obvious.
\begin{lem}\label{Lem:TheFundamentalMeasureSpace}
$d$ is a complete metric on $\cC(X,\cX,\nu)$ such that $\cC(X,\cX,\nu)$ becomes separable topological vector space and the topology is generated by $\{ \cP_n \}_{n=0}^\infty$. Furthermore, $L^2(X,\cX,\nu,\cF_{b}(\cH))\subset \cC(X,\cX,\nu)$ and convergence in $L^2(X,\cX,\nu,\cF_{b}(\cH))$ implies convergence in $\cC(X,\cX,\nu)$. Also, the map $x\mapsto \lVert f(x)\lVert_{a,+}$ is measurable for any $f$ in $\cC(X,\cX,\nu)$ and $a\in \RR$.
\end{lem}
\noindent In the last part of this section we will need some results from the theory of direct integrals. Readers are assumed to have the basic knowledge found in \cite[page 280-286]{RS4}. Let $\ell \in \NN$, $v\in \cH$, $U$ be unitary on $\cH$, $\omega=(\omega_1,\dots,\omega_p)$ a tuple of selfadjoint multiplication operators on $\cH$, $m:\cM^\ell \rightarrow \RR^p$ measurable and $g:\RR^p\rightarrow \RR$ a measurable map. Then we wish to define operators on $\cC(\cM^\ell,\cE^{\otimes \ell},\mu^{\otimes \ell}  )$ by
\begin{align*}
(a^{\dagger}_{\oplus,\ell}(v)f)(k)&=a^{\dagger}_+(v)f(k)\\
(a_{\oplus,\ell}(v)f)(k)&=a_+(v)f(k)\\
(\varphi_{\oplus,\ell}(v)f)(k)&=\varphi_+(v)f(k)\\
(\Gamma_{\oplus,\ell}(U)f)(k)&=\Gamma_+(U)f(k)\\
(g(d\Gamma_{\oplus,\ell}(\omega)+m)f)(k)& =g(d\Gamma_+(\omega)+m(k))f(k).
\end{align*} 
We further define $\cC(\cM^0,\cE^{\otimes 0},\mu^{\otimes 0}  )=\cF_+(\cH)$ along with $a^{\dagger}_{\oplus,0}(v)=a^{\dagger}_{+}(v)$, $a_{\oplus,0}(v)=a_{+}(v)$, $\varphi_{\oplus,0}(v)=\varphi_{+}(v)$ and $\Gamma_{\oplus,0}=\Gamma_{+}(U)$.
\begin{lem}\label{Lem:LiftToIntegralOperators}
The following holds
\begin{enumerate}
	\item [\textup{(1)}] The operators $a^{\dagger}_{\oplus,\ell}(v)$, $a_{\oplus,\ell}(v)$, $\varphi_{\oplus,\ell}(v)$ and $\Gamma_{\oplus,\ell}(U)$ are well defined and continuous for all $\ell\in \NN_0$. 
	
	\item [\textup{(2)}] Let $f\in \cC(\cM^\ell,\cE^{\otimes \ell},\mu^{\otimes \ell}  )$. If $f(k)\in \cD(g(d\Gamma_+(\omega)+m(k)))$ for all $k\in \cM^\ell$ then $k\mapsto g(d\Gamma_+(\omega)+m(k))f(k)$ is measurable. Therefore we may define the domain of  $g(d\Gamma_{\oplus,\ell}(\omega)+m)$ to be
	\begin{align*}
	\biggl \{ f\in \cC(\cM^\ell,\cE^{\otimes \ell},\mu^{\otimes \ell} )&\biggl \lvert  f(k)\in \cD(g(d\Gamma_+(\omega)+m(k)))\,\, \text{for almost every} \,\, k\in \cM^\ell,\\& \int_{\cM^\ell} \lVert g(d\Gamma^{(n)}(\omega)+m(k))(\cP_nf)(k) \lVert^2 d\mu^{\otimes \ell}(k)<\infty\,\, \forall n\in \NN_0  \biggl \}.
	\end{align*}

\end{enumerate}
\end{lem}
\begin{proof}
Let $h\in \cC(\cM^\ell,\cE^{\otimes \ell},\mu^{\otimes \ell}  )$. Then
\begin{align*}
k\mapsto P_na_+(v)h(k)&=k\mapsto \left(\int_{\cM^\ell}^{\oplus} a_n(v) d\mu^{\otimes \ell} \cP_{n+1}h\right )(k)  \\ k\mapsto  P_na^\dagger_+(v)h(k)&= k\mapsto\left(\int_{\cM^\ell}^{\oplus}a^\dagger_{n-1}(v)d\mu^{\otimes \ell} \cP_{n-1}h\right)(k)\,\,\,\,\,\,\,\,\,\,\, n\geq 1\\ k\mapsto P_0a^\dagger_+(v)h(k)&=k\mapsto 0\\
k\mapsto P_n\Gamma_+(U)h(k)&=k\mapsto\left(\int_{\cM^\ell}^{\oplus}\Gamma^{(n)}(U)d\mu^{\otimes \ell} \cP_{n}h\right)(k).
\end{align*}
Part (1) is now a consequence of Lemma \ref{Lem:TheFundamentalMeasureSpace}. For the next claim we note
\begin{equation*}
k\mapsto P_ng(d\Gamma_+(\omega)+m(k))f(k)=g(d\Gamma^{(n)}(\omega)+m(k))(\cP_nf)(k).
\end{equation*}
$d\Gamma^{(n)}(\omega_i)+m_i(k)$ is strongly resolvent measurable for each $i\in \{1,\dots,p\}$ which implies $g(d\Gamma^{(n)}(\omega)+m(k))$ is strongly resolvent measurable so the conclusion follows from standard theorems (See e.g \cite[Theorem XIII.85]{RS4}).
\end{proof}
We will now introduce the pointwise annihilation operators. Let $\psi=(\psi^{(n)}) \in \cF_+(\cH)$. Note that $\psi^{(n+\ell)}$ is symmetric and square integrable for all $n\in \NN_0$ and $\ell\in \NN$, so we may pick a representative such that 
\begin{align*}
g(k_1,\dots,k_\ell)=\psi^{(n+\ell)}(k_1,\dots,k_\ell,\cdot,\dots,\cdot)
\end{align*}
is symmetric in $(k_1,\dots,k_\ell)$ and takes values in $\cH^{\otimes_s n}$. It is easy to see, that the choice of representative only changes $g$ up to a $\mu^{\otimes \ell}$ zeroset and that $g\in L^2(\cM^{\ell},\cE^{\otimes \ell},\mu^{\otimes \ell},\cH^{\otimes_s n})$. Therefore we may define $A_\ell \psi\in \cC(\cM^\ell,\cE^{\otimes \ell},\mu^{\otimes \ell}  )$ by the formula
\begin{align*}
\cP_n(A_\ell \psi)=\sqrt{(n+\ell)\cdots (n+1)}\psi^{(n+\ell)}(k_1,\dots,k_\ell,\cdot,\dots,\cdot)
\end{align*}
Note $A_\ell$ is a map from $\cF_+(\cH)$ to $\cC(\cM^\ell,\cE^{\otimes \ell},\mu^{\otimes \ell}  )$ and
\begin{equation*}
\lVert \cP_n(A_\ell\psi)-\cP_n(A_\ell\phi)\lVert=\sqrt{(n+\ell)(n+\ell-1)\cdots(n+1)}\lVert \psi^{(n+\ell)}-\phi^{(n+\ell)}\lVert 
\end{equation*}
for $\psi,\phi\in \cF_+(\cH)$. So $A_\ell$ is continuous from $\cF_+(\cH)$ into $\cC(\cM^\ell,\cE^{\otimes \ell},\mu^{\otimes \ell}  )$. One observes that $A_\ell\psi\in L^2(\cM^{\ell},\cE^{\otimes \ell},\mu^{\otimes \ell},\cF_{a,+}(\cH))$ if and only if
\begin{align*}
 \infty &> \int_{\cM^\ell} \lVert A_\ell\psi(k_1,\dots,k_\ell) \lVert_{a,+}^2 d\mu^{\otimes \ell}(k_1,\dots,k_\ell)\\&= \sum_{n=0}^{\infty}(n+1)^{2a}(n+\ell)(n+\ell-1)\cdots(n+1)\lVert \psi^{(n+\ell)} \lVert^2\\&\iff 
 \sum_{n=0}^{\infty}(n+\ell)^{2a+\ell}\lVert \psi^{(n+\ell)}\lVert^2<\infty,
\end{align*}
which is equivalent to $\psi\in \cD(N^{\frac{\ell}{2}+a})$ if $\frac{\ell}{2}+a\geq 0$. If $\psi,\phi\in \cD(N^{\frac{\ell}{2}})$ we apply the above calculations with $a=0$ to obtain
\begin{align}\label{eq:CentralNumerEstimate}
\lVert A_\ell\psi-A_\ell\phi\lVert^2  &=\sum_{n=0}^{\infty}(n+\ell)(n+\ell-1)\cdots(n+1)\lVert \psi^{(n+\ell)}-\phi^{(n+\ell)}\lVert^2\\&\leq  \lVert N^{\frac{\ell}{2}}(\psi-\phi)\lVert.\nonumber
\end{align}
We summarise:
\begin{lem}\label{Lem:NumberEstimates}
$A_\ell$ is a continuous linear map from $\cF_+(\cH)$ to $\cC(\cM^\ell,\cE^{\otimes \ell},\mu^{\otimes \ell}  )$ and $A_\ell\psi\in L^2(\cM^{\ell},\cE^{\otimes \ell},\mu^{\otimes \ell},\cF_{-\frac{\ell}{2},+}(\cH))$ if $\psi \in \cF_b(\cH)$. Furthermore, $A_\ell$ maps $\cD(N^{\frac{\ell}{2}})$ continuously into $L^2(\cM^{\ell},\cE^{\otimes \ell},\mu^{\otimes \ell},\cF_b(\cH))$ and $\psi\in \cD(N^{\ell/2})$ if and only if $A_\ell\psi\in L^2(\cM^{\ell},\cE^{\otimes \ell},\mu^{\otimes \ell},\cF_b(\cH))$.
\end{lem}
\noindent Fix $v\in \cH$ and $\ell\in \mathbb{N}_0$. Define  $z_\ell(v):\cC(\cM^\ell,\cE^{\otimes \ell},\mu^{\otimes \ell})\rightarrow \cC(\cM^{\ell+1},\cE^{\otimes (\ell+1)},\mu^{\otimes (\ell+1)})$ by 
\begin{equation*}
(z_{0}(v)\psi)(x)=v(x)\psi \,\,\, \text{and}\,\,\,  
(z_{\ell}(v)\psi)(x,k)= v(x)\psi(k)
\end{equation*}
when $\ell\geq 1$. Note that
\begin{equation*}\label{contz}
\int_{\cM^{\ell+1}} \lVert P_n(z_{\ell}(v)\psi)(k)\lVert^2 d\mu^{\otimes (\ell+1)}(k)=\lVert v\lVert^2\lVert \cP_n\psi \lVert^2
\end{equation*}
which implies $z_{\ell}(v)$ is well defined and continuous. A similar argument (which is left to the reader) shows that $z_{\ell}(v)$ maps $L^2(\cM^\ell,\cE^{\otimes \ell},\mu^{\otimes \ell},\cF_b(\cH))$ continuously into $L^2(\cM^{\ell+1},\cE^{\otimes (\ell+1)},\mu^{\otimes (\ell+1)},\cF_b(\cH))$. We summarise:
\begin{lem}\label{Lem:ContOfMultOperator}
The map $z_{\ell}(v)$ introduced above is linear and continuous. Both as a map from $\cC(\cM^\ell,\cE^{\otimes \ell},\mu^{\otimes \ell})$ into the space $\cC(\cM^{\ell+1},\cE^{\otimes (\ell+1)},\mu^{\otimes (\ell+1)})$ and from $L^2(\cM^\ell,\cE^{\otimes \ell},\mu^{\otimes \ell},\cF_b(\cH))$ into  $L^2(\cM^{\ell+1},\cE^{\otimes (\ell+1)},\mu^{\otimes (\ell+1)},\cF_b(\cH))$. 
\end{lem}
\noindent Let $\ell\in \NN$ and $\sigma\in \cS_\ell$ where $\cS_\ell$ is the set of permutations of $\{ 1,\dots,\ell \}$. Define $\widetilde{\sigma}:\cM^\ell\rightarrow \cM^\ell$ by $\widetilde{\sigma}(k_1,\dots,k_\ell)=(k_{\sigma(1)},\dots,k_{\sigma(\ell)})$ and observe that $\widetilde{\sigma}$ is $\cE^{\otimes \ell}$-$\cE^{\otimes \ell}$ measurable. Define $\widehat{\sigma}:\cC(\cM^\ell,\cE^{\otimes \ell},\mu^{\otimes \ell})\rightarrow \cC(\cM^\ell,\cE^{\otimes \ell},\mu^{\otimes \ell})$ by
\begin{equation*}
(\widehat{\sigma}f)(k_1,\dots,k_\ell)=f(k_{\sigma(1)},\dots,k_{\sigma(\ell)})=(f\circ \widetilde{\sigma})(k_1,\dots,k_\ell).
\end{equation*}
$\widehat{\sigma}$ is a well defined isometry on $\cC(\cM^\ell,\cE^{\otimes \ell},\mu^{\otimes \ell})$ since $\widetilde{\sigma}$ is measurable and $\mu^{\otimes \ell}=\mu^{\otimes \ell}\circ \widetilde{\sigma}^{-1}$ so
\begin{equation*}
\int_{\cM^{\ell}} \lVert f^{(n)}(k_1,\dots,k_\ell)\lVert^2 d\mu^{\otimes \ell}(k)=\int_{\cM^{\ell}} \lVert f^{(n)}(k_{\sigma(1)},\dots,k_{\sigma(\ell)})\lVert^2 d\mu^{\otimes \ell}(k).
\end{equation*}
A similar calculation shows that $\widehat{\sigma}$ is isometric on $L^2(\cM^\ell,\cE^{\otimes \ell},\mu^{\otimes \ell},\cF_b(\cH))$. For $\pi\in \cS_\ell$ we have
\begin{equation*}
\widehat{\sigma}\widehat{\pi}f=f\circ \widetilde{\pi}\circ \widetilde{\sigma}=f\circ \widetilde{(\sigma \circ \pi)}=\widehat{\sigma\circ \pi}f
\end{equation*}
and hence the inverse map of $\widehat{\sigma}$ is $\widehat{\sigma^{-1}}$. Define now
\begin{equation*}
S_\ell:=\frac{1}{(\ell-1)!}\sum_{\sigma \in \cS_\ell} \widehat{\sigma}.
\end{equation*}
For $\pi \in \cS_\ell$ we have
\begin{equation*}
\widehat{\pi}S_\ell=\frac{1}{(\ell-1)!}\sum_{\sigma \in \cS_\ell} \widehat{\pi}\widehat{\sigma}=\frac{1}{(\ell-1)!}\sum_{\sigma \in \cS_\ell} \widehat{\pi\circ \sigma}=S_\ell,
\end{equation*}
so $S_\ell^2=\ell S_\ell$. We summarise:
\begin{lem}
Let $\ell\in \NN$. For $\sigma\in \cS_\ell$ the map $\widehat{\sigma}$ defines a linear bijective isometry from $\cC(\cM^\ell,\cE^{\otimes \ell},\mu^{\otimes \ell})$ to $\cC(\cM^{\ell},\cE^{\otimes \ell},\mu^{\otimes \ell})$ and from $L^2(\cM^\ell,\cE^{\otimes \ell},\mu^{\otimes \ell},\cF_b(\cH))$ to $L^2(\cM^\ell,\cE^{\otimes \ell},\mu^{\otimes \ell},\cF_b(\cH))$.

$S_\ell$ is continuous and linear from $\cC(\cM^\ell,\cE^{\otimes \ell},\mu^{\otimes \ell})$ into $\cC(\cM^\ell,\cE^{\otimes \ell},\mu^{\otimes \ell})$ and from  $L^2(\cM^\ell,\cE^{\otimes \ell},\mu^{\otimes \ell},\cF_b(\cH))$ into $L^2(\cM^\ell,\cE^{\otimes \ell},\mu^{\otimes \ell},\cF_b(\cH))$. Furthermore, $S_\ell^2=\ell S_\ell$.
\end{lem}
\noindent We can now calculate commutators

\begin{lem}\label{Lem:CommutatingPointwiseAnihilation}
Let $\omega:\cM\rightarrow \RR^p$ be measurable, $v,g\in \cH$ and $f:\RR^p\rightarrow \RR$ be measurable. Define $A_0$ as the identity map from $\cF_+(\cH)$ to $\cF_+(\cH)$ and $z^\dagger_{\ell-1}(v)=S_\ell z_{\ell-1}(v)$ for $\ell\in \NN$. Then:
\begin{enumerate}
\item [\textup{(1)}] We have the following operator identities for $\ell\in \NN$, $k\in \NN_0$ and $B\in \{ \varphi,a^\dagger, a \}$  
\begin{align}
a_{\oplus,\ell}(g)A_\ell&=A_\ell a_{+}(g) \label{COM1} \\ A_\ell a^\dagger_{+}(g)-a^\dagger_{\oplus,\ell}(g)A_\ell&=z^\dagger_{\ell-1}(g)A_{\ell-1} \label{COM2} \\
B_{\oplus,\ell}(g)z_{\ell-1}(v) &=z_{\ell-1}(v)B_{\oplus,\ell-1}(g) \label{COM3}\\
B_{\oplus,\ell}(g)S_{\ell}&=S_{\ell}B_{\oplus,\ell}(g) \label{COM4}\\ A_\ell\varphi_{+}(g)^k&=\sum_{q=0}^{\min\{\ell,k \}} \begin{pmatrix}
k \\ q
\end{pmatrix}  \left(\prod_{c=0}^{q-1} z^\dagger_{\ell-c-1}(g)\right)\varphi_{\oplus,\ell-q}(g)^{k-q} A_{\ell-q} \label{COM5}\\
\Gamma_{\oplus,\ell}(-1)A_\ell &=(-1)^\ell A_\ell\Gamma_+(-1). \label{COM6}
\end{align}

\item [\textup{(2)}] Let $\ell\in \NN$ and $\omega_\ell(k_1,\dots,k_\ell)=\omega(k_1)+\cdots+\omega(k_\ell)$. If $\psi\in \cD(f(d\Gamma(\omega)))$ then $A_\ell\psi\in \cD(f(d\Gamma_{\oplus}(\omega)+\omega_\ell) )$ and
\begin{equation*}
f(d\Gamma_{\oplus}(\omega)+\omega_\ell) A_\ell\psi=A_\ell f(d\Gamma_{+}(\omega))\psi.
\end{equation*}
\end{enumerate}
\end{lem}
\begin{proof}
We start by proving part (1). Note it is enough to prove the equations (\ref{COM1})-(\ref{COM6}) on the set $\cA$ from equation (\ref{Defn:A}) by continuity and linearity of all involved operators.

Let $h^{\otimes n}\in \cA$ (with $h^{\otimes 0}=\Omega$). We start by proving (\ref{COM1}). If $n<\ell+1$ then $a_{\oplus,\ell}(v)A_\ell h^{\otimes n}=0=A_\ell a_{+}(v)h^{\otimes n}$. Otherwise we calculate
\begin{align*}
(a_{\oplus,\ell}(g)A_\ell h^{\otimes n})(k_1,\dots,k_\ell)&=\sqrt{n(n-1)\dots(n-\ell)}h(k_1)\dots h(k_\ell)\langle g,h \rangle h^{\otimes n-\ell-1}\\&= (A_\ell a_{+}(g)h^{\otimes n})(k_1,\dots,k_\ell).
\end{align*}
We now prove equation (\ref{COM2}). If $n<\ell-1$ we find 
\begin{equation*}
A_\ell a^\dagger_{+}(g)h^{\otimes n}=0=a^\dagger_{\oplus,\ell}(g)A_\ell h^{\otimes n}=0=z^\dagger_{\ell-1}(g)A_{\ell-1}h^{\otimes n}
\end{equation*}
If $n\geq \ell-1$ we have (in the following calculation we define $ h^{\otimes -1}=0$)
\begin{align*}
(A_\ell a^\dagger_{+}(g)&h^{\otimes n})(k_1,\dots,k_\ell)=A_\ell\left (\frac{1}{\sqrt{n+1}}\sum_{a=1}^{n+1} h^{\otimes a-1}\otimes g\otimes h^{\otimes n-a+1}\right )(k_1,\dots,k_\ell)\\&= \sqrt{n(n-1)\dots(n-\ell+2)} \sum_{a=1}^{\ell} h(k_1)\dots g(k_a)\dots h(k_\ell)  h^{\otimes n-\ell+1}\\&+\sqrt{n(n-1)\dots(n-\ell+2)} \sum_{a=\ell+1}^{n+1} h(k_1)\dots h(k_\ell)  h^{\otimes a-1-\ell}\otimes g\otimes h^{\otimes n-a+1}  \\&= \sum_{a=1}^{\ell} g(k_a)(A_{\ell-1}h^{\otimes n})(k_1,\dots,\widehat{k}_a,\dots,k_\ell)\\&+\sqrt{n(n-1)\dots(n-\ell+1)}h(k_1)\dots h(k_\ell)a^{\dagger}_+(g)h^{\otimes n-\ell}\\&
=(S_\ell z_{\ell-1}(g)A_{\ell-1} h^{\otimes n})(k_1,\dots,k_\ell)+(a^\dagger_{\oplus,\ell}(g)A_\ell h^{\otimes n})(k_1,\dots,k_\ell).
\end{align*}
We now prove equation (\ref{COM6}). If $n< \ell$ we have $(-1)^\ell\Gamma_{\oplus,\ell}(-1)A_\ell h^{\otimes n}=0=A_\ell \Gamma_+(-1)h^{\otimes n}$. Writing $k=(k_1,\dots,k_\ell)$ we obtain for $n\geq \ell$
\begin{align*}
(-1)^\ell\Gamma_{\oplus,\ell}(-1)A_\ell h^{\otimes n}(k)&=(-1)^\ell\sqrt{n(n-1)\dots(n-\ell+1)}h(k_1)\dots h(k_\ell)(-h)^{\otimes n-\ell}\\&=A_\ell(-h)^{\otimes n}(k)\\&=A_\ell\Gamma_+(-1)h^{\otimes n}(k).
\end{align*}
We now prove equation (\ref{COM3}) and (\ref{COM4}). Let $\psi\in C(\cM^{\ell},\cE^{\otimes \ell},\mu^{\otimes \ell})$ and $\sigma\in \cS_{\ell}$. Then
\begin{align}\label{eq:proof1}
(B_{\oplus,\ell}(g) \widehat{\sigma}\psi)(k)&=B_+(g)(\psi\circ \widetilde{\sigma})(k)=(\widehat{\sigma}B_{\oplus,\ell}(g)\psi)(k)\\
(B_{\oplus,\ell+1}(g) z_{\ell}(v)\psi)(x,k)&=B_+(g)v(x)\psi(k)=(z_{\ell}(v)B_{\oplus,\ell}(g)\psi)(x,k)\label{eq:proof2}
\end{align}
Equation (\ref{eq:proof1}) shows equation (\ref{COM3}) and equation (\ref{eq:proof2}) shows equation (\ref{COM4}) in the special case where $\ell\geq 2$. The $\ell=1$ case is similar and is left to the reader.

 We will now prove equation (\ref{COM5}). It clearly holds in the $\ell=0$ case. We proceed by induction in $\ell$. Adding the two equations in (\ref{COM1}) and (\ref{COM2}) we find the $k=1$ case. Using the $k=1$ case, the induction hypothesis, equation (\ref{COM3}) and equation (\ref{COM4}) we find
\begin{align*}
A_\ell\varphi_+(g)^{k+1}&=\varphi_{\oplus,\ell}(g)^{k+1}A_\ell+\sum_{a=0}^{k} \varphi_{\oplus,\ell}(g)^{a} (A_\ell\varphi_+(g)-\varphi_{\oplus,\ell}(g)A_\ell) \varphi_+(g)^{k-a}\\=\varphi_{\oplus,\ell}(g)^{k+1}A_\ell&+\sum_{a=0}^{k} \varphi_{\oplus,\ell}(g)^{a} z^\dagger_{\ell-1}(g)A_{\ell-1}  \varphi_+(g)^{k-a}\\=\varphi_{\oplus,\ell}(g)^{k+1}A_\ell&+\sum_{a=0}^{k} \sum_{q=0}^{\min\{\ell-1,k-a  \}} \begin{pmatrix}
k-a \\ q
\end{pmatrix} \varphi_{\oplus,\ell}(g)^{k-q} \left(\prod_{c=-1}^{q-1} z^\dagger_{\ell-c-2}(g) \right)A_{\ell-q-1}\\
=\varphi_{\oplus,\ell}(g)^{k+1}A_\ell&+ \sum_{q=1}^{\min\{\ell,k+1  \}}\sum_{a=0}^{k+1-q} \begin{pmatrix}
k-a \\ q-1
\end{pmatrix} \varphi_{\oplus,\ell}(g)^{k+1-q} \left(\prod_{c=0}^{q-1} z^\dagger_{\ell-c-1}(g) \right) A_{\ell-q}\\=\varphi_{\oplus,\ell}(g)^{k+1}A_\ell&+ \sum_{q=1}^{\min\{\ell,k+1  \}} \begin{pmatrix}
k+1 \\ q
\end{pmatrix} \left (\prod_{c=0}^{q-1} z^\dagger_{\ell-c-1}(g) \right)\varphi_{\oplus,\ell-q}(g)^{k+1-q} A_{\ell-q}.
\end{align*}
We now prove part (2). Let $\psi\in \cD(f(d\Gamma_+(\omega)))$. Note that
\begin{equation*}
(f(d\Gamma_+(\omega))\psi)^{(n+\ell)}(k_1,\dots,k_{n+\ell})=f(\omega(k_1)+\cdots+\omega(k_{n+\ell}))\psi^{(n+\ell)}(k_1,\dots,k_{n+\ell} )
\end{equation*}
is in $\cH^{\otimes_s(n+\ell)}$ for all $\ell\in \NN$ and $n\in \NN_0$. Standard integration theory yields $\psi^{(n+\ell)}(k_1,\dots,k_\ell,\cdot,\dots,\cdot )\in \cD(f(d\Gamma^{(n)}(\omega)+\omega_\ell(k_1,\dots,k_\ell)))$ almost everywhere. Since $(\cP_nA_\ell\psi)(k_1,\dots,k_\ell)=\sqrt{(n+1)\cdots(n+\ell)}\psi^{(n+\ell)}(k_1,\dots,k_\ell,\cdot,\dots,\cdot )$ we find $A_\ell\psi(k)\in \cD(f(d\Gamma_{+}(\omega)+\omega_\ell(k)))$ for almost all $k\in \cM^{\ell}$. Furthermore,
\begin{align*}
&f(d\Gamma^{(n)}(\omega)+\omega_\ell(k_1,\dots,k_{\ell}))\psi^{(n+\ell)}(k_1,\dots,k_\ell,\cdot,\dots,\cdot )\\& \quad \quad\quad \quad\quad \quad \quad \quad\quad \quad\quad \quad =(f(d\Gamma_+(\omega))\psi)^{(n+\ell)}(k_1,\dots,k_{\ell},\cdot,\dots,\cdot)\\&
\int_{\cM^\ell}\lVert f(d\Gamma^{(n)}(\omega)+\omega_\ell(k_1,\dots,k_{\ell}))\psi^{(n+\ell)}(k_1,\dots,k_\ell,\cdot,\dots,\cdot )\lVert^2d\mu^{\otimes \ell}(k_1,\dots,k_\ell)\\&\quad \quad\quad \quad\quad \quad \quad \quad\quad \quad\quad \quad=\lVert P_{n+\ell}f(d\Gamma_+(\omega))\psi  \lVert^2<\infty.
\end{align*}
So $A_\ell\psi\in \cD(f(d\Gamma_{\oplus,\ell}(\omega)+\omega_\ell))$ by Lemma \ref{Lem:LiftToIntegralOperators} and
\begin{align*}
&(\cP_nf(d\Gamma_{\oplus,\ell}(\omega)+\omega_\ell)A_\ell\psi)(k_1,\dots,k_\ell)\\&=\sqrt{(n+1)\cdots(n+\ell)}f(d\Gamma^{(n)}(\omega)+\omega_\ell(k_1,\dots,k_\ell))\psi^{(n+\ell)}(k_1,\dots,k_\ell,\cdot,\dots,\cdot )\\& =\sqrt{(n+1)\cdots(n+\ell)}(f(d\Gamma_+(\omega))\psi)^{(n+\ell)}(k_1,\dots,k_{\ell},\cdot,
\dots,\cdot)\\&=(\cP_nA_\ell f(d\Gamma_{+}(\omega))\psi)(k_1,\dots,k_\ell).
\end{align*}
This finishes the proof.
\end{proof}
\noindent Commutation relations with Weyl operators can also be calculated but only on restricted domains. For future reference we prove
\begin{lem}
Let $\psi\in \cD(N^{\frac{1}{2}})$ and $g\in \cH$. Then the following holds
\begin{equation}\label{eqn:CommuteWeyl}
A_1W(g,1)\psi=\int_{\cM}^{\oplus}W(g,1)d\mu(k)A_1\psi+z_0(g)W(g,1)\psi
\end{equation}
\end{lem}
\begin{proof}
We calculate on an exponential vector $\epsilon(v)$:
\begin{align*}
(A_1W(g,1)\epsilon(v))(k)&=e^{-\lVert g\lVert^2/2-\text{Im}(\langle g,v \rangle)}A_1(\epsilon(v+g))(k)\\&=(v(k)+g(k))W(g,1)\epsilon(v)\\&=\left(\int_{\cM}^{\oplus}W(g,1)d\mu(k)A_1\epsilon(v)\right)(k)+z_0(g)W(g,1)\epsilon(v).
\end{align*}
Hence the result holds on the span of exponential vectors. The collection of exponential vectors span a core for the number operator $N$ and thus for $N^{\frac{1}{2}}$. Therefore, a general element in $\psi\in \cD(N^{\frac{1}{2}})$ may be approximated in $N^{\frac{1}{2}}$-norm by a sequence $\{\psi_n\}_{n=1}^\infty$ inside the span of exponential vectors. Lemmas \ref{Lem:NumberEstimates} and \ref{Lem:ContOfMultOperator} now imply
\begin{align*}
\lim_{n\rightarrow \infty}\int_{\cM}^{\oplus}&W(g,1)d\mu(k)A_1\psi_n+z_0(g)W(g,1)\psi_n\\&=\int_{\cM}^{\oplus}W(g,1)d\mu(k)A_1\psi+z_0(g)W(g,1)\psi
\end{align*}
in $L^2(\cM,\cE,\mu,\cF_b(\cH))$ and therefore in $\cC(\cM,\cE,\mu)$ as well by Lemma \ref{Lem:TheFundamentalMeasureSpace}. Lemma $\ref{Lem:ContOfMultOperator}$ implies
\begin{align*}
\lim_{n\rightarrow \infty}A_\ell W(g,1)\psi_n=A_\ell W(g,1)\psi_n
\end{align*}
in $\cC(\cM,\cE,\mu)$ which finishes the proof.
\end{proof}
\noindent The pointwise annihilation operators are useful for calculating expectation values. We shall need that $L^2(\cM,\cF,\mu,\cF_{b}(\cH))$ is a tensor product $\cH\otimes \cF_{b}(\cH)$ under the identification $f\otimes \phi=k\mapsto f(k)\phi$. If $\omega$ is a multiplication operator on $\cH$ then
\begin{align*}\smash
\omega\otimes 1&=\int^{\oplus}_{\cM} \omega(k)d\mu(k)\\
\cD(\omega\otimes 1)&=\{f\in L^2(\cM,\cF,\mu,\cF_{b}(\cH))\mid k\mapsto \omega(k) f(k)\in L^2(\cM,\cF,\mu,\cF_{b}(\cH))\}.
\end{align*}
In particular, $\cD(\omega\otimes 1)=\cD(\lvert \omega\lvert \otimes 1)$. We now prove:
\begin{thm}\label{Thm:CalculatingSecondQuantisedUsingAnihilation}
Let $\psi,\phi \in \cF_b(\cH)$ and $B $ be a selfadjoint operator on $\cH$. Define $B_{+}=B1_{[0,\infty)}(B)$ and $B_{-}=B1_{(-\infty,0)}(B)$. 
\begin{enumerate}
	\item[\textup{(1)}] We have
\begin{align}
\label{eq:Dom1}\cD(d\Gamma(B_+)^{\frac{1}{2}})\cap \cD(d\Gamma(B_-)^{\frac{1}{2}})&=\cD(d\Gamma(\lvert B\lvert)^{\frac{1}{2}})\\\cD(d\Gamma(\lvert B\lvert))&\subset \cD(d\Gamma(B_+)), \cD(d\Gamma(B_-)),\cD(d\Gamma(B))\label{eq:Dom2}
\end{align}
and $d\Gamma(B_+)-d\Gamma(B_-)=d\Gamma(B)$ on $\cD(d\Gamma(\lvert B\lvert))$.

\item[\textup{(2)}] Assume $B$ is a multiplication operator. Then $\psi\in \cD(d\Gamma(\lvert B\lvert )^{\frac{1}{2}})\iff  \lvert B\lvert^{\frac{1}{2}} A_1\psi\in L^2(\cM,\cE,\mu,\cF_b(\cH))$. Furthermore, for $\phi,\psi\in \cD(d\Gamma(\lvert B\lvert )^{\frac{1}{2}})$ we have
\begin{equation}\label{formula2}
\sum_{\sigma \in \pm}\sigma \langle  d\Gamma(B_\sigma )^{\frac{1}{2}}\phi,d\Gamma(B_\sigma )^{\frac{1}{2}}\psi\rangle=\int_{\cM}B(k) \langle A_1\phi(k),A_1\psi(k) \rangle  d\mu(k),
\end{equation}
and $A_1\psi(k)\in \cF_b(\cH)$ almost everywhere on $\{ \lvert B(k)\lvert>0 \}$.

\item[\textup{(3)}] For $\psi,\phi \in \cD(d\Gamma(\lvert B \lvert )^{\frac{1}{2}})\cap \cD(N^{\frac{1}{2}})$ we have $A_1\psi,A_1\phi\in \cD( \lvert B\lvert^{\frac{1}{2}} \otimes 1  )$ and
\begin{equation}
\langle  d\Gamma(\lvert B\lvert )^{\frac{1}{2}}\phi,d\Gamma(\lvert B\lvert )^{\frac{1}{2}}\psi\rangle=\langle  (\lvert B\lvert^{\frac{1}{2}} \otimes 1)A_1\phi, (\lvert B\lvert^{\frac{1}{2}} \otimes 1)A_1\psi \rangle .
\end{equation}

\item[\textup{(4)}] For $\psi \in \cD(d\Gamma(\lvert B \lvert ))\cap \cD(N^{\frac{1}{2}})$ and $\phi\in \cD(N^{\frac{1}{2}})$ we have $A_1\psi \in \cD( \lvert B\lvert \otimes 1  )= \cD(  B \otimes 1  )$ and
\begin{equation}
\langle  \phi,d\Gamma( B )\psi\rangle=\langle  A_1\phi, ( B \otimes 1) A_1\psi \rangle. 
\end{equation}

\item[\textup{(5)}] Let $v\in \cH$ and $\psi\in \cF_b(\cH)$. If $x\mapsto \overline{v}(k)(A_1\psi)(k)$ is Fock space valued and integrable in the weak sense then $\psi \in \cD(a(v))$ and
\begin{equation}
a(v)\psi=\int_{\cM}\overline{v}(k)(A_1\psi)(k)d\mu(k).
\end{equation}
\end{enumerate}
\end{thm}
\begin{proof}
We start by proving the first four statements of the theorem when $B$ is a multiplication operator. Let $A\in \{ B,B_+,B_- \}$ and note $A\leq \lvert B\lvert $ and $\lvert B\lvert=B_+ + B_-$. We prove equations (\ref{eq:Dom1}) and (\ref{eq:Dom2}) as follows
\begin{align*}
\cD(&d\Gamma(B_+)^{\frac{1}{2}})\cap \cD(d\Gamma(B_-)^{\frac{1}{2}})\\&=\left \{ (\psi^{(n)})\in \cF_b(\cH) \biggl \lvert \sum_{n=1}^{\infty} \int_{\cM^n} (B_\pm(k_1)+\cdots+B_\pm(k_n))\lvert \psi^{(n)} \lvert^2d\mu^{\otimes n}<\infty \right \}\\& =\left \{ (\psi^{(n)})\in \cF_b(\cH)\biggl \lvert \sum_{n=1}^{\infty} \int_{\cM^n} (\lvert B (k_1)\lvert+\cdots+\lvert B(k_n)\lvert)\lvert \psi^{(n)} \lvert^2d\mu^{\otimes n}<\infty \right \}\\&=\cD(d\Gamma(\lvert B\lvert )^{\frac{1}{2}})\\
\cD(&d\Gamma(\lvert B\lvert ))\\&=\left \{ (\psi^{(n)})\in \cF_b(\cH) \biggl \lvert \sum_{n=1}^{\infty} \int_{\cM^n} (\lvert B (k_1)\lvert+\cdots+\lvert B(k_n)\lvert)^2\lvert \psi^{(n)} \lvert^2d\mu^{\otimes n}<\infty \right \}\\&\subset \left \{ (\psi^{(n)})\in \cF_b(\cH) \biggl \lvert \sum_{n=1}^{\infty} \int_{\cM^n} (A (k_1)+\cdots+ A(k_n))^2\lvert \psi^{(n)} \lvert^2d\mu^{\otimes n}<\infty \right \}\\&=\cD(d\Gamma(A)).
\end{align*}
The identity $d\Gamma(B_+)-d\Gamma(B_-)=d\Gamma(B)$ on $\cD(d\Gamma(\lvert B\lvert))$ is now a simple computation. We now prove statement (2). Let $\psi\in \cF_b(\cH)$ and note that
\begin{align}
\sum_{n=1}^{\infty}& \int_{\cM^n} (\lvert B (k_1)\lvert+\cdots+\lvert B(k_n)\lvert)\lvert \psi^{(n)}(k_1,\dots,k_n) \lvert^2d\mu^{\otimes n}(k_1,\dots,k_{n})\nonumber 
\\\nonumber&=\int_{\cM}\lvert B(k_1)\lvert  \sum_{n=1}^{\infty} n\int_{\cM^{n-1}} \lvert \psi^{(n)}(k_1,\dots,k_n) \lvert^2d\mu^{\otimes n-1}(k_2,\dots,k_{n})d\mu(k_1)\\&=\int_{\cM}\lvert B(k)\lvert \lVert A_1\psi(k)\lVert^2 d\mu(k).\nonumber
\end{align}
This shows statement (2) except equation (\ref{formula2}). We have however proven equation (\ref{formula2}) in the case $\phi=\psi$ and $B\geq 0$. Using linearity and equation (\ref{eq:Dom1}), we find equation (\ref{formula2}) holds for $\phi=\psi$. One may now apply the polarisation identity to finish the proof of statement (2). Statement (3) follows trivially from statement (2) when $B$ is a multiplication operator.

We now prove statement (4). Let $\phi\in \cD(N^{1/2})$ and $\psi\in \cD(N^{1/2})\cap \cD( d\Gamma(\lvert B\lvert ) )$. First we note that
\begin{equation*}
 B(k_1)^2+\cdots+B(k_n)^2\leq  (\lvert B(k_1)\lvert +\cdots+\lvert B(k_n)\lvert )^2
\end{equation*}
so $\cD(d\Gamma(\lvert B\lvert))\subset \cD(d\Gamma(B^2)^{\frac{1}{2}})$. This implies $A_1\psi\in  \cD( \lvert B\lvert \otimes 1  )= \cD(  B \otimes 1  )$ by statement (3). If $\phi \in \cD(N^{\frac{1}{2}})\cap \cD(d\Gamma(\lvert B\lvert))$ the formula in statement (4) will follow from statements (1) and (2). To finish the proof, it is by Lemma \ref{Lem:NumberEstimates} enough to find a sequence $\{\phi_n \}_{n=1}^\infty\subset \cD(N^{\frac{1}{2}})\cap \cD(d\Gamma(\lvert B\lvert ))$ that converges to $\phi$ in the graph norm of $N^{\frac{1}{2}}$. 

Let $\phi\in \cD(N^{\frac{1}{2}})$ and let $\phi_n=1_{[-n,n]}(d\Gamma(\lvert B\lvert))\phi$. Since $d\Gamma(\lvert B\lvert)$ and $N^{\frac{1}{2}}$ commute strongly we find $\phi_n\in \cD(N^{\frac{1}{2}})\cap \cD(d\Gamma(\lvert B\lvert))$ and
\begin{equation*}
\lVert N^{\frac{1}{2}}(\phi_n-\phi) \lVert=\lVert (1-1_{[-n,n]}(d\Gamma(\lvert B\lvert)))N^{\frac{1}{2}}\phi \lVert
\end{equation*}
which converges to 0. This finishes the proof of statements (1)-(4) when $B$ is a multiplication operator.

For general $B$ we may pick an $L^2$ space $\cK$ and a unitary map $U:\cH\rightarrow \cK$ such that $UBU^*=\omega$ is a multiplication operator on $\cK$. Note that $\Gamma(U)$ transforms $d\Gamma(f(B))$ into $d\Gamma(f(\omega))$ for any realvalued and measurable $f$. This implies that statement (1) holds since it is true with $\Gamma(U)$ applied on both sides of each equation.

Let $\widetilde{N}$ be the number operator on $\cF_b(\cK)$ and $\widetilde{A}_1$ denote the pointwise annihilation operator on $\cF_+(\cK)$. First we prove that
\begin{equation}\label{eq:tr}
U^*\otimes \Gamma(U)^* \widetilde{A}_1\Gamma(U)=A_1
\end{equation}
as maps from $\cD(N^{\frac{1}{2}})$ to $L^2(\cM,\cE,\mu,\cF_b(\cH))$. Note $\Gamma(U)$ maps $\cD(N^{\frac{1}{2}})$ continuously into $\cD(\widetilde{N}^{\frac{1}{2}})$ with respect to the graph norms. Hence both sides of equation (\ref{eq:tr}) are continuous as maps from $\cD(N^{\frac{1}{2}})$ into $L^2(\cM,\cE,\mu,\cF_b(\cH))$ by Lemma \ref{Lem:NumberEstimates}. The set $\cA$ from equation (\ref{Defn:A}) spans a core for $N^{1/2}$ so we just need to see equation (\ref{Defn:A}) holds on $\cA$. Let $h^{\otimes n}\in \cA$ and calculate
\begin{align*}
U^*\otimes \Gamma(U)^* \widetilde{A}_1\Gamma(U) h^{\otimes n}(k)&=\sqrt{n}(U^*\otimes \Gamma(U)^*)  (Uh)(k)(Uh)^{\otimes n-1}\\&=\sqrt{n}h(k)h^{\otimes n-1}\\&=A_1h^{\otimes n}.
\end{align*}
We now prove statement (3). Under the assumptions in statement (3) we have $\Gamma(U)\psi,\Gamma(U)\phi\in \cD(  d\Gamma(\lvert \omega \lvert)^{\frac{1}{2}} )\cap\cD(  \widetilde{N}^{\frac{1}{2}} )$ so $A_1\psi,A_1\phi\in U^*\otimes \Gamma(U)^*\cD(\lvert \omega \lvert^{\frac{1}{2}}\otimes 1)$ $=\cD(\lvert B \lvert^{\frac{1}{2}}\otimes 1)$. We may then calculate
\begin{align*}
\langle  d\Gamma(\lvert B\lvert )^{\frac{1}{2}}\phi,d\Gamma(\lvert B\lvert )^{\frac{1}{2}}\psi\rangle&=\langle  d\Gamma(\lvert \omega\lvert )^{\frac{1}{2}}\Gamma(U)\phi,d\Gamma(\lvert \omega\lvert )^{\frac{1}{2}}\Gamma(U)\psi\rangle\\&=\langle  (\lvert \omega\lvert^{\frac{1}{2}} \otimes 1)  \widetilde{A}_1\Gamma(U)\phi,  (\lvert \omega\lvert^{\frac{1}{2}} \otimes 1) \widetilde{A}_1\Gamma(U)\psi \rangle  \\&=\langle  (\lvert B\lvert^{\frac{1}{2}} \otimes 1)A_1\phi, (\lvert B\lvert^{\frac{1}{2}} \otimes 1)A_1\psi \rangle .
\end{align*}
We now prove statement (4). Under the assumptions in statement (4) we have $\Gamma(U)\psi\in \cD(  d\Gamma(\lvert \omega \lvert) )\cap\cD(  \widetilde{N}^{\frac{1}{2}} )$ and so $A_1\psi \in U^*\otimes \Gamma(U)^*\cD(\lvert \omega \lvert \otimes 1)=\cD(\lvert B \lvert\otimes 1)$. Hence
\begin{align*}
\langle  \phi,d\Gamma(B )\psi\rangle&=\langle \Gamma(U)\phi,d\Gamma( \omega)\Gamma(U)\psi\rangle\\&=\langle   \widetilde{A}_1\Gamma(U)\phi, ( \omega \otimes 1)  \widetilde{A}_1\Gamma(U)\psi \rangle  \\&=\langle  A_1\phi, (B\otimes 1)A_1\psi \rangle .
\end{align*}
We now prove statement (5). Let $\phi\in \cH^{\otimes_s n}$ and note that
\begin{align*}
&\biggl\langle \phi,P_n \int_{\cM} \overline{v}(k)  (A_1\psi)(k) d\mu(k)   \biggl\rangle  \\&=\sqrt{n+1} \int_{\cM} \int_{\cM^n}\overline{v(k) \phi(k_1,\dots,k_n)}\psi^{(n+1)}(k,k_1,\dots,k_n)d\mu^{\otimes n}(k_1,\dots,k_n)d\mu(k).
\end{align*}
Using Fubinis Theorem we see
\begin{align*}
P_n \int_{\cM} \overline{v}(k)  (A_1\psi)(k) d\mu(k) =a_{n}(v)\psi^{(n+1)}.
\end{align*}
 Hence $(a_{n}(v)\psi^{(n+1)})\in \cF_b(\cH)$ so $\psi \in \cD(a(v))$ and the desired equality holds.
\end{proof}
\noindent We can now prove the pull-trough formula.
\begin{thm}\label{Thm:ThepulllThrough}
Let $\alpha\in \RR^{2n}$, $\eta \in \RR$, $f\in \cH^{2n}$ and $\omega$ be a selfadjoint multiplication operator on $\cH$. Assume $(\alpha,f,\omega)$ satisfies Hypothesis \ref{Hyp1}, \ref{Hyp2}, \ref{Hyp3} and \ref{Hyp4}. Define $\cE_{\pm 1}=\cE_{\pm  \eta}(\alpha,f,\omega)$, $F_{\pm 1} :=F_{\pm \eta}(\alpha,f,\omega)$ and $\omega_\ell(k_1,\dots,k_\ell)=\omega(k_1)+\cdots+\omega(k_\ell)$. Let $\lambda\leq \cE_{(-1)^\ell}$ for all $\ell\in \NN_0$ and define $R_\ell(a)=(F_{(-1)^\ell}-\lambda+a)^{-1}$ for $a> 0$ and $\ell\in \NN_0$. 

If $\psi \in \cD(F_{-1})=\cD(F_1)$ and $A_q(F_{1}-\lambda)\psi$ is Fock space valued for all $q\leq \ell$ then $(A_\ell \psi)(k)\in \cD(F_{-1})=\cD(F_1)$ for almost every $k\in \cM^\ell$ and
\begin{align}\label{PullthroughFormula}
A_\ell\psi=&-R_\ell(\omega_\ell(\cdot ))\sum_{i=1}^{2n} \alpha_i\sum_{q=1}^{\min\{ i,\ell \}} \begin{pmatrix}
i \\ q
\end{pmatrix}  \biggl(\prod_{c=0}^{q-1} S_{\ell-c} z_{\ell-c-1}(f_i)\biggl)  \varphi_{\oplus,\ell-q}(f_i)^{i-q}A_{\ell-q}\psi\nonumber\\&+ R_\ell(\omega_\ell(\cdot))A_\ell(F_0-\lambda)\psi.
\end{align}
If we further assume that Hypothesis 5 holds, $\eta \leq 0$ and $\psi$ is a ground state for $F_{ 1 }$ then $A_\ell\psi\in L^2(\cM^\ell,\cE^{\otimes \ell},\mu^{\otimes \ell},\cF_b(\cH))$ for all $\ell\in \NN$.
\end{thm}
\begin{proof}
$F_{(-1)^\ell}-\lambda\geq 0$ for all $\ell\in \NN$ by definition, so $R_\ell(\omega_\ell(k))$ exists almost everywhere since $\{ \omega\leq 0 \}$ is a $\mu$ zero set. Define the operators
\begin{align*}
F_{+,\ell}&=(-1)^\ell\eta\Gamma_+(-1)+d\Gamma_+(\omega)+\sum_{i=1}^{2n} \alpha_{i}\varphi_+(f_i)^i\\
F_{\oplus,\ell}&=(-1)^\ell\eta\Gamma_{\oplus,\ell}(-1)+d\Gamma_{\oplus,\ell}(\omega)+\omega_\ell+\sum_{i=1}^{2n} \alpha_{i}\varphi_{\oplus,\ell}(f_i)^i
\end{align*}
with domains $\cD(F_+)=\cD(d\Gamma_+(\omega))$ and $\cD(F_\oplus)=\cD(d\Gamma_{\oplus,\ell}(\omega)+\omega_\ell)$. Let $\psi \in \cD(F_{-1})=\cD(F_1)$ and assume $A_q(F_1-\lambda)\psi$ is Fock space valued for all $q\leq \ell$. By Lemma \ref{Lem:CommutatingPointwiseAnihilation} we have $A_\ell\psi \in \cD(F_\oplus)$ and using Lemmas \ref{Lem:AnihilCreaFieldOnF_+}, \ref{Lem:2ndQuantisedF_+} and  \ref{Lem:CommutatingPointwiseAnihilation} we obtain
\begin{align*}\label{eq:defnOfgl}
\nonumber g_\ell:=&-\sum_{i=1}^{2n} \alpha_i\sum_{q=1}^{\min\{ i,\ell \}} \begin{pmatrix}
i \\ q
\end{pmatrix}  \biggl(\prod_{c=0}^{q-1} S_{\ell-c} z_{\ell-c-1}(f_i)\biggl)  \varphi_{\oplus,\ell-q}(f_i)^{i-q}A_{\ell-q}\psi \nonumber \\&+A_\ell(F_0-\lambda)\psi  \\=&(F_{\oplus,\ell }-\lambda)A_\ell\psi.
\end{align*}
Assume that $g_\ell$ is almost everywhere Fock space valued. Let $M$ be a zeroset such that:
\begin{itemize}
\item $A_\ell\psi(k)\in \cF_{-\ell/2,+}(\cH)$ for all $k\in M^c$ (see Lemma \ref{Lem:NumberEstimates}).
\item $g_\ell(k)=(F_{+,\ell} +\omega_\ell(k)+\lambda)(A_\ell\psi)(k)$ and $g_\ell(k)\in \cF_b(\cH)$ for all $k\in M^c$. 
\item $R_\ell(\omega_\ell(k))$ exists for all $k\in M^c$.
\end{itemize}
Let $k\in M^c$ and $\cK=(F_{(-1)^\ell}-\lambda+\omega_\ell(k))\cN\cap \cD(d\Gamma(\omega))$. Then $\cK$ is dense by  Proposition \ref{Lem:BasicPropertiesSBmodel}, $\cK\subset \cN$ and $R_\ell(\omega_\ell(k))\cK\subset \cN$. Let $\phi\in \cK$. Using Lemma \ref{Lem:FormalAdjoints} we find
\begin{align*}
\langle \phi, A_\ell\psi (k) \rangle_+&=\langle (F_{(-1)^\ell}+\omega_\ell(k)-\lambda) R_\ell(\omega_\ell(k))\phi, A_\ell\psi (k) \rangle_+\\&=\langle R_\ell(\omega_\ell(k))\phi, g_\ell (k) \rangle=\langle \phi, R_\ell(\omega_\ell(k))g_\ell (k) \rangle_+.
\end{align*}
Corollary \ref{Cor:seperatingDense} now shows that $A_\ell\psi(k)= R_\ell(\omega_\ell(k))g_\ell(k)$. We conclude that equation  (\ref{PullthroughFormula}) is true pointwise on $M^c$ finishing the proof. We now prove $g_\ell$ is almost everywhere Fock space valued by induction in $\ell$.

If $\ell=1$ then $g_\ell$ is a linear combination of $A_1(F_{1}-\lambda)\psi$ and functions of the form $k\mapsto f_i(k) \varphi(f_i)^{i-1}\psi$ which all takes values in Fock space. Hence $g_1$ is almost everywhere Fock space valued and so equation (\ref{PullthroughFormula}) will hold for $A_1$. Assume now that $g_1,\dots,g_{\ell-1}$ are almost everywhere Fock space valued. Then equation (\ref{PullthroughFormula}) holds for $A_1\psi,\dots,A_{\ell-1}\psi$ and so $A_i\psi$ is almost everywhere $\cD(F_{-1})=\cD(F_{1})$-valued for $i\in \{ 1,\dots,\ell-1 \}$. Using Proposition \ref{Lem:BasicPropertiesSBmodel} and Lemma \ref{Lem:AnihilCreaFieldOnF_+} we find for all $q\in \{ 0,\dots, i \}$ that
\begin{equation*}
\varphi_{\oplus,\ell-q}(f_i)^{i-q} A_{\ell-q}\psi=k\mapsto \varphi_+(f_i)^{i-q}(A_{\ell-q}\psi)(k)=k\mapsto \varphi(f_i)^{i-q}(A_{\ell-q}\psi)(k).
\end{equation*}
In particular, $\varphi_{\oplus}(f_i)^{i-q} A_{\ell-q}\psi$ is almost everywhere Fock space valued for $q\in \{ 0,\dots, i \}$. Since $z_q (f_i)$ and $S_q $ map Fock space valued maps into Fock space valued maps, we see that $g_{\ell}$ is almost everywhere Fock space valued. This finishes the proof of equation (\ref{PullthroughFormula}).

For the second part we note that $(F_{1}-\cE_{1})\psi=0$ and $\cE_{1}\leq \cE_{(-1)^\ell}$ for all $\ell\in \NN_0$ by Theorem \ref{HVZ}. Hence we may apply equation (\ref{PullthroughFormula}) with $\lambda=\cE_{1}$. Using that $A_\ell\psi$ is $\cD(F_{-1})=\cD(F_1)$ valued almost everywhere we see
\begin{equation*}
k\mapsto \varphi(f_i)^q(A_\ell\psi)(k)=\varphi_{\oplus,\ell}(f_i)^{q} (A_{\ell}\psi)(k)
\end{equation*}
will be $\cF_b(\cH)$-valued and measurable for all $\ell\in \NN$, $i\in \{1, \dots, 2n \}$ and $q\in \{ 1,\dots, i \}$. We will prove that $\varphi_{\oplus,\ell}(f_i)^{q} A_{\ell}\psi$ is square integrable for all $\ell\in \NN$, $i\in \{1, \dots, 2n \}$ and $q\in \{ 1,\dots, i \}$. First we note that there is a constant $C_{q,i,\ell}$ such that
\begin{align*}
\lVert \varphi(f_i)^qR_\ell(\omega_\ell(k))\lVert^2\leq C_{q,i,\ell}\left(1+ \frac{1}{\omega_\ell(k)} \right)^2.
\end{align*}
Hence it is enough to prove that $\omega_\ell^{-2}\lVert g_\ell\lVert^2$ and $\lVert g_\ell\lVert^2$ are integrable for all $\ell\in \NN$ which will now be done via induction in $\ell$. If $\ell=1$ then $g_\ell$ is a linear combination of elements of the form $k\mapsto f_c(k) \varphi(f_c)^{c-1}\psi$  and since $f_c\in \cD(\omega^{-1})$ the claim follows.

	Inductively we now assume that $\omega_u^{-2}\lVert g_u\lVert^2$ and $\lVert g_u\lVert^2$ are integrable for all $u\in \{1,\dots, \ell-1  \}$. Then $k\mapsto \varphi(f_i)^q(A_u\psi)(k)$ is square integrable for all $i\in \{1, \dots ,2n \}$, $q\in \{ 1,\dots, i \}$ and $u\in \{1,\dots, \ell-1  \}$. Now $g_\ell$ is a linear combination of functions of the form
	\begin{align*}
	(k_1,\dots,k_\ell)\mapsto f_c(k_{\sigma(1)})\cdots f_c(k_{\sigma(b)})  \varphi(f_c)^{c-b} (A_{\ell-b}\psi)(k_{\sigma(b+1)},\dots,k_{\sigma(\ell)})
	\end{align*}
	where $b\in \{1,\dots, \ell  \}$, $c\in \{ 1,\dots,2n \}$ and $\sigma\in \cS_\ell$. Combining the observations that $\frac{1}{\omega_\ell(k)}\leq \frac{1}{\omega(k_{\sigma(1)})}$, $f_c\in \cD(\omega^{-1})$ and  $(\varphi(f_c)^{c-b} A_{\ell-b}\psi)(k_{\sigma(b+1)},\dots,k_{\sigma(\ell)})$ is square integrable with respect to $(k_{\sigma(b+1)},\dots,k_{\sigma(\ell)})$ we find the desired result.
\end{proof}

\begin{proof}[Proof of Theorem \ref{Thm:Numberstructure_Massless} part (1)]
Lemma  \ref{Lem:TrasformationCanTrans} and Theorem $\ref{unique}$ shows it is enough to prove the claim for the fiber operator. Lemmas \ref{Lem:SeconduantisedBetweenSPaces} and \ref{Lem:CannonicalSpace} show we may assume $\cH=L^2(\cM,\cF,\mu)$ with $(\cM,\cF,\mu)$ a $\sigma$-finite measure space. This case is dealt with in Lemma \ref{Lem:NumberEstimates} and  Theorem \ref{Thm:ThepulllThrough}.
\end{proof}

\section{Q-spaces and functional analysis}
\noindent Following the approach in \cite{Hasler} we have
\begin{lem}\label{Lem:ExsistenceOfNiceRealSubspace}
Let $\{ f_\alpha \}_{\alpha\in \cI}\subset \cH$ and $\omega$ be a selfadjoint and nonnegative operator on $\cH$. Write $ \cM_{b}(\sigma(\omega),\RR)$ for the set of maps from $\sigma(\omega)$ to $\RR$ which are bounded and measurable. Assume that $\langle f_\alpha,g(\omega)f_\beta \rangle\in \RR $ for all $\alpha,\beta\in \cI$ and $ g\in \cM_{b}(\sigma(\omega),\RR)$. Then there is a real Hilbert space $\cH_{\RR}$ such that $\cH=\cH_{\RR}+ i\cH_{\RR}$, $e^{-t\omega}$ maps $\cH_\RR$ to $\cH_\RR$ for all $t\geq 0$ and $f_\alpha\in \cH_{\RR}$ for all $\alpha\in \cI$.
\end{lem}
\begin{proof}
Let
\begin{equation*}
\cH'=\overline{\textup{Span}_{\RR}\{g(\omega)f_\alpha\mid g\in \cM_{b}(\sigma(\omega),\RR),\alpha\in \cI  \}}.
\end{equation*}
Note that $\cH'$ is a real Hilbert space. For every $f\in (\cH')^{\perp}\backslash \{0\}$ we define
\begin{equation*}
\cH(f)=\overline{\textup{Span}_{\RR}\{g(\omega)f\mid g\in\cM_{b}(\sigma(\omega),\RR) \} }.
\end{equation*}
It is clear that $e^{-t\omega}$ maps $\cH'$ to $\cH'$ and $\cH(f)$ to $\cH(f)$. Define
\begin{equation*}
\cA=\{  A \subset  (\cH')^{\perp}\backslash \{0 \} \mid \cH(f)\perp \cH(g) \,\, \forall \,\, f, g\in A  \,\,\, \text{with} \,\,\, f\neq g \}.
\end{equation*}
We partially order $\cA$ by inclusion and take a maximal totally ordered subset $\cB$. Let $B$ be the union of all elements in $\cB$. If $f,g\in B$ and $f\neq g$ then there is an element in $\cB$ that contains both $f$ and $g$ (since $\cB$ is totally ordered). This implies $\cH(f)\perp \cH(g)$ and so $B\in \cA$. Define
\begin{equation*}
\cH_{\RR}:=\cH'\oplus \bigoplus_{a\in B} \cH(a)\subset \cH.
\end{equation*}
$\cH_{\RR}$ is clearly a real Hilbert space containing $\{ f_\alpha \}_{\alpha\in \cI}$ and it is left invariant by $e^{-t\omega}$ since each component is. Assume towards contradiction that there is an element $f\in \cH_\RR^\perp \backslash \{ 0 \}$. Then for every $g_1,g_2\in \cM_{b}(\sigma(\omega),\RR)$ and $ h \in B$ we would have
\begin{equation*}
\langle g_2(\omega)f,g_1(\omega)h \rangle=\langle f,g_2(\omega)g_1(\omega)h \rangle=0
\end{equation*}
and so $\cH(f)$ is orthogonal to $\cH(h)$ for all $h\in B$. In particular, $B\cup\{ f \}\in \cA$ so $\cB\cup \{ B\cup\{ f \} \}$ is larger than $\cB$ and totally ordered which is not possible. Hence $\cH_\RR^\perp \backslash \{ 0 \}=\emptyset $.

Let $\{e_n \}_{n=1}^N$ be an orthonormal basis for $\cH_{\RR}$ ($N\leq \infty$) which is also an orthonormal basis for $\cH$. Hence we may write any element $f\in \cH$ as
\begin{equation*}
f=\sum_{j=1}^{N}\textup{Re}(\langle e_j,f \rangle)e_j+i\sum_{j=1}^{N}\textup{Im}(\langle e_j,f \rangle)e_j.
\end{equation*}
This finishes the proof.
\end{proof}

\begin{thm}\label{Thm:PropertiesOfQspace}
Let $\cH_{\RR}\subset \cH$ be a real Hilbert space such that $\cH=\cH_\RR+ i\cH_\RR$. Then there exists a probability space $(X,\cX,\QQ)$ such that $\cF_b(\cH)$ is unitarily isomorphic to $L^2(X,\cX,\QQ)$ via a map $V$. Furthermore, the following properties hold:
\begin{enumerate}
    \item[\textup{(1)}] If $U$ is a bounded operator on $\cH$ such that $U\cH_{\RR}\subset \cH_{\RR}$ and $\lVert U\lVert\leq 1$ then $V \Gamma(U) V^*$ is positivity preserving.

	\item[\textup{(2)}] Assume $\omega$ is a selfadjoint, nonnegative and injective operator on $\cH$. If  $e^{-t\omega}$ maps $\cH_{\RR}$ into $\cH_{\RR}$ for all $t\geq 0$ then $V e^{-td\Gamma(\omega)} V^*$ is positivity improving. If $\inf(\sigma(\omega))>0$ then $V e^{-td\Gamma(\omega)} V^*$ is hypercontractive.

	\item[\textup{(3)}] If $v\in \cH_\RR$ then $V\varphi(v)V^*$ acts like multiplication by a normally distributed variable $\widetilde{\varphi}(v)$ with mean 0 and variance $\lVert v\lVert^2$. If $v_1,v_2\in \cH_\RR$ and $a,b\in \RR$ then 
	\begin{align*}
	a\widetilde{\varphi}(v_1)+b\widetilde{\varphi}(v_2)=\widetilde{\varphi}(av_1+bv_2)
	\end{align*}
	almost everywhere.

\item[\textup{(4)}] If $\{v_n \}_{n=1}^\infty\subset \cH_\RR$ converges to $v\in \cH_\RR$ then $\widetilde{\varphi}(v_n)^\ell$ converges to $\widetilde{\varphi}(v)^\ell$ in $L^q(X,\cX,\QQ)$ for all $\ell\in \NN$ and $q\geq 1$.

\item[\textup{(5)}] Let $\alpha\in \RR^{2n}$, $q>0$ and $r>0$. Define
\begin{equation*}
\cK=\{ f\in \cH^{2n}\mid  \text{$(\alpha,f)$ satisfies part (\ref{Hyp1:1}) of Hypothesis \ref{Hyp1}} \,\,\, \text{and} \,\,\, \lVert f_1 \lVert<r     \}
\end{equation*}
There is a constant $C:=C(\alpha,r,q)$ such that for all $f\in \cK$ we have
\begin{equation*}
  \lVert e^{ \widetilde{H}_I(\alpha,f)} \lVert_q \leq C ,
\end{equation*}
where $\widetilde{H}_I(\alpha,f)=\sum_{j=1}^{2n}\alpha_i \widetilde{\varphi}(f_i)
$.
\end{enumerate}

\end{thm}
\begin{proof}
Everything in parts (1)-(3) can be found in \cite{Hirokawa1} and \cite{RS2}. We now prove part (4). For any $N(0,\sigma^2)$ distributed variable $X$ we have
\begin{align*}
\lVert \lvert X\lvert^a   \lVert_q =\sigma^a E[\lvert X/\sigma\lvert^{aq}]^{1/q}.
\end{align*}
Since $X/\sigma$ is $N(0,1)$ distributed we find that $E[\lvert X/\sigma\lvert^{qa}]^{1/q}$ depends only on $q$ and $a$. Write $B(q,a)$ for this constant. Then we may calculate using Hölders inequality
\begin{align*}
\lVert \widetilde{\varphi}(v_n)^\ell-\widetilde{\varphi}(v)^\ell \lVert_q&\leq \sum_{j=0}^{\ell-1} \lVert \widetilde{\varphi}(v_n)^{\ell-j-1}\widetilde{\varphi}(v_n-v)\widetilde{\varphi}(v)^j \lVert_q\\&\leq \sum_{j=0}^{\ell-1} \lVert \widetilde{\varphi}(v_n)^{(\ell-j-1)}\lVert_{3q}   \lVert \widetilde{\varphi}(v_n-v)\lVert_{3q} \lVert   \widetilde{\varphi}(v)^j \lVert_{3q}\\& \leq \sum_{j=0}^{\ell-1} B(3q,\ell-j-1) B(3q,1)B(3q,j)  \lVert v_n \lVert^{\ell-j-1}\lVert v_v-v\lVert \lVert v\lVert^{j}  
\end{align*}
showing the desired result.

We now prove part (5). Let $f\in \cK(\alpha)$. Using Lemma \ref{Lem:FundamentalLowerbound} we find $\sum_{j=2}^{2n}\alpha_i \widetilde{\varphi}(f_i)$ is uniformly bounded below by a constant $C_1$ depending only on $\alpha$. Therefore we find
\begin{align*}
\lVert e^{ -\widetilde{H}_I(\alpha,f)} \lVert_q\leq e^{-C_1}  E[e^{-q\alpha_1\widetilde{\varphi}(f_1)}]^{1/q}=e^{-C_1} (e^{-q^2\alpha_1^2\lVert f_1\lVert^2/2})^{1/q}\leq e^{-C_1}e^{-r^2\alpha_1^2q/2}. 
\end{align*}
This finishes the proof.
\end{proof}
\begin{lem}\label{Lem:KonvOfBoundedPer}
Let $\{ A_n \}_{n=1}^\infty$ be a sequence of selfadjoint operators on a Hilbert space $\cH$ converging to $A$ in norm resolvent sense. If $B$ is a bounded and selfadjoint operator on $\cH$ then $\{ A_n+B \}_{n=1}^\infty$ will converge in norm resolvent sense to $A+B$.
\end{lem}
\begin{proof}
For $\lambda>\lVert B\lVert +1$ we have $\max\{\lVert B(A-i\lambda)^{-1}\lVert,\lVert B(A_n-i\lambda)^{-1}\lVert\}<\frac{\lVert B\lVert}{1+\lVert B\lVert}$ and so we may calculate
\begin{align*}
(A+B-i\lambda)^{-1}-&(A_n+B-i\lambda)^{-1}\\&=\sum_{k=0}^{\infty}(A-i\lambda)^{-1}(B(A-i\lambda)^{-1})^k-(A_n-i\lambda)^{-1}(B(A_n-i\lambda)^{-1})^k.
\end{align*}
Each term in the series converge to 0 as $n$ tends to $\infty$. Furthermore,
\begin{equation*}
\lVert (A-i\lambda)^{-1}(B(A-i\lambda)^{-1})^k-(A_n-i\lambda)^{-1}(B(A_n-i\lambda)^{-1})^k \lVert\leq \frac{2}{\lambda}\left( \frac{\lVert B\lVert}{1+\lVert B\lVert} \right)^k
\end{equation*}
which is summable. The conclusion now follows by dominated convergence.
\end{proof}

\subsection*{Acknowledgements}
Thomas Norman Dam was supported by the Independent Research Fund Denmark through the project "Mathematics of Dressed Particles".

\end{document}